\documentclass[authorcolumns,numberwithinsect,a4paper,footnotebackref]{no-lipics}

\usepackage{bm}
\usepackage{stmaryrd}
\usepackage{booktabs}
\usepackage{ltablex}

\usepackage{xstring}
\usepackage{xspace}
\usepackage{tikz}
\usepackage{amssymb}
\usepackage{mathtools}
\usepackage[framemethod=TikZ]{mdframed}

\usetikzlibrary{arrows,arrows.meta,shapes.misc, decorations.pathreplacing,decorations.pathmorphing,calligraphy, patterns.meta}
\tikzset{dotmark/.style={circle,fill,inner sep=1.5pt}}
\tikzset{emptymark/.style={circle,draw,fill=white,inner sep=1.5pt}}
\tikzset{crossmark/.style={thick,inner sep=1.5pt}}

\newcommand{\A}{\mathcal{A}}

\newcommand{\qcsp}{$q$-\textsc{CSP}-$B$\xspace}
\newcommand{\qcspa}{$q$-\textsc{CSP}-$\mname$\xspace}

\DeclareMathOperator{\poly}{poly}

\newcommand{\fragment}[2]{\bm{[}\,#1\,\bm{.\,.}\,#2\,\bm{]}}
\newcommand{\position}[1]{\bm{[}\,#1\,\bm{]}}
\newcommand{\vposition}[1]{\bm{[}\,#1\,\bm{]}}

\makeatletter
\renewenvironment{cases}{%
  \matrix@check\cases\env@cases
}{%
  \endarray\right.%
}
\def\env@cases{%
\let\@ifnextchar\new@ifnextchar
\left\lbrace
\def\arraystretch{1.1}%
\array{@{\;}c@{\quad}l@{}}%
}
\makeatother

\def\emptyset{\varnothing}

\newcommand{\Nat}{\mathbb Z_{\ge 0}} 
\newcommand{\numb}[1]{\fragment{1}{#1}} 
\newcommand{\numbZ}[1]{\fragment{0}{#1}} 

\newcommand{\EVEN}{\ensuremath{\{ x \in \Nat \mid x \equiv_2 0 \}}}
\newcommand{\ODD} {\ensuremath{\{ x \in \Nat \mid x \equiv_2 1 \}}}
\newcommand{\Field}{\ensuremath{\mathbb F}}

\newcommand{\NP}{\textnormal{\textsf{NP}}\xspace}
\newcommand{\W}[1][1]{\textnormal{\textsf{W}[\ensuremath{#1}]}\xspace}

\newcommand{\ttop}{\textup{top}}
\newcommand{\rhoMax}{r_{\ttop}}
\newcommand{\sigMax}{s_{\ttop}}
\newcommand{\allMax}{t_{\ttop}}

\newcommand{\rhoMin}{r_{\min}}
\newcommand{\sigMin}{s_{\min}}

\newcommand{\rhoStates}{{\mathbb R}}
\newcommand{\sigStates}{{\mathbb S}}
\newcommand{\allStates}{{\mathbb A}}

\newcommand{\port}{u} 
\newcommand{\Port}{U} 

\newcommand{\sigMod}{\mname_\sigma}
\newcommand{\rhoMod}{\mname_\rho}

\newcommand{\inverse}[2][]{{\operatorname{\mathsf{inv}}^{#1}({#2})}}
\newcommand{\inv}{\overline}

\newcommand{\mname}{{\mathrm{m}}}
\def\sigvec#1{\ensuremath \overrightarrow{\sigma}(#1)} 
\def\degvec#1{\ensuremath \overrightarrow{\mathrm w}(#1)} 
\def\comp#1#2{\ensuremath {{#2\!\!\downarrow}}}
\newcommand{\decomp}[2][o_1+o_2]{\ensuremath {#2\!\!\uparrow_{#1}}}

\newcommand{\witnessvec}[2][\ell]{w_{#2,#1}}

\newcommand{\remvec}[3][]{\ensuremath \mathrm{rem}_{#1}(#2, #3)} 

\newcommand{\manager}[1]{\ensuremath{#1}\text{-manager}}
\newcommand{\scope}{\operatorname{\mathsf{scp}}}

\newcommand{\acc}{\operatorname{\mathsf{acc}}} 

\newcommand{\rank}{rank\xspace}
\newcommand\Bl{B} 
\newcommand\Br{\inv B} 

\newcommand{\Relation}[3]{\ensuremath{\mathtt{#1}%
    \ifthenelse{\equal{#2}{}}{}{_{#2}}%
    \ifthenelse{\equal{#3}{}}{}{^{(#3)}}}\xspace}
\newcommand{\HWeq}[2][]{\Relation{HW}{=#2}{#1}}

\newcommand{\HWin}[2][]{\Relation{HW}{\in #2}{#1}}
\newcommand{\HWone}[1][]{\HWin[#1]{\rho-\min\rho+1}}

\newcommand{\HWsetGen}[1]{\ensuremath{\mathtt{HW}_{#1}}}
\newcommand{\HWset}[1]{\HWsetGen{= #1}}

\mathchardef\mhyph="2D
\newcommand{\DomSetGeneral}[4]{\ensuremath{(#1,#2)\mhyph}\textsc{{#3}Dom\-Set\ensuremath{^{#4}}}\xspace}
\newcommand{\DomSet}[2]{\DomSetGeneral{#1}{#2}{}{}}
\newcommand{\srDomSet}{\DomSetGeneral{\sigma}{\rho}{}{}}
\newcommand{\DomSetRel}[3][\Rel]{\DomSetGeneral{#2}{#3}{}{#1}}
\newcommand{\srDomSetRel}[1][\Rel]{\DomSetGeneral{\sigma}{\rho}{}{#1}}
\newcommand{\srDomSetShift}{\ensuremath{(\sigma,\rho)\mhyph}\textsc{Dom\-Set with Shift-Vectors}\xspace}

\newcommand{\AllOff}{\textsc{AllOff}\xspace}
\newcommand{\ReflAllOff}{\textsc{Reflexive-AllOff}\xspace}
\newcommand{\kSAT}[1][k]{\ensuremath{#1}-SAT\xspace}

\newcommand{\LightsOut}{Lights~Out\@\xspace}

\DeclarePairedDelimiter{\abs}{\lvert}{\rvert}

\let\pos\position

\newcommand{\tw}{\textup{\textsf{tw}}} 
\newcommand{\pw}{\textup{\textsf{pw}}} 

\newcommand{\deff}{\coloneqq}

\newcommand{\Oh}{\mathcal O} 

\newcommand*\from{\colon}
\newcommand{\compl}{\overline} 

\let\epsilon\varepsilon
\let\bar\overline

\usepackage{bm}
\def\nset#1{\numb{#1}}

\newcommand{\CZ}{\mathcal{Z}}
\newcommand{\CS}{\mathcal{S}}
\newcommand{\CC}{\mathcal{C}}
\newcommand{\CW}{\mathcal{W}}
\newcommand{\ZZ}{\mathbb{Z}}
\newcommand{\FF}{\mathbb{F}}

\newcommand{\cutset}{\textsf{\upshape{cut}}}

\newcommand{\reductionGraph}{G_I}

\newcommand{\variableCount}{n}
\newcommand{\constraintCount}{\ell}
\newcommand{\constraint}{C}
\newcommand{\variable}{x}
\newcommand{\colCount}{\constraintCount}
\newcommand{\rowCount}{\variableCount}
\newcommand{\rowSet}{\numb{\rowCount}}
\newcommand{\colSet}{\numb{\colCount}}

\newcommand{\cutStateSet}{\rhoStates_\cutset}

\newcommand{\consistencyRelation}{\mathtt{R}}
\newcommand{\constraintRelation}[1]{\mathtt{C}^{#1}}

\newcommand{\tupleindex}[1]{\lambda_{#1}}

\newcommand{\satAssignment}{\pi}

\title{Residue Domination in Bounded-Treewidth Graphs}

\author{Jakob Greilhuber}
{TU Wien, Vienna, Austria \and CISPA Helmholtz Center for Information Security, Saarbrücken, Germany}
{jakob.greilhuber@cispa.de}
{https://orcid.org/0009-0001-8796-6400}
{}
\author{Philipp Schepper}
{CISPA Helmholtz Center for Information Security, Saarbrücken, Germany}
{}
{https://orcid.org/0000-0002-5810-7949}
{}
\author{Philip Wellnitz}{National Institute of Informatics\\The Graduate University for Advanced Studies, SOKENDAI\\Tokyo, Japan}{wellnitz@nii.ac.jp}{https://orcid.org/0000-0002-6482-8478}{}

\authorrunning{J. Greilhuber, P. Schepper, and P. Wellnitz}


\relatedversion{%
The Master's Thesis of Jakob Greilhuber~\cite{master_thesis_greilhuber}
is based on the lower bound results of this work.
\\
An extended abstract of this work appeared at STACS 2025 \cite{greilhuberResidueDominationBoundedTreewidth2025}.}

\funding{\emph{Jakob Greilhuber and Philipp Schepper} are members of the
Saarbrücken Graduate School of Computer Science, Germany.}

\acknowledgements{The work of Jakob Greilhuber has been carried out mostly during a summer
internship at the Max Planck Institute for Informatics, Saarbr\"ucken, Germany. The work of Philip Wellnitz was
partially carried out at the Max Planck Institute for Informatics.}

\begin{document}

\pagenumbering{roman}

\maketitle
\begin{abstract}
    For the vertex selection problem \srDomSet
one is given two fixed sets $\sigma$ and $\rho$ of integers
and the task is to decide whether we can select vertices of the input graph
such that, for every selected vertex,
the number of selected neighbors is in $\sigma$
and, for every unselected vertex,
the number of selected neighbors is in $\rho$
[Telle, Nord.\ J.\ Comp.\ 1994].
This framework covers many fundamental graph problems such as
\textsc{Independent Set} and \textsc{Dominating Set}.

We significantly extend the recent result by Focke~et~al.~[SODA~2023]
to investigate the case when $\sigma$ and $\rho$
are two (potentially different) residue classes modulo $\mname\ge 2$.
We study the problem parameterized by treewidth
and present an algorithm that solves
in time $\mname^\tw \cdot n^{\Oh(1)}$
the decision, minimization and maximization version of the problem.
This significantly improves upon the known algorithms
where for the case $\mname \ge 3$ not even an explicit running time is known.
We complement our algorithm by providing matching lower bounds
which state that there is no $(\mname-\epsilon)^\pw \cdot n^{\Oh(1)}$-time
algorithm parameterized by pathwidth~$\pw$,
unless SETH fails.
For $\mname = 2$, we extend these bounds to the minimization version
as the decision version is efficiently solvable.

\end{abstract}
\clearpage
\thispagestyle{plain}
\tableofcontents
\clearpage
\pagenumbering{arabic}

\section{Introduction}

Classical graph problems such as
Dominating Set or Independent Set
are ubiquitous in computer science.
These problems are not only of theoretical interest
but also have many practical applications;
including
facility location, coding theory, modeling communication networks,
map labeling, or even similarity measures on molecules
\cite{doi:10.1089/cmb.2010.0280,%
  DBLP:conf/gis/BarthNNS16,%
  1083994,%
  gagarinMultipleDominationModels2018,%
  HalldorssonKT00-mod-2,%
  DBLP:journals/dm/HedetniemiL90a%
}.
Therefore, these problems are extensively studied on plenty of graph classes
and several generalizations and variations have been formulated and considered
\cite{%
  BorradaileL16,%
  DBLP:conf/fct/CattaneoP13,%
  DBLP:conf/soda/DubhashiMPRS03,%
  FominT06,%
  DBLP:journals/talg/GaspersKLT09,%
  DBLP:conf/soda/KhullerPS14,
  DBLP:conf/csr/MisraR19,%
  DBLP:conf/esa/PhilipRS09,%
  DBLP:conf/stoc/SolomonU23,%
  DBLP:conf/sigal/TsaiH90}.
Moreover, the problems seem to come with a significant complexity
but also sufficient structural properties
to serve as a testing point for new techniques
which frequently result in faster algorithms for the problems
\cite{DBLP:journals/talg/CyganNPPRW22,FominGK09,RooijBR09}.

In 1993, Telle and Proskurowski introduced the general class of \srDomSet problems
which capture several well-known vertex selection problems
for appropriately chosen sets $\sigma, \rho \subseteq \Nat$
\cite{Telle94,tellePracticalAlgorithmsPartial1993}.
In this problem the input is an undirected graph
and the task is to decide if we can select vertices
such that~(1), for every selected vertex,
the number of selected neighbors
is contained in the set~$\sigma$
and~(2), for every unselected vertex,
the number of selected neighbors is contained in the set~$\rho$.
Formally, for a graph~$G$, decide
if there exists a vertex set~$S \subseteq V(G)$
such that, for all $v \in S$, we have $\abs{N(v) \cap S} \in \sigma$,
and, for all $v \notin S$, we have $\abs{N(v) \cap S} \in \rho$.
Such a set $S$ is called a $(\sigma,\rho)$-set.

It is easy to see that \srDomSet captures classical Dominating Set
when we set~$\sigma = \Nat$ and~$\rho = \Nat \setminus\{0\}$
and ask for a selection of bounded size.
Moreover,
with different requirements imposed on the size of the selection,
we can also reformulate other problems
such as Independent Set ($\sigma = \{0\}$ and $\rho = \Nat$),
Perfect Code ($\sigma = \{0\}$ and $\rho= \{1\}$),
Induced $q$-regular Subgraph ($\sigma = \{q\}$ and $\rho = \Nat$),
Odd Domination ($\sigma = \{0,2,\dots\}$ and $\rho = \{1,3,\dots\}$),
and many more.
We refer to \cite{Bui-XuanTV13,Telle94} for a longer
list of problems that can be described as \srDomSet.

Since \srDomSet generalizes
many fundamental
graph problems,
the ultimate goal is to settle the complexity of \srDomSet
for \emph{all} (decidable) sets~$\sigma$ and~$\rho$.
We know that for many choices of~$\sigma$ and~$\rho$
the problem is \NP-hard.
Hence,
we frequently either restrict the input to special graph classes
or parameterize by some (structural) measure of the input
(for example, the solution size).

One of the best explored structural parameters is treewidth
\cite{%
  DBLP:conf/stoc/Bodlaender93,%
  CurticapeanM16,%
  DBLP:journals/talg/CyganNPPRW22,%
  FockeMINSSW23,%
  DBLP:conf/focs/Korhonen21,%
  lokshtanovKnownAlgorithmsGraphs2018,%
  DBLP:conf/birthday/Marx20,%
  DBLP:conf/soda/OkrasaR20,%
  RooijBR09%
}
which measures how similar a graph is to a tree
(see \cite[Chapter~7]{cyganParameterizedAlgorithms} for a more thorough introduction).
Many problems admit efficient algorithms on trees
with a simple dynamic program.
With treewidth as parameter,
we can lift these programs
to more general graphs
and obtain fast algorithms
especially compared to the running time
obtained from Courcelle's Theorem~\cite{Courcelle90}.
For most of these problems the goal is to find the smallest constant~$c$
such that the respective problem can be solved
in time~$c^{\tw} \cdot n^{\Oh(1)}$.

For problems parameterized by treewidth
a perpetual improvement of the algorithm seems unlikely.
After a few iterations,
frequently
conceptually new ideas seems to be necessary
to obtain further improvements.
Lokshtanov, Marx, and Saurabh initiated a line of research
that proves that such limitations
are often not a shortcoming of the techniques at hand,
but rather an inherent property of the problem itself
\cite{lokshtanovKnownAlgorithmsGraphs2018}.
For example, they prove that the known algorithm for Dominating Set~\cite{RooijBR09}
which takes time $3^\tw \cdot n^{\Oh(1)}$
cannot be improved further
unless the Strong Exponential-Time Hypothesis (SETH) \cite{cip09,ip01} fails.

Hence, the ultimate goal for \srDomSet is to show the following result
(or to prove that no such constant exists in the respective setting):

\begin{center}
  For all sets~$\sigma$ and~$\rho$,
  determine
  the constant~$c_{\sigma,\rho}$
  such that \srDomSet \\
  can be solved in time
  $c_{\sigma,\rho}^{\tw} \cdot n^{\Oh(1)}$ \\
  but not in time
  $(c_{\sigma,\rho}-\epsilon)^{\tw} \cdot n^{\Oh(1)}$
  for any~$\epsilon > 0$,
  unless SETH fails.
\end{center}

For certain choices of the sets $\sigma$ and $\rho$,
some (partial) results of this form are already known
\cite{DBLP:journals/tcs/MeybodiFMP20%
}.
A broad class of algorithms was given by
van Rooij, Bodlaender, and Rossmanith~\cite{RooijBR09}
for the case of finite and cofinite sets.
These algorithms were
later improved by van Rooij~\cite{Rooij_FastJoinOperations}.
Focke et al.~\cite{FockeMINSSW23,focke_tight_2023_ub}
recently introduced highly non-trivial techniques
to improve these algorithms further for an infinite class
of choices for~$\sigma$ and~$\rho$.
Moreover, Focke et al.\ additionally provide matching lower bounds
for these new algorithms
that rule out additional improvements~\cite{FockeMINSSW23,focke_tight_2023_lb}.

\subparagraph*{Beyond Finite and Cofinite Sets.}
Although the known results already capture large classes of problems,
they are limited to finite and cofinite sets.
This leaves open the entire range of infinite sets (with infinite complements),
which contains not only ``unstructured'' sets
like the set of prime numbers, for example,
but also easy to describe and frequently used
sets like the even or odd numbers,
and arithmetic progressions in general.

One important example for families
that are neither finite nor cofinite
are \emph{residue classes}.
We say that a set~$\tau \subseteq \Nat$ is a residue class modulo $\mname$
if there are two integers~$a$ and~$\mname$
such that $\tau = \{ n \in \Nat \mid n \equiv_\mname a \}$;
we usually require that $0 \le a < \mname$
for the canonical representation.
Again, the two most natural residue classes are the even and odd numbers.

Surprisingly, for such infinite sets the complexity of \srDomSet
is significantly underexplored and heavily fragmented
even in the classical, non-parameterized, setting.
This is especially surprising as this variant of the problem
has direct applications
in other fields like coding theory~\cite{DBLP:conf/tamc/CattaneoP14,HalldorssonKT00-mod-2}.

Halld\'orsson, Kratochv\'il, and Telle
consider a variation of Independent Set
where the unselected vertices have parity constraints~\cite{HalldorssonKT00}.
They provide a complete dichotomy between polynomial-time solvable cases and \NP-hard cases.
Later the same group of authors considered the case
where each of the sets comprises either the even or odd integers
and proved similar hardness results~\cite{HalldorssonKT00-mod-2}.
Caro, Klostermeyer, and Goldwasser consider
a variant of \srDomSet with residue classes as sets
where they restrict the \emph{closed} neighborhood of a vertex
\cite{caroOddResidueDomination2001}.
In this setting they prove new upper bounds
for specific graphs classes including complements of powers of cycles
and grid graphs.
Fomin, Golovach, Kratochv\'il, Kratsch, and Liedloff
consider general graphs and provide \emph{exponential-time} algorithms
for general residue classes~\cite{FominGKKL09}.

Although for the case of general sets
some more results are known
in the parameterized setting
\cite{AlberBFKN02,Bui-XuanTV13,%
  FominT06,golovachParameterizedComplexityGeneralized2012a},
surprisingly few involve
sets that are neither finite nor cofinite
and the parameter treewidth.
Gassner and Hatzl \cite{gassnerParityDominationProblem2008}
consider the problem of residue domination
where the sets are either all even or all odd numbers.
They provide an algorithm for the problem
with running time~$2^{3\tw} \cdot n^{\Oh(1)}$.
Gassner and Hatzl also conjecture that their algorithm works
when the sets have a larger modulus
but unfortunately do not state the expected running time
or the actual algorithm.

Chapelle in contrast provides a general algorithm for \srDomSet
which covers \emph{all} ultimately periodic sets%
\footnote{A set~$\tau$ is ultimately periodic
  if there is a finite automaton (over a unary alphabet)
  such that the length of the accepted words
  is precisely described by the set~$\tau$.}
but, unfortunately, does not provide an explicit running time
of the algorithm
\cite{Chapelle10v5,Chapelle11}.

In this work, we answer the main question from above
and settle the complexity of \srDomSet
for another large class of sets,
namely for the residue classes.
\begin{restatable*}{mtheorem}{thmTwUpperBound}
  \label{thm:twUpperBound}
  \label{thm:algorithm_runtime}
  Write $\sigma, \rho \subseteq \Nat$ for two residue classes
  modulo $\mname \ge 2$.

  Then, in time $\mname^\tw \cdot \abs{G}^{\Oh(1)}$ we can decide
  simultaneously for all $s$
  if the given graph $G$ has a $(\sigma,\rho)$-set of size~$s$
  when a tree decomposition of width $\tw$ is given with the input.
\end{restatable*}

We remark that our algorithm
does not only solve the decision version
but also the maximization, minimization and exact version
of \srDomSet.

Despite the fact that there are some pairs
for which the decision version of \srDomSet is efficiently solvable
(for example, the empty set is a trivial minimum solution if $0 \in \rho$),
we prove that for all other ``difficult'' cases our algorithm is optimal
even for the decision version
and cannot be improved unless SETH fails.
We refer to \cref{def:easyCases} for a complete list of the ``easy'' pairs
for which the decision version can be solved in polynomial time;
to all other pairs we refer as ``difficult''.

\begin{restatable*}{mtheorem}{thmPwLowerBound}
  \label{thm:pwLowerBound}
  Write $\sigma, \rho \subseteq \Nat$ for difficult residue classes
  modulo $\mname \ge 2$.

  Unless SETH fails,
  for all $\epsilon > 0$,
  there is no algorithm that can decide in time
  $(\mname - \epsilon)^\pw \cdot \abs{G}^{\Oh(1)}$
  whether the input graph $G$
  has a $(\sigma,\rho)$-set,
  when a path decomposition of width $\pw$ is given with the input.
\end{restatable*}

Observe that our lower bound is for the larger parameter pathwidth,
which immediately implies the result
for the smaller parameter treewidth.

\subparagraph*{Our Contribution.}
Before we outline the formal ideas behind our results,
we first highlight why these bounds are
more surprising than they seem to be at the first glance.

To that end, let us take a deeper look at the algorithms
of~\cite{focke_tight_2023_ub,Rooij_FastJoinOperations}.
Typically, the limiting factor for faster algorithms
parameterized by treewidth
is the number of states that have to be considered
for each bag of the tree decomposition.
For vertex selection problems, the state of a vertex
is defined by two values:
(1)~whether it is selected
and (2)~how many selected neighbors it gets in some (partial) solution.
To bound this latter number, we identify the largest ``reasonable'' state
a vertex can have when it is selected and when it is unselected.

For finite sets $\sigma$ and $\rho$
this largest reasonable state
is simply determined by the maximum of the respective sets,
that is,
we set $\sigMax=\max \sigma$ and $\rhoMax=\max \rho$
as the largest reasonable number of neighbors, respectively.

Then,
for a selected vertex, the allowed number of selected neighbors
ranges from~$0$ to~$\sigMax$, yielding $\sigMax + 1$
states for selected vertices.
Similarly, we need to consider $\rhoMax + 1$ states for unselected vertices.
Combining the two cases, for each bag of the tree decomposition
there are at most $(\sigMax + \rhoMax + 2)^{\tw+1}$ states to consider.
Surprisingly, Focke et\ al.\ proved
that, for an infinite number of finite sets,
even at most $(\allMax+1)^{\tw+1}$ of said states suffice,
where $\allMax = \max(\sigMax, \rhoMax)$~\cite{FockeMINSSW23}.%
\footnote{To keep notation simple,
  we omit the special case where the bound is $(\allMax+2)^{\tw+1}$.}
When $\sigMax = \rhoMax$ this improves the algorithm by a factor of~$2^{\tw+1}$.

Similarly,
for the case of residue classes with modulus~$\mname$,
the number of selected neighbors effectively ranges from~$0$ to~$\mname-1$
as for all larger values
the behavior is equivalent to some smaller value.
This gives us $\mname$ states if a vertex is unselected
and $\mname$ states if a vertex is selected.
Hence, the straight-forward bound for the number of states is
$(2 \mname)^{\tw+1}$
which is a factor of $2^{\tw+1}$ worse
than the bound from the running time of our algorithm.

We remark that the most naive approach,
for which we remove all integers from the sets~$\sigma$ and~$\rho$
that are larger than the number of vertices of the input graph
and then apply the improved result by Focke et al.~for finite sets,
fails miserably.
This approach would
merely give an \textsc{XP}-algorithm
as the size of the sets now depends on the size of the input graph.

Hence, it is far from trivial to obtain the claimed running time
since the classical approach would not give something better
than $(2\mname)^{\tw} \cdot n^{\Oh(1)}$.

\subparagraph*{Our Techniques.}
Although we use the algorithmic result
by Focke et al.\ as a basis
for our upper bounds,
our algorithm does not follow as an immediate corollary.
There are two main challenges that we need to overcome to obtain the fast running time.

First, all previous results with \emph{tight} bounds
considered finite or cofinite sets.
In this setting all integers (starting from some threshold)
are somewhat equivalent
in the sense that they are either all contained in the set
or all not contained in the set.
This makes defining a largest reasonable state quite convenient.
For the residue classes this is not as easily possible
as the integers change between membership and non-membership.
Hence, we need an even more careful construction and analysis
when improving upon the naive bound for the number of states.

Second, it does not suffice
to bound the number of states
which the dynamic program considers,
but we also need to be able to combine these states
efficiently at the join nodes.
Although this is a general issue when parameterizing by treewidth,
until today there does not seem to be one solution
which works in a black-box manner in all settings.
Instead, we need to carefully design a new approach
that takes care of the particular setting we consider
which is based on but differs from the existing results
for finite and cofinite sets.

For the lower bound result we are at a similar situation
as for the upper bound;
the setting is similar to what is known but still different.
There are several known lower bounds for
Dominating Set and its related problems,
but these reductions are usually quite tailored to the specific problem
and lack a modular construction%
---it is difficult to reuse existing results.
The proofs are usually based on a direct reduction from $k$-SAT
which introduces an additional overhead
that obfuscates the high-level idea by technicalities
(SAT has a running time bound of the form~$2^n$ but we want a bound of the form~$c^{\pw}$).

We avoid this overhead by reducing
from an appropriate Constraint Satisfaction Problem
introduced by Lampis
for precisely such settings~\cite{Lampis20}.
Our results can be seen as one of the first applications
of this new approach (outside the original setting)
that can potentially also serve as a blue-print
to simplify many other reductions and lower bound proofs
or to directly obtain simpler results from scratch.

\subparagraph*{The Special Case of Parity Domination.}
When taking a closer look at the precise statement of our lower bound
(and \cref{def:easyCases}),
we are reminded that both results do not apply for the same set of pairs.
Especially for the case of residue classes with modulus $\mname = 2$,
our algorithm solves the minimization version,
but \cref{thm:pwLowerBound} provides no matching lower bound.
As we may solve the decision problem for these cases in polynomial-time
via Gaussian elimination
(see, for example,~\cite{AndersonF98,dodisUniversalConfigurationsLightflipping2001,GoldwasserKT97,HalldorssonKT00-mod-2,Sutner89}),
our lower bound explicitly excludes these cases
by referring to them as ``easy''.

Surprisingly, exactly these easy cases
can be related to a single-player game called
\emph{\LightsOut} that was published 1995.
In this game the unassuming
player is presented with a $5 \times 5$ grid of switches
and lamps,
some or all of them initially turned on,
and the task is to turn off all lamps
by pressing the switches.
The catch is that every switch
flips not only the state of its corresponding lamp
(from ``on'' to ``off'' or vice versa),
but also the states of the neighboring lamps
in the grid~\cite{BermanBH21,FleischerY13}.%
\footnote{Similar games have also been released
  under the names \emph{Merlin} and \emph{Orbix} \cite{FleischerY13}.}

Since the order in which the buttons are pressed
does not matter
and every button has to be pressed at most once
as a second press would undo the first operation,
we can describe a \emph{solution} to an initial configuration
as a set of switches
that need to be flipped to turn all lights off.

When we assume that initially all lights are turned on,
then we can directly treat \LightsOut as a variant
of \srDomSet with $\sigma = \EVEN$ and $\rho = \ODD$.
We also refer to this problem, where the input is an arbitrary graph, as \ReflAllOff
since we assume that each switch triggers the corresponding lamp.
When this is not the case but still all lights are initially turned on,
we have $\sigma = \ODD = \rho$ and refer to the problem as \AllOff
as the corresponding switch does not trigger the associated lamp.

Despite the fact that it is easy to find \emph{some} solution for these two problems if one exists,
the minimization versions do not have such a trivial answer
and are known to be \NP-complete~\cite{%
  caroOddResidueDomination2001,%
  HalldorssonKT00-mod-2,sutner1988additive}.
Hence,
we investigate the minimization versions for these two problems
and complement the algorithmic result
from \cref{thm:twUpperBound} as follows.

\begin{restatable*}{mtheorem}{thmLightsOutLowerPW}
  \label{thm:lightsOut:lowerPW}
  Unless SETH fails,
  for all $\epsilon > 0$, there is no algorithm
  for each of the problems \ReflAllOff and \AllOff
  that can decide in time $(2-\epsilon)^\pw \cdot \abs{G}^{\Oh(1)}$
  whether there exists a solution
  of arbitrary size (the size is given as input)
  for a graph~$G$
  that is given with a path decomposition of width $\pw$.
  \qedhere
\end{restatable*}
Together with the lower bound for the general case,
we conclude that our algorithm is the best possible,
unless a major breakthrough for solving SAT happens.

\subparagraph*{Further Directions.}
When taking a step back,
the results of this work serve two purposes
which can be seen as starting points
for further investigations and improvements.
\begin{itemize}
  \item
        First, we settle the complexity of \srDomSet
        conclusively for the case of residue classes
        by providing matching upper and lower bounds.

        We later list some candidates
        that
        might allow similar improvements.
        Such improvements would, similar to our results,
        extend the list of problems
        started by Focke et al.~\cite{FockeMINSSW23}
        where significant improvements
        for the supposedly optimal algorithms are possible.

  \item
        Second, in comparison to the fairly complicated results
        for the case of finite sets in~\cite{FockeMINSSW23},
        this work can be seen as a significantly simpler introduction
        to those techniques that are relevant to obtain faster algorithms
        by exploiting the structural properties of the sets.

        We believe that
        for many other (parameterized) problems%
        ---including but not limited to \srDomSet---%
        the algorithms can be improved exponentially
        by using these new techniques.

\end{itemize}

In the following we list several possible directions
that could serve as a next step on the route to
a complete picture of the complexity of \srDomSet.

\begin{itemize}
  \item
        A first natural case could be pairs of two residue classes
        with \emph{different} moduli $\mname_\sigma$ and $\mname_\rho$.
        Then, the natural structural parameter $\mname$
        (which is the modulus in our case)
        is the greatest common divisor of $\mname_\sigma$ and $\mname_\rho$.
        In this setting
        the case $\mname = 1$
        is also relevant as this does not directly imply
        that the sets contain all natural numbers.

  \item
        A different direction considers the combination of a residue class
        with a finite or cofinite set.
        Focke et~al.\ show
        that representative sets~\cite{FominLPS16,KratschW20,MarxSS22,ShachnaiZ14}
        can be used to speed up the algorithm
        even further for the case of cofinite sets~\cite{focke_tight_2023_ub}.
        Independently of finding the optimal algorithm to handle the join operation
        for representative sets,
        it is not even clear what the optimal running time should be in such a case.

  \item
        Caro and Jacobson~\cite{caroNonzModDominating2003}
        introduced the problem \textsc{Non-$z(\bmod~k)$ Dominating~Set}
        which can also be described as a \srDomSet problem
        where the sets are complements of residue classes,
        which is equivalent to a finite union of residue classes.
        For example, for $z=0$ and $k = 3$,
        we set $\sigma = \{0,1,3,4,\dots \}
          = \{0, 3, 6,\dots\} \cup \{1,4,7,\dots\}$
        and $\rho = \{1, 2, 4, 5, 7, 8,\dots\}
          = \{1, 4, 7, \dots\} \cup \{2, 5, 8, \dots\}$.
        What is the optimal running time in this case?

  \item
        The general algorithm by Chapelle
        for the case
        when both sets are ultimately periodic
        has a running time single-exponential in treewidth
        despite being stated implicitly only~\cite{Chapelle10v5,Chapelle11}.
        What is the best running time for an algorithm solving \emph{all}
        cases of \srDomSet that are currently known to be fixed-parameter tractable?

  \item
        Are there more classes of sets for which there is an fpt algorithm parameterized by treewidth?
        Chapelle showed that once
        there are large gaps in the set,
        the problem becomes significantly harder~\cite{Chapelle10v5,Chapelle11}.

        \begin{theorem}
          [{\cite[Theorem~1]{Chapelle10v5} and \cite[Th\'eorème~3.3.1]{Chapelle11}}]
          \label{thm:intro:w1hardness}
          Write $\sigma$ for a set with arbitrarily large gaps
          between two consecutive elements
          (such that a gap of length $t$ is at distance $\poly(t)$ in $\sigma$),
          and write $\rho$ for a cofinite set with $\min \sigma \ge 1$ and $\min \rho \ge 2$.
          Then, the problem \srDomSet is \W[1]-hard
          when parameterized by the treewidth of the input graph.
        \end{theorem}

        Examples are the two natural sets
        where $\sigma = \{2^i \mid i \in \Nat\}$
        or when $\sigma$ is the set of all Fibonacci numbers~\cite{Chapelle10v5}.
        We observe that this is one of the rare cases
        where a problem is \W[1]-hard even when parameterizing by treewidth.

        The classification by Chapelle is not a dichotomy result
        in the sense that it provides a full classification between
        the fpt cases and the ones that are \W[1]-hard.
        For instance, what is the complexity
        for sets like $\sigma = \Nat \setminus \{2^i \mid i \in \Nat\}$
        which have gaps of constant size only
        but are not ultimately periodic?

  \item
        With our results,
        there are improved algorithms
        for the case when the sets are finite, cofinite or residue classes.
        Nevertheless,
        the description of the exact running time is highly non-uniform,
        that is, the exact complexity explicitly depends on the underlying set.
        Can we describe the complexity of optimal algorithms
        in a compact form as, for example,
        done by Chapelle for the general algorithm
        via finite automata~\cite{Chapelle10v5,Chapelle11}?
        This notation suffices
        to describe the state of a single vertex,
        but the representation of the structural insights
        leading to fewer states and faster algorithms
        remains open.

  \item
        Lastly,
        for which other problems besides \srDomSet
        can the techniques from our upper bounds
        (sparse languages and compression of vectors)
        be used to obtain faster algorithms?
        As the high-level idea of our lower bounds is quite modular,
        it should also be possible to use these concepts as blue-prints
        to achieve matching bounds for other problems as well.
\end{itemize}

\section{Technical Overview}

In this section we give a high-level overview of the results in this paper
and outline the main technical contributions we use.
We start by rigorously defining the main problem
considered in this work and
the property of a set being $\mname$-structured.

\begin{definition}[$(\sigma,\rho)$-sets, \srDomSet]
    \label{def:srSet}
    \label{def:srDomSet}
    Fix two non-empty sets $\sigma$ and $\rho$ of non-negative integers.

    For a graph $G$,
    a set $S \subseteq V(G)$ is a $(\sigma,\rho)$-set for $G$
    if and only if
    (1) for all $v \in S$,
    we have $\abs{N(v) \cap S} \in \sigma$,
    and (2), for all $v \in V(G) \setminus S$,
    we have $\abs{N(v) \cap S} \in \rho$.

    The problem \srDomSet asks for a given graph $G$,
    whether there is a $(\sigma,\rho)$-set $S$ or not.
\end{definition}

We also refer to the problem above as the \emph{decision} version.
The problem naturally also admits related problems
such as asking for a solution of a specific size,
or for the smallest or largest solution,
that is, the \emph{minimization} and \emph{maximization} version.

For the case of finite and cofinite sets,
Focke et al.~\cite{focke_tight_2023_ub, focke_tight_2023_lb} realized
that the complexity of \srDomSet significantly changes
(and allows faster algorithms)
when $\sigma$ and $\rho$ exhibit a specific structure,
which they refer to as \emph{$\mname$-structured}.

\begin{definition}[{$\mname$-structured sets~\cite[Definition 3.6]{focke_tight_2023_lb}}]
    \label{def:m-structured}
    Fix an integer $\mname \geq 1$.
    A set $\tau \subseteq \ZZ_{\ge 0}$ is \emph{$\mname$-structured} if
    all numbers in \(\tau\) are in the same residue class modulo \(\mname\), that is,
    if there is an integer $c^*$ such that
    \(c \equiv_{\mname} c^* \)
    for all $c \in \tau$.
\end{definition}

Observe that every set $\tau$ is $\mname$-structured for $\mname = 1$.
Therefore, one is usually interested in the largest $\mname$
such that a set is $\mname$-structured.
When considering two sets $\sigma$ and $\rho$,
we say that this pair is $\mname$-structured
if each of the two sets is $\mname$-structured.
More formally, assume that $\sigma$ is $\mname_\sigma$-structured
and $\rho$ is $\mname_\rho$-structured.
In this case the pair $(\sigma,\rho)$ is $\mname$-structured
where $\mname$ is the greatest common divisor
of $\mname_\sigma$ and $\mname_\rho$.
As in our case the sets $\sigma$ and $\rho$ are residue classes
modulo $\mname \ge 2$,
the sets are always $\mname$-structured.

In the following we first present the algorithmic result,
which outlines the proof of \cref{thm:twUpperBound}.
Afterward we move to the lower bounds where we consider
\cref{thm:pwLowerBound}
and finally we focus on the special case of \LightsOut
from \cref{thm:lightsOut:lowerPW}.

\subsection{Upper Bounds}

The basic idea to prove the upper bound is to provide
a dynamic programming algorithm that operates on a tree decomposition
of the given graph.
For each node of this decomposition we store all valid states,
where each such state describes how a possible solution,
i.e., a set of selected vertices,
interacts with the bags of the corresponding node.
We formalize this by the notion of a partial solution.

For a node $t$ with associated bag $X_t$,
we denote by $V_t$ the set of vertices introduced in the subtree rooted at $t$
and by $G_t$ the graph induced by these vertices.
We say that a set $S \subseteq V_t$
is a partial solution (for $G_t$) if
\begin{itemize}
    \item
          for each vertex $v \in S \setminus X_t$,
          we have $|N(v) \cap S| \in \sigma$,
          and
    \item
          for each vertex $v \in V_t \setminus (S\cup X_t)$,
          we have $|N(v) \cap S| \in \rho$.
\end{itemize}
The solution is partial in the sense that there are no constraints
imposed on the number of neighbors of the vertices in~$X_t$,
that is, only the vertices in $V_t \setminus X_t$
must have a valid number of neighbors.

We characterize the partial solutions by the states of the vertices in the bag.
When $\sigma$ and $\rho$ are finite or cofinite sets,
the largest reasonable state
is included in the respective set,
which is not necessarily the case for residue classes.
Consider two fixed, residue classes $\sigma$ and $\rho$
modulo $\mname \ge 2$.
Every selected vertex can have up to $\mname$ different states
and similarly, every unselected vertex can have $\mname$ different states.
Hence, for each bag, the number of relevant different partial solutions is bounded by
\((2 \mname)^{\abs{X_t}}.\)

\subparagraph*{High-level Idea.}
The crucial step to fast and efficient algorithms
is to provide a better bound on the number of states for each bag
when the sets $\sigma$ and $\rho$ are residue classes
modulo $\mname\ge 2$.
We denote by $\allStates$ the set of all possible states
a vertex might have in a valid solution.
Then, let $L \subseteq \allStates^{X_t}$
be the set of all possible state-vectors corresponding to partial solutions
for $G_t$.
Our first goal is to show that $\abs{L} \le \mname^{\abs{X_t}}$,
which means that not all
theoretically possible combinations of states
can actually have a corresponding partial solution in the graph.

Moreover, we also need to be able to combine two partial solutions
at the join nodes of the tree decomposition.
For a fast join operation, it does not suffice to bound the size of $L$.
This follows from the observation that
the convolution algorithm used to handle the join operation
does not depend on~$L$
but on the \emph{space} where the states come from.
In our case, the size of the space where~$L$ comes from
is still $(2\mname)^{\abs{X_t}}$, which is too large.
To decrease this size,
we observe that a significant amount of information
about the states of the vertices
can be inferred from other positions,
that is, we can \emph{compress} the vectors.

As a last step it remains to combine the states significantly faster than a naive algorithm.
To efficiently compute the join,
we use an approach based on the fast convolution techniques by van~Rooij
which was already used for the finite case \cite{Rooij_FastJoinOperations}.
However, we have to ensure that the compression of the vectors
is actually compatible with the join operation,
that is, while designing the compression we already have to take in mind
that we later join two (compressed) partial solutions together.
Since the compressed vectors are significantly simpler,
these states can now be combined much faster.

\subparagraph*{Bounding the Size of a Single Language.}
Recall that every partial solution $S$ can be described
by a state-vector $x \in \allStates^{n}$
where we abuse notation and set $n \deff \abs{X_t}$.
When $x$ describes the partial solution $S$,
we also say that $S$ is a witness for $x$.
We denote the set of the state-vectors of all partial solutions for $G_t$ as $L$.
We later refer to the set $L$ as the realized language of $(G_t,X_t)$.
To provide the improved bound on the size of $L$,
we decompose each state-vector~$x$ into two vectors:
The \emph{selection-vector} of~$x$, also called the $\sigma$-vector and
denoted by $\sigvec{x}$,
indicates whether each vertex in $X_t$ is selected or not.
The \emph{weight-vector} of $x$, denoted by $\degvec{x}$,
contains the number of selected neighbors of the vertices in~$X_t$.

The key insight into the improved bound is that
for two partial solutions of similar size
(with regard to modulo $\mname$),
the $\sigma$-vectors and the weight-vectors
of these two solutions are orthogonal.
This observation was already used to prove the improved bound
when $\sigma$ and $\rho$ are finite~\cite{focke_tight_2023_ub}.
We extend this result
to the case of residue classes.

To formally state the key insight, we define a graph with portals as a pair $(G,U)$, where $U \subseteq V(G)$ is a set of \emph{portal vertices} (see \cref{def:graph_with_portals}).
Intuitively, one can think of $G$ being the graph $G_t$, and $U$ the set $X_t$ for a node $t$ of a tree decomposition.

\begin{restatable*}[{Compare~\cite[Lemma~4.3]{focke_tight_2023_ub}}]{lemma}{lemStructuralProperty}
    \label{thm:structural_property}
    Let $\sigma$ and $\rho$ denote two residue classes
    modulo $\mname \ge 2$.
    Let $(G,U)$ be a graph with portals and
    let \(L \deff L(G,U) \subseteq \allStates^U\) denote its realized language.
    Consider two strings $x,y \in L$ with witnesses
    $S_x, S_y \subseteq V(G)$ such that
    \(|S_{x} \setminus U| \equiv_{\mname} |S_{y} \setminus U|\).
    Then, \(\sigvec{x}\cdot\degvec{y}
    \equiv_{\mname} \sigvec{y}\cdot\degvec{x}\).
\end{restatable*}

To prove this result, we count edges
between the vertices in $S_x$ and the vertices in $S_y$
in two different ways.
We first count the edges based on their endpoint in $S_x$.
These vertices can be partitioned into three groups:
(1)~the vertices contained in $U$,
(2)~the vertices outside $U$ which are not in $S_y$,
and~(3) the vertices outside $U$ which are in $S_y$.
Then, the number of edges $\abs{E(S_x \to S_y)}$ from $S_x$ to $S_y$
satisfies
\[
    \abs{E(S_x \to S_y)}
    \equiv_\mname
    \min\rho \cdot |S_x \setminus (S_y \cup U)|
    + \min\sigma \cdot |(S_x \cap S_y) \setminus U|
    + \sigvec{x} \cdot \degvec{y}
\]
because the sets $\sigma$ and $\rho$ are residue classes modulo $\mname$.
When counting the edges based on their endpoint in $S_y$,
the positions of $x$ and $y$ flip and the result follows.
As this property enables us to prove that the size of $L$ is small,
we refer to this property as \emph{sparse}.

Even though intuitively this orthogonality provides
a reason why the size of the language is not too large,
this does not result in a formal proof.
However, when fixing which vertices are selected,
that is, when fixing a $\sigma$-vector $\vec s$,
then there is an even stronger restriction on the values of the weight-vectors.
Instead of restricting the entire vector,
it actually suffices to fix the vector on a certain number of positions
which are described by some set $S$
to which we refer as $\sigma$-defining set.
If two $\sigma$-vectors of strings of the language agree on these positions from $S$,
then \emph{all} remaining positions of the two $\sigma$-vectors
must be identical as well.

With the sparseness property we then show that
it suffices to fix the $\sigma$-vectors on the positions from $S$
(which then determines the values on $\compl S$),
and the weight-vector on the positions from $\compl S$
(which then determines the weight-vector on the positions from $S$).
Formally, we prove \cref{thm:s_complement_determines_weight_vector}
which mirrors \cite[Lemma~4.9]{focke_tight_2023_ub}
in the case of residue classes.
We use $\sigvec{L}$ to denote the set $\{\sigvec{z} \mid z \in L\}$.

\begin{restatable*}[{Compare~\cite[Lemma~4.9]{focke_tight_2023_ub}}]{lemma}{lemSComplDeterminesWeight}
    \label{thm:s_complement_determines_weight_vector}
    Let $\sigma$ and $\rho$ denote two residue classes
    modulo $\mname \ge 2$.
    Let $L \subseteq \allStates^n$ be a sparse language
    with a \(\sigma\)-defining set $S$ for $\sigvec{L}$.
    Then, for any two strings \(x, y\in L\) with \(\sigvec{x} = \sigvec{y}\),
    the positions $\compl S$ uniquely characterize the
    weight vectors of \(x\) and \(y\), that is, we have
    \begin{equation*}
        \degvec{x}\vposition{\compl{S}}
        =
        \degvec{y}\vposition{\compl{S}}
        \quad\text{implies}\quad
        \degvec{x} = \degvec{y}.
        \qedhere
    \end{equation*}
\end{restatable*}

With this result it is straight-forward to bound the size of a sparse language of dimension~$n$.
Our goal is to bound the number of weight-vectors
that can be combined with a fixed $\sigma$-vector to form a valid type.
Assume we fixed a $\sigma$-vector $\vec s$
on the positions from~$S$.
Since this already determines the remaining positions of the $\sigma$-vector
(even if we do not know the values a priori),
the number of possible $\sigma$-vectors is at most $2^{|S|}$.
For the weight-vector there are $\mname$ choices
for each of the positions from $\compl S$.
Then, the values for the positions from $S$ are uniquely determined by those on $\compl S$
because of the previous result.
Using $\mname \ge 2$ this allows us to bound the size of a sparse language
by
\begin{align*}
    \mname^{\abs{\compl S}}
    \cdot
    2^{\abs{S}}
    \le \mname^n
    .
\end{align*}

\subparagraph*{Compressing Weight-Vectors.}
Based on the previous observations and results,
we focus on the analysis for a fixed $\sigma$-vector $\vec s$.
Though we could iterate over all
at most $2^{\abs{X_t}}$ possible $\sigma$-vectors
without dominating the running time,
the final algorithm only considers the $\sigma$-vectors
resulting from the underlying set $L$.
Hence, we assume that all vectors in $L$
share the same $\sigma$-vector $\vec s$.

When looking again at the bound for the size of $L$,
it already becomes apparent how we can compress the weight-vectors.
Recall that once we have fixed the entries of a weight-vector
of some vector $x \in L$
at the positions of $\compl S$,
the entries of the weight-vector on $S$ are predetermined
by \cref{thm:s_complement_determines_weight_vector}.
Hence, for the compressed vector
we simply omit the entries on the positions in $S$,
that is, the compressed weight-vector is the projection
of the original weight-vector to the dimensions from $\compl S$.

It remains to recover the original vectors from their compression.
As the implication from \cref{thm:s_complement_determines_weight_vector}
actually
does not provide the values for the positions on $S$,
it seems tempting
to store a single representative,
which we refer to as \emph{origin}-vector $o$
to recover the omitted values for all compressed vectors.
Unfortunately, this is not (yet) sufficient.

Observe that
\cref{thm:s_complement_determines_weight_vector},
which serves as basis for the compression, requires
that the weight-vector~$u$ and the origin-vector~$o$
agree on the coordinates from $\compl S$.
Therefore, it would be necessary to store one origin-vector
for each possible choice of values on $\compl S$,
which would not yield any improvement in the end.

In order to recover the values of the compressed weight-vector,
we use our structural property
from \cref{thm:structural_property} once more.
Intuitively, we use that changing the weight-vector at one position
(from $\compl S$, in our case),
has an effect on the value at some other position (from $S$, in our case).
Based on this idea we define an auxiliary vector,
which we refer to as \emph{remainder}-vector.
Intuitively, the entries of this vector capture the difference of
the weight-vector~$u$
and the origin-vector~$o$ on the positions in $\compl S$.
By the previous observation this also encodes
how much these two vectors~$u$ and~$o$
differ on the positions from $S$.
This remainder-vector then allows us
to efficiently decompress the compressed weight-vectors again.
In consequence,
the final compression reduces the size of the space
where the weight-vectors are chosen from,
which is a prerequisite for the last part of the algorithm.

\subparagraph*{Faster Join Operations.}
To obtain the fast join operation,
we apply the known convolution techniques
by van Rooij~\cite{Rooij_FastJoinOperations}.
As the convolution requires that all operations
are done modulo some small number,
we can directly apply it as every coordinate of the compressed vector
is computed modulo $\mname$.
As the convolution operates in the time of the space where the vectors are from,
we obtain an overall running time of $\mname^{\abs{X_t}}$ for the join operation.

The final algorithm is then a dynamic program
where the procedures for all nodes except the join node
follow the standard procedure.
For the join node, we iterate over all potential $\sigma$-vectors
of the combined language, then join the compressed weight-vectors,
and finally output the union of their decompressions.

By designing the algorithm such that we consider solutions of a certain size,
we achieve that the considered languages are sparse
and thus, the established machinery provides the optimal bound
for the running time. In total, we obtain \cref{thm:twUpperBound}.

\thmTwUpperBound*

\subsection{Lower Bounds}

After establishing the upper bounds,
we focus on proving matching lower bounds,
that is, we prove the previous algorithm to be optimal under SETH.
For all difficult cases, we provide a general lower bound
and for the easy cases that are solvable in polynomial time
but are non-trivial,
we prove a lower bound for the minimization version
by a separate reduction.
In the following we first focus on the difficult cases.

\begin{definition}[Easy and Difficult Cases]
    \label{def:easyCases}
    Let $\sigma$ and $\rho$ be two residue classes.
    We say that this pair is \emph{easy}
    if $0 \in \rho$ or
    \begin{itemize}
        \item $\sigma = \EVEN$ and $\rho = \ODD$,
              or
        \item $\sigma = \ODD$  and $\rho = \ODD$.
    \end{itemize}
    Otherwise, we say that the pair is \emph{difficult}.
\end{definition}

Clearly the case~$0 \in \rho$ is trivial
since the empty set is a valid solution.
For the other two cases we can formulate the problem
as a system of linear equations over~$\Field_2$:
for each vertex we create a variable indicating the selection status
and introduce one appropriately chosen constraint
involving the neighboring vertices.
Then, Gaussian elimination provides a solution in polynomial time.
We refer to \cite{AndersonF98,GoldwasserKT97,HalldorssonKT00-mod-2,Sutner89}
for a formal proof.
Thus, unless mentioned otherwise, we henceforth focus on the difficult cases.

Starting from the first SETH-based lower bounds
when parameterizing by treewidth by
Lokshtanov, Marx and Saurabh~\cite{lokshtanovKnownAlgorithmsGraphs2018}
(see also references in \cite{Lampis20} for other applications)
many reductions suffered from the following obstacle:
SETH provides a lower bound of the form $(2-\epsilon)^n$
whereas for most problems a lower bound of the form $(c-\epsilon)^\tw$ is needed
for some integer $c > 2$.
To bridge this gap, several technicalities are needed
to obtain the bound with the correct base.
Lampis introduced the problem (family) \qcsp,
which hides these technicalities and allows for cleaner reductions.
This problem generalizes $q$-SAT such that
every variable can now take $B$ different values,
that is, for $B=2$ this is the classical $q$-SAT problem.
Formally \qcsp is defined as follows.

\begin{restatable}[\qcsp~\cite{Lampis20}]{definition}{defCSP}
    \label{def:q-CSP-B}
    Fix two numbers $q, B \ge 2$.
    An instance of \qcsp is a tuple $(X, \CC)$
    that consists of a set $X$ of $n$ variables
    having the domain $D = \numb{B}$ each,
    and a set $\CC$ of constraints on $X$.
    A constraint $C$ is a pair $(\scope(C), \acc(C))$
    where $\scope(C) \in X^q$ is the scope of $C$
    and $\acc(C) \subseteq D^q$ is the set of accepted states.

    The task of the problem is to decide if $(X, \CC)$ is satisfiable,
    that is, decide if there exists an assignment
    $\pi \from X \to D$ such that, for all constraints $\CC$
    with $\scope(C) = (v_{\lambda_1},\dots,v_{\lambda_q})$
    it holds that
    $(\pi(v_{\lambda_1}),\dots,\pi(v_{\lambda_q})) \in \acc(C)$.
\end{restatable}
In other words,
the constraints specify valid assignments for the variables,
and we are looking for a variable assignment that satisfies all constraints.

Apart from introducing this problem,
Lampis also proved a conditional lower bound based on SETH
which allows us to base our reduction on this special type of
CSP.

\begin{restatable}[{\cite[Theorem~3.1]{Lampis20}}]{theorem}{qcspHardness}
    \label{thm:qcsp_hardness}
    For any $B \ge 2, \varepsilon > 0$ we have the following:
    assuming SETH, there is a $q$
    such that $n$-variable \qcsp with $\ell$ constraints
    cannot be solved in time
    $(B - \varepsilon)^n \cdot (n+\ell)^{\Oh(1)}$.
\end{restatable}
To obtain the correct lower bound the most suitable version of \qcsp
can be used, which then hides the unwanted technicalities.

In our case we cover numerous (actually infinitely many) problems.
This creates many positions in the potential proof where
(unwanted) properties of the sets $\sigma$ and $\rho$
have to be circumvented or exploited.
In order to minimize these places
and to make use of the special starting problem,
we split the proof in two parts.
This concept of splitting the reduction
has already proven to be successful for several other
problems~\cite{CurticapeanM16,focke_tight_2023_ub,MarxSS21,MarxSS22}.

As synchronizing point, we generalize the known \srDomSet problem
where we additionally allow that \emph{relations} are added to the graph.
Therefore, we refer to this problem as \srDomSetRel.
Intuitively one can think of these relations as constraints
that observe a predefined set of vertices, which we refer to as \emph{scope},
and enforce that only certain ways of selecting these vertices
are allowed in a valid solution.
To formally state this intermediate problem,
we first define the notion of a \emph{graph with relations}.

\begin{definition}[{Graph with Relations \cite[Definition~4.1]{focke_tight_2023_lb}}] \label{def:graphWithRelations}
    We define a \emph{graph with relations} as a tuple $G = (V,E,\CC)$,
    where $V$ is a set of vertices,
    $E$ is a set of edges on $V$,
    and $\CC$ is a set of relational constraints,
    that is, each $C\in \CC$ is in itself a tuple $(\scope(C), \acc(C))$.
    Here the \emph{scope} $\scope(C)$ of $C$ is an unordered tuple
    of $\abs{\scope(C)}$ vertices from $V$.
    Then, $\acc(C)\subseteq 2^{\scope(C)}$ is a $\abs{\scope(C)}$-ary relation
    specifying possible selections within $\scope(C)$.
    We also say that $C$ \emph{observes} $\scope(C)$.

    The size of $G$ is $\abs{G}\deff \abs{V}+\sum_{C\in \CC} \abs{\acc(C)}$.
    Slightly abusing notation, we usually do not distinguish between $G$
    and its underlying graph $(V,E)$.
    We use $G$ to refer to both objects depending on the context.
\end{definition}

We define the treewidth and pathwidth of a graph with relations
as the corresponding measure of the modified graph
that is obtained by replacing all relations by a clique
on the vertices from the scope.
See \cref{def:graphWithRelationsWidthMeasures} for a formal definition.

We lift the notion of $(\sigma,\rho)$-set from \cref{def:srSet}
to graphs with relations
by requiring that every relation has to be satisfied as well.
Hence,
the definition of \srDomSetRel follows naturally.
These definitions are a reformulation
of~\cite[Definition~4.3 and~4.8]{focke_tight_2023_lb}.

\begin{definition}[$(\sigma,\rho)$-Sets of a Graph with Relations,
        \srDomSetRel]
    \label{def:srDomSetRel}
    Fix two non-empty sets $\sigma$ and $\rho$ of non-negative integers.

    For a graph with relations $G = (V,E,\CC)$,
    a set $S \subseteq V$ is a \emph{$(\sigma,\rho)$-set} of $G$
    if and only if
    (1) $S$ is a $(\sigma,\rho)$-set of the underlying graph $(V,E)$
    and (2) for every $C \in \CC$, the set $S$ satisfies
    $S \cap \scope(C) \in \acc(C)$.
    We use $\abs{G}$ as the size of the graph
    and say that the \emph{arity} of $G$ is the maximum arity of a relation of $G$.

    The problem \srDomSetRel
    asks for a given graph with relations $G=(V, C, \CC)$,
    whether there is such a $(\sigma,\rho)$-set or not.
\end{definition}

With this intermediate problem,
we can now formally state the two parts of our lower bound proof.
The first step embeds the \qcsp problem (for appropriately chosen $B$)
into the graph problem \srDomSetRel.
We design the reduction in such a way that the resulting instance has a small pathwidth
(namely, roughly equal to the number of variables).
Combined with the conditional lower bound for \qcsp based on SETH
from \cref{thm:qcsp_hardness},
we prove the following intermediate lower bound.

\begin{restatable*}{lemma}{lowerBoundsRelations}
    \label{thm:lower_bound_relations}
    Let $\sigma$ and $\rho$ be two residue classes
    modulo $\mname \ge 2$.

    Then, for all $\epsilon > 0$, there is a constant $d$
    such that \srDomSetRel on instances of arity at most $d$
    cannot be solved in time $(\mname - \varepsilon)^{\pw} \cdot |G|^{\Oh(1)}$,
    where $\pw$ is the width of a path decomposition provided with the input instance~$G$,
    unless SETH fails.
\end{restatable*}

For the second step, we then remove the relations
from the constructed \srDomSetRel instance,
to obtain a reduction to the \srDomSet problem.
Observe that the construction from \cref{thm:lower_bound_relations}
works for the general case (even when $0 \in \rho$ is allowed).
Hence, our second step now exploits that the sets are difficult.

\begin{restatable*}{lemma}{reductionToRemoveRelations}
    \label{lem:relations:reductionToRemoveRelations}
    Let $\sigma$ and $\rho$ be two difficult residue classes
    modulo $\mname$.
    For all constants $d$,
    there is a polynomial-time reduction from \srDomSetRel
    on instances with arity $d$ given with a path decomposition of width $\pw$
    to \srDomSet
    on instances given with a path decomposition of width $\pw + \Oh(2^d)$.
\end{restatable*}

Combining these two intermediate results
directly leads to the proof of \cref{thm:pwLowerBound}.

\setcounter{mtheorem}{1}
\thmPwLowerBound
\begin{proof}
    Assume we are given a faster algorithm for \srDomSet
    for some $\epsilon > 0$.
    Let $d$ be the constant from \cref{thm:lower_bound_relations}
    such that there is no algorithm solving \srDomSetRel
    in time $(\mname-\epsilon)^\pw \cdot |G|^{\Oh(1)}$
    when the input instance~$G$ is
    given with a path decomposition of width $\pw$.

    Consider an instance $G$ of \srDomSetRel with arity $d$
    along with a path decomposition of width $\pw(G)$.
    We use \cref{lem:relations:reductionToRemoveRelations}
    to transform this instance into an instance $G'$ of \srDomSet
    with a path decomposition of width $\pw(G') = \pw(G) + \Oh(2^d)$.

    We apply the fast algorithm for \srDomSet to the instance $G'$
    which correctly outputs the answer for the original instance $G$
    of \srDomSetRel.
    The running time of this entire procedure is
    \begin{align*}
        \abs{G}^{\Oh(1)} + (\mname-\epsilon)^{\pw(G')} \cdot \abs{G'}^{\Oh(1)}
        = (\mname-\epsilon)^{\pw(G)+\Oh(2^d)} \cdot \abs{G}^{\Oh(1)}
        = (\mname-\epsilon)^{\pw(G)} \cdot \abs{G}^{\Oh(1)}
    \end{align*}
    since \(d\) is a constant only depending on $\epsilon$.
    This contradicts SETH and concludes the proof.
\end{proof}

The following highlights the main technical contributions
leading to the results from
\cref{thm:lower_bound_relations,lem:relations:reductionToRemoveRelations}.

\paragraph*{Step 1: Encoding the CSP as a Graph Problem.}

Focke et~al.\ already established the corresponding intermediate result
when $\sigma$ and $\rho$ are finite~\cite{focke_tight_2023_lb}.
Hence, we could try to reuse their lower bound
for \DomSetRel{\hat\sigma}{\hat\rho}
for two finite sets $\hat \sigma \subseteq \sigma$
and $\hat \rho \subseteq \rho$.
However, since $\sigma$ and $\rho$ are residue classes,
several solutions could be indistinguishable from each other
(not globally but locally from the perspective of a single vertex)
which would result in unpredictable behavior of the construction.
Thus, we need to come up with a new intermediate lower bound.

To prove this lower bound for \srDomSetRel,
we provide a reduction from \qcsp
where $B = \mname$
but reuse some ideas
from the known lower bounds
in~\cite{CurticapeanM16,focke_tight_2023_lb,MarxSS21,MarxSS22}.
This allows for a much cleaner reduction
(especially compared to the one from \cite{focke_tight_2023_lb})
that focuses on the conversion of a constraint satisfaction problem
into a vertex selection problem
without having to deal with technicalities.
Consult \cref{fig:lower:highLevel} for an illustration
of the high-level idea of the construction.

\begin{figure}[tp]
    \graphicspath{{figures/ipe},{figures/tikz}}
    \centering
    \begin{subfigure}[t]{.48\textwidth}
        \includegraphics[page=1,width=\textwidth]{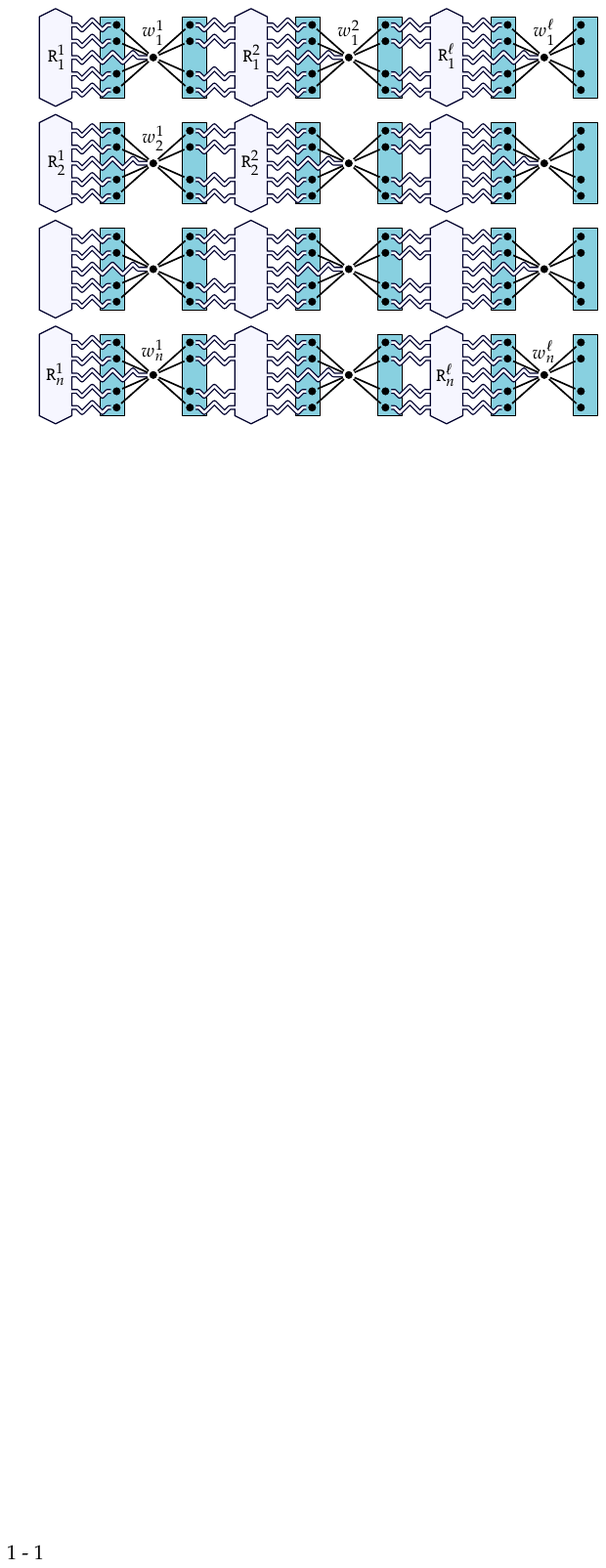}
        \caption{%
            The information vertices together with the managers
            and the consistency relations $\consistencyRelation_i^j$.}
    \end{subfigure}%
    \hfill
    \begin{subfigure}[t]{.48\textwidth}
        \centering
        \includegraphics[page=2,width=\textwidth]{construction-lower_bound.pdf}
        \caption{%
            The information vertices together with the managers
            and the constraint relations $\constraintRelation{j}$.}
    \end{subfigure}

    \caption{%
        A depiction of the construction from the lower bound
        where $\mname = 5$, $n = 4$, and $\ell = 3$.
    }
    \label{fig:lower:highLevel}
\end{figure}

Consider a \qcspa instance $I$ with $n$ variables and $\ell$ constraints.
To achieve a low treewidth (or actually pathwidth),
we construct a graph with $n \cdot \ell$ vertices,
which we refer to as \emph{information vertices},
that are arranged as an $n$ times $\ell$ grid;
rows corresponding to variables
and columns corresponding to constraints.
We refer to the information vertex from row $i$ and column $j$ as $w_i^j$.
We encode the $\mname$ different values of each variable
by the states of the information vertices in the graph.

To provide sufficiently many neighbors to these information vertices,
we introduce \emph{managers}.
In our case a manager consists of $2n$ blocks,
$n$ left blocks and $n$ right blocks,
and each block can provide up to $\mname-1$ neighbors to a single vertex.
We create one manager for each column
(i.e., constraint)
and assign one left block and one right block to each information vertex.
Then, we make each information vertex
adjacent to the two associated blocks by $\mname-1$ edges each.

We use the number of selected neighbors from the left block
to determine the state of an information vertex
(though the vertex might have selected neighbors in the right block too).
This directly relates
the states of the information vertices
to the variable assignments.

Recall that we create a separate manager for each column and that
the managers are not connected to each other.
Thus,
despite the mentioned correspondence,
even for a single row the information vertices can have different states.
Phrased differently, the encoded assignment is not necessarily consistent.
To keep treewidth low,
we cannot simply add a single big relation for each row
enforcing the intended behavior.
Instead, for each row $i$, we add a small \emph{consistency relation} $\consistencyRelation_i^j$
between every two consecutive columns $j$ and $j+1$.

The relation $\consistencyRelation_i^j$
ensures the consistency
between the information vertices $w_i^j$ and $w_i^{j+1}$,
and thus, additionally
observes the right block of $w_i^j$
and the left block of $w_i^{j+1}$.
First, $\consistencyRelation_i^j$ ensures that information vertex~$w_i^j$ is unselected.
Now assume that
$w_i^j$ has $b_1$ selected neighbors in its right block
and $w_i^{j+1}$ has $b_2$ selected neighbors in its left block.
Then, relation~$\consistencyRelation_i^j$ ensures that $b_2$ complements $b_1$
in the sense that
$b_2 = \min\rho - b_1 \mod \mname$,
that is,
$b_2$ is the smallest number such that
$b_1+b_2 \equiv_\mname \min\rho$.

It remains to analyze the influence of the information vertices
themselves
on the consistency of the encoded assignment.
By our construction,
information vertex $w_i^j$
receives $b_0$ neighbors from its left block of the manager,
receives $b_1$ neighbors from its right block of the manager,
and receives no other neighbors.
Since we consider a solution,
vertex $w_i^j$ must have a valid number of neighbors,
that is,
the solution must satisfy $b_0+b_1 \in \rho$.
Since $\rho$ is a residue class modulo $\mname$,
we get $b_0 + b_1 \equiv_\mname \min\rho$
which implies that $b_1 = \min\rho - b_0 \mod \mname$.
When combining this with the observation from the previous paragraph,
where we consider two different information vertices,
we get $b_2 = \min\rho - (\min\rho - b_0 \mod \mname) \mod \mname$
and hence, $b_2 = b_0 \mod \mname$
which implies that all information vertices of one row
have the same state.

As a last step we encode the constraints of the CSP instance.
For each constraint $C_j$ we add one
\emph{constraint relation} $\constraintRelation{j}$
which observes,
for each variable appearing in $C_j$,
the neighbors of
the corresponding information vertex in the left block of the manager
(they are needed to infer the state of the information vertices).
The relation $\constraintRelation{j}$ then accepts a selection of these vertices
if and only if it corresponds to a satisfying assignment.

This concludes the lower bound for the intermediate problem \srDomSetRel.
Next, we remove the relations and replace them by appropriate gadgets
to lift the result to \srDomSet.

\paragraph*{Step 2: Realizing the Relations.}
Formally, the second step is a reduction
from \srDomSetRel to \srDomSet.
We replace each relation by a suitable gadget
that precisely mimics the behavior of the original relation,
that is, we \emph{realize} the relation.
The realization gadget accepts a selection of vertices if and only if
the original relation also accepted this selection.
Moreover, such a gadget
must not add any selected neighbors to a vertex from the scope,
as that could affect the existence of a solution
(in the positive but also in the negative).
See \cref{def:realization} for a formal definition of a realization.

Curticapean and Marx~\cite{CurticapeanM16}
show that just two types of relations suffice
to realize arbitrary relations.
Focke et al.\ prove that for \srDomSet
only \HWeq{1} relations are needed~\cite[Corollary~8.8]{focke_tight_2023_lb},
that is, once we can realize such \HWeq{1} relations,
then every relation can be realized.
The \HWeq{1} relation
accepts if exactly one vertex from the scope of the relation is selected,
that is, if the Hamming weight of the $\sigma$-vector is exactly one.
We strengthen this result further by using an observation from \cite{MarxSS21}
such that only
realizations of \HWeq{1} with arity one, two or three
are needed.

To realize these relations,
we use an auxiliary relation.
For some set $\tau$, the relation \HWin{\tau} accepts,
if and only if the number of selected vertices from the scope of the relation,
i.e., the Hamming weight of the $\sigma$-vector,
is contained in $\tau$.
Once we set $\tau-k \deff \{t -k \mid t \in \tau\}$
to simplify notation,
our main result for realizing relations reads as follows.

\begin{restatable*}{lemma}{relationsHWOne}
    \label{lem:relations:hwOne}
    Let $\sigma$ and $\rho$ be two difficult residue classes
    modulo $\mname$.
    Then, the relation \HWone can be realized.
\end{restatable*}
When $\sigma$ and $\rho$ are difficult we have $\mname \geq 3$, and hence
this gadget directly gives \HWeq{1}
when restricting to arity at most $3$
as $\min\rho+1, \min\rho+2 \notin \rho$.
Thus, together with the previous intermediate lower bound
this concludes the proof of the lower bound.

For $\mname = 2$
the decision version is easy, so we cannot expect to realize the \HWeq{1} relation.
So, we consider the \emph{minimization version} instead
and focus on the two non-trivial cases;
for $\rho = \ODD$, we consider
$\sigma = \ODD$ and $\sigma = \EVEN$.
For the lower bounds from \cref{thm:lightsOut:lowerPW},
we modify the known \NP-hardness reductions by Sutner~\cite{sutner1988additive}
to keep pathwidth low and to obtain the matching bounds.

\section{Preliminaries}
\label{sec:prelims}

In this section we introduce some notation and concepts
which we use in the following proofs.
As we are dealing with the same problem as the work in~\cite{FockeMINSSW23},
we reuse many of their concepts and notation
so that our results blend in more seamlessly.

\subsection{Basic Notation}
\label{sec:prelims:numbers_sets_strings_vectors}

We write $\Nat = \{ n \in \ZZ \mid n \ge 0\}$
to denote the set of non-negative integers.
For integers $i$ and $j$, we write \(\fragment{i}{j}\)
for the set $\{n \in \ZZ \mid i \leq n \leq j\}$.

When \(\Sigma\) is some alphabet,
we write \(\Sigma^{n}\) for the set (or \emph{language}) of all strings
of length \(n\) over $\Sigma$.
For a string $s \in \Sigma^n$,
we index the positions of $s$ by integers from $\numb{n}$
such that \(s =  s\position{1} s\position{2}\cdots  s\position{n}\).
For a set of positions \(P \subseteq \numb{n}\),
we write \(s\position{P} \deff \bigcirc_{p \in P}\, s\position{p} \)
for the string of length \(|P|\)
that contains only the symbols of $s$ whose indices are in \(P\).
We extend these notations in the standard way to vectors.

We use the notation \(\Sigma^X\)
for a finite set $X$ (for instance, a set of vertices of a graph),
to emphasize that we index strings from \(\Sigma^{\abs{X}}\)
by elements from \(X\).

\subsection{Graphs}
\label{sec:prelims:graphs}

Unless mentioned otherwise,
a \emph{graph} is a pair $G = (V(G),E(G))$
with a finite vertex set $V(G)$
and a finite edge set $E(G) \subseteq \binom{V(G)}{2}$.
When the graph is clear from the context, we also just write $G=(V,E)$.
In this paper, all considered graphs are undirected and simple,
that is, they have no loops or multiple edges,
unless explicitly stated otherwise.
For an edge $\{u,v\} \in E(G)$, we also write $uv$ to simplify notation.

For a vertex $v \in V(G)$,
we denote by \(N_G(v)\) the \emph{(open) neighborhood} of $v$
that is, $N_G(v) \deff \{w \in V(G) \mid vw \in E(G)\}$.
We denote the \emph{closed neighborhood} of $v$ by
$N_G\position{v} \deff N_G(v) \cup \{v\}$.
The \emph{degree} of $v$ is the size of its open neighborhood, that is,
$\deg_G(v) \deff |N_G(v)|$.
When $X \subseteq V(G)$ is a set of vertices,
we define the closed neighborhood of $X$ as
$N_G\position{X} \deff \bigcup_{v \in X} N_G\position{v}$
and the open neighborhood of $X$ as $N_G(X) \deff N_G\pos{X} \setminus X$.
We may drop the subscript $G$ in all settings
if the graph is clear from the context.
We define the size of a graph as $\abs{G} = \abs{V(G)}$.

For a vertex set $X \subseteq V(G)$,
we denote by $G\position{X}$ the \emph{induced subgraph} on the vertex set $X$.
By $G-X$ we denote the induced subgraph on the complement of $X$,
and formally define $G - X \deff G\position{V(G) \setminus X}$.
For an edge set $X \subseteq E(G)$, we define $G - X \deff (V(G),E(G)\setminus X)$.

\subsection{Treewidth}
\label{sec:prelims:treewidth}

Our algorithmic results are based on tree decompositions.
For completeness,
we restate the definition and the basic properties of such a decomposition.
We refer the reader to~\cite[Chapter~7]{cyganParameterizedAlgorithms}
for a more detailed introduction to the concept.

Consider a graph $G$.
A \emph{tree decomposition} of $G$ is a pair $(T,\beta)$ that consists
of a rooted tree $T$ and a function $\beta\from V(T) \to 2^{V(G)}$ such that
\begin{enumerate}
      \item $\bigcup_{t \in V(T)} \beta(t) = V(G)$,
      \item for every edge $vw \in E(G)$, there is some node $t \in V(T)$ such that $\{u,v\} \subseteq \beta(t)$, and
      \item for every $v \in V(G)$, the set $\{t \in V(T) \mid v \in \beta(t)\}$ induces a connected subtree of $T$.
\end{enumerate}
The \emph{width} of a tree decomposition $(T,\beta)$ is defined as $\max_{t \in V(T)}|\beta(t)|-1$.
The \emph{treewidth} of a graph $G$, denoted by $\tw(G)$, is the minimum width of a tree decomposition of $G$.

In order to design algorithm based on tree decompositions,
it is usually helpful to use \emph{nice tree decompositions}.
Let $(T,\beta)$ denote a tree decomposition
and write $X_t \deff \beta(t)$ for the bag of a node $t \in V(T)$.
We say that \((T,\beta)\) is a \emph{nice tree decomposition},
or \emph{nice} for short,
if the tree $T$ is a binary tree rooted at some node $r$
such that $X_r = \emptyset$
and every node $t \in V(T)$ of the decomposition has one of the following types:
\begin{description}[itemindent=!]
      \item[Leaf Node:]
            Node $t$ has no children and an empty bag,
            that is, $t$ is a \emph{leaf} of $T$ and $X_t = \emptyset$.
      \item[Introduce Node:]
            Node $t$ has exactly one child $t'$
            and $X_t = X_{t'} \cup \{v\}$ for some $v \notin X_{t'}$,
            we say that the vertex $v$ \emph{is introduced at $t$}.
      \item[Forget Node:]
            Node $t$ has exactly one child $t'$ and $X_t = X_{t'} \setminus \{v\}$
            for some $v \in X_{t'}$,
            we say that the vertex $v$ \emph{is forgotten at $t$}.
      \item[Join Node:]
            Node $t$ has exactly two children $t_1,t_2$ and $X_t = X_{t_1} = X_{t_2}$.
\end{description}

In time $\Oh(\tw^{2} \cdot \max\{|V(G),V(T)|\})$,
we can transform every given tree decomposition $(T, \beta)$
of width $\tw$ for a graph $G$
into a nice tree decomposition of size $\Oh(\tw \cdot V(T))$ of the same width
(see, for example, \cite[Lemma 7.4]{cyganParameterizedAlgorithms}).

The concepts \emph{path decomposition} and \emph{pathwidth}
follow analogously by additionally requiring that the underlying tree~$T$
is a path.

\subsection{Partial Solutions and States}
\label{sec:prelims:partial_solutions_states}

When formally describing the algorithm and designing the lower bounds,
we frequently consider subgraphs (of the final graph)
and argue about the intersection of solutions with this subgraph.
In order to describe the interaction between the subgraph
and the remaining graph,
we use the notion of a \emph{graph with portals}.
These portals are then separating the subgraph from the remaining graph.
For the algorithmic result the portals are the vertices of the bag,
and for the lower bounds the portals are vertices in the scope of some relation.

\begin{definition}[{Graph with Portals; compare~\cite[Section~3.2]{focke_tight_2023_lb}}]
      \label{def:graph_with_portals}
      A \emph{graph with portals} $G$ is a pair $(G', U)$, where $G'$ is a graph and $U \subseteq V(G')$.
      If $U = \{u_1,\dots,u_k\}$, then we also write $(G', u_1,\dots,u_k)$ instead of $(G', U)$.

      If it is clear from the context, we also refer to a graph with portals simply as a graph.
\end{definition}

With the formal notion of a graph with relations,
we can now also define a \emph{partial solution}.
These partial solutions capture the intersection of a (hypothetical) solution
with the vertices from a graph with portals.

\begin{definition}[{Partial Solution;~\cite[Definition~3.7]{focke_tight_2023_lb}}]
      \label{def-partial-sol}
      \label{def:partialsol}
      Fix a graph with portals $(G,U)$.
      A set \(S \subseteq V(G)\) is a \emph{partial solution} (with respect to \(U\)) if
      \begin{enumerate}
            \item \label{item:realizable-1s}\label{item:realizable-1}
                  for each $v \in S \setminus U$, we have $|N(v) \cap S| \in \sigma$,
                  and
            \item \label{item:realizable-3s}\label{item:realizable-3}
                  for each $v \in V(G) \setminus (S\cup U)$, we have $|N(v) \cap S| \in \rho$.
                  \qedhere
      \end{enumerate}
\end{definition}

When designing the algorithm or when constructing the gadgets,
we usually do not want to argue about every possible partial solution
but identify those solutions that behave equivalently
when considering their extension to the remaining graph.
Formally, we associate with each partial solution for a graph with relations
a state for each portal vertex.
This state describes whether a vertex is selected or not
and how many neighbors it gets from this partial solution.

In order to argue about these different states,
we define different sets of states.%
\footnote{Note that this type of notation is standard for the \srDomSet{} problem and, for example, also used in \cite{Telle94,tellePracticalAlgorithmsPartial1993,RooijBR09,focke_tight_2023_lb,focke_tight_2023_ub}.}
\begin{definition}[States]
      \label{def:vertex_states}
      For all $i \in \Nat$,
      we define a $\sigma$-state $\sigma_i$
      and a $\rho$-state $\rho_i$.
      We then define a set
      \begin{align}
            \sigStates \deff \begin{cases}
                                   \{\sigma_0,\dots,\sigma_{\mname_\sigma - 1}\} & \text{ if $\sigma$ is a residue class modulo } \mname_\sigma, \\
                                   \{\sigma_i \mid i \in \Nat\}                  & \text{otherwise},                                             \\
                             \end{cases}
      \end{align}
      of $\sigma$-states, and a set
      \begin{align}
            \rhoStates \deff \begin{cases}
                                   \{\rho_0,\dots,\rho_{\mname_\rho - 1}\} & \text{ if $\rho$ is a residue class modulo } \mname_\rho, \\
                                   \{\rho_i \mid i \in \Nat\}              & \text{otherwise},                                         \\
                             \end{cases}
      \end{align}
      of $\rho$-states.%
      \footnote{%
            We remark that in this work we do not use real numbers.
            Thus, the set~$\rhoStates$ always denotes a set of $\rho$-states.
      }
      We denote by $\allStates \deff \sigStates \cup \rhoStates$
      the set of all states.
\end{definition}

Even if the sets are residue classes,
these states capture all relevant states.
For example, consider the case in which $\sigma$ and $\rho$
are the sets of all even integers.
Then, it does not matter whether a vertex has $2$
or $42$ selected neighbors;
both states show exactly the same behavior
with regard to adding more neighbors to the respective vertex.

With the above notation about strings and vectors in mind,
we usually define strings over the alphabets
$\sigStates$, $\rhoStates$, or $\allStates$.
In this case we usually index the positions of such strings
by a vertex of a bag or by a vertex of the scope of a relation.
Then, we can relate the partial solutions
and the states of the portal vertices via \emph{compatible strings}.

\begin{definition}[{Compatible Strings; extension of~\cite[Definition~3.8]{focke_tight_2023_lb}}]
      \label{def-compatible-general}
      \label{def:compatibleString}
      Fix a graph with portals $(G,U)$.
      A string \( x \in \allStates^{U}\) is \emph{compatible with $(G,U)$}
      if there is a partial solution $S_{x} \subseteq V(G)$ such that
      \begin{enumerate}
            \item
                  for each $v \in U \cap S_{x}$, we have \( x\position{v} = \sigma_{s}\), where
                  \begin{itemize}
                        \item $s = |N(v) \cap S_{x}| \bmod \sigMod$ when $\sigma$ is a residue class modulo $\sigMod$,
                        \item $s = |N(v) \cap S_{x}|$
                              otherwise,
                  \end{itemize}
            \item for each $v \in U \setminus S_{x}$, we have
                  \( x\position{v} = \rho_{r}\), where
                  \begin{itemize}
                        \item $r = |N(v) \cap S_{x}| \bmod \rhoMod$ when $\rho$ is a residue class modulo $\rhoMod$,
                        \item $r = |N(v) \cap S_{x}|$
                              otherwise.
                  \end{itemize}
      \end{enumerate}
      We also refer to the vertices in $S_{x}$ as being \emph{selected} and say that $S_{x}$ is a \emph{(partial) solution}, \emph{selection}, or \emph{witness} that \emph{witnesses} $x$.
\end{definition}

Then, given a graph with portals $(G, U)$,
the set of all strings compatible with this graph is of utmost importance,
we call this set the \emph{realized language} of the graph.

\begin{definition}[{Realized Language and $L$-provider;~\cite[Definition~3.9]{focke_tight_2023_lb}}]\label{def:provider}
      For a graph with portals \((G, U)\), we define its realized language as
      \[L(G,U) \coloneqq \{x \in \allStates^U \mid x \text{ is compatible with } (G,U)\}.\]

      For a language \(L \subseteq \allStates^U\), we say that \emph{$(G,U)$}
      is an \(L\)-realizer if \(L = L(G,U)\).

      For a language \(L \subseteq \allStates^U\), we say that \emph{$(G,U)$}
      is an \(L\)-provider if \(L \subseteq L(G,U)\).
\end{definition}

\section{Upper Bound}
\label{sec:upperTw}

In this section we prove the algorithmic result for \srDomSet
when the sets $\sigma$ and $\rho$ are residue classes modulo $\mname \ge 2$
when parameterizing by treewidth.
This closely matches our lower bound from \cref{thm:pwLowerBound}.
Note that when $\mname = 1$, the sets contain all numbers
and thus, the problem is trivially solvable
(any subset of the vertices is a solution).
Formally, we prove \cref{thm:algorithm_runtime}.

\thmTwUpperBound*

We use the known results from \cite{focke_tight_2023_ub}
as a starting point and extend these results to the case
when we have residue classes as sets.
This requires us to adjust several concepts and introduce new results
where the results from the finite case are not applicable.
However, as we are dealing with residue classes only,
this allows a simpler presentation of the results.

As already mentioned earlier,
the final algorithm is based on dynamic programming
on the tree decomposition of the given graph.
In order to obtain an upper bound,
which matches the stated lower bound from \cref{thm:pwLowerBound},
there are three main problems that we have to overcome:
\begin{itemize}
    \item
          We need to prove a tight bound on the \emph{number of different types} of
          partial solutions for a certain graph.
          We store these states for each node
          of the tree decomposition
          in order to correctly solve the problem.
    \item
          Since the running time of the join operation (from the next step)
          depends on the size of the space,
          we additionally need to \emph{compress} the space of state-vectors.
          A simple bound on the number of states is actually not sufficient.
    \item
          To compute the join nodes efficiently
          without dominating the running time of the problem,
          we need to employ a fast convolution technique to compute the join nodes.
\end{itemize}

We handle these problems in the stated order and
start in \cref{sec:upperTW:boundSingle} by bounding the
number of different types of partial solutions for a graph with portals.
Then, in \cref{sec:upperTw:compression} we formally introduce
the concept of compressions,
which we then need in \cref{sec:upperTW:join} to state the procedure
which computes the join nodes efficiently.
As a last step we state the final algorithm in \cref{sec:upperTW:dp}.

Before moving to the first step,
we formally define how a state vector can be decomposed
into a $\sigma$-vector and weight-vector.
The definition closely mirrors~\cite[Definition~4.2]{focke_tight_2023_ub}
adjusted to our setting.

\begin{definition}[Decomposing States]
    \label{def:vectors}
    For a string \(x \in \allStates^n\) representing a state-vector,
    we define
    \begin{itemize}
        \item
              the \emph{\(\sigma\)-vector of \(x\)} as \(\sigvec{x} \in \{0, 1\}^n\)
              with
              \[
                  \sigvec{x}\vposition{i} \deff \begin{cases}
                      1, & \text{if } x\position{i} \in \sigStates, \\
                      0, & \text{if } x\position{i} \in \rhoStates,
                  \end{cases}
              \]
        \item
              and the \emph{weight-vector of \(x\)} as \(\degvec{x} \in \ZZ_{\geq 0}^n\)
              with
              \[
                  \degvec{x}\vposition{i} \deff c,
                  \quad\text{where } x\position{i} \in \{\sigma_c, \rho_c \}
                  .
              \]
    \end{itemize}
    We extend this notion to languages \(L \subseteq \allStates^n\)
    in the natural way.
    Formally,
    we write \(\sigvec{L} \deff \{ \sigvec{x} \mid x \in L\}\)
    for the set of all \(\sigma\)-vectors of~\(L\),
    and we write \(\degvec{L} \deff \{ \degvec{x} \mid x \in L\}\)
    for the set of all weight-vectors of \(L\).
\end{definition}

\subsection{Bounding the Size of A Single Language}
\label{sec:upperTW:boundSingle}

We start our theoretical considerations by considering
the realized language of a graph with portals.
The goal of this section is to show that for any graph with portals $(G,U)$,
the size of the realized language is small.

To formally state the result,
we first introduce the concept of a \emph{sparse} language.

\begin{definition}[{Sparse Language; compare \cite[Section~4.1]{focke_tight_2023_ub}}]
    \label{def:sparse_language}
    Let $\sigma$ and $\rho$ denote two residue classes
    modulo $\mname \ge 2$.
    We say that a language $L \subseteq \mathbb{A}^n$ is \emph{sparse}
    if for all $x,y \in L$
    the following holds:
    $\sigvec{x} \cdot \degvec{y} \equiv_\mname \sigvec{y} \cdot \degvec{x}$.
\end{definition}

Even though the definition of a sparse language does not
say anything about the size of the realized language,
we show that this is the correct notation by proving the following result.
For the case when both sets are finite,
the result is covered by~\cite[Theorem~4.4]{focke_tight_2023_ub}
which we extend to residue classes.

\begin{restatable*}{lemma}{thmUpperBoundOnLanguage}
    \label{thm:upper:boundOnLanguage}
    Let $\sigma$ and $\rho$ denote two residue classes
    modulo $\mname \geq 2$.
    Every sparse language $L \subseteq \mathbb{A}^n$ satisfies
    $|L| \leq \mname^n$.
\end{restatable*}

Instead of later directly proving that a language is sparse,
we usually use the following sufficient condition,
which intuitively reads as,
``if the number of selected non-portal vertices for every pair of solutions
is equal modulo $\mname$, then the language is sparse''.
From an intuitive perspective this can be understood as proving that
every language is the union of $\mname$ sparse languages.

\lemStructuralProperty
\begin{proof}
    We observe that the proof of~\cite[Lemma~4.3]{focke_tight_2023_ub}
    does not use the finiteness of $\sigma$ or $\rho$.
    Although this means that the proof works for our setting as well,
    we provide it here for completeness.%
    \footnote{%
        Actually their proof and our proof only require
        that $\sigma$ and~$\rho$ are $\mname$-strucutured.
    }
    We denote by $s \in \sigma$ and $r \in \rho$
    two elements of the sets
    in the following.

    We prove the claim by counting edges between $S_x$ to $S_y$
    in two different ways.
    First we count the edges as going from $S_x$ to $S_y$.
    Let $E(X \to Y) \deff \{(u,v) \in E(G) \mid u \in X, v \in Y\}$
    denote the set of edges from~$X$ to~$Y$.

    We partition the vertices of $S_x$ into
    (1) vertices that are neither in $U$ nor $S_y$,
    (2) the vertices that are in $S_y$,
    and (3) the vertices that are in $U$.
    Every vertex in $S_x$ that is neither in $S_y$ nor in $U$,
    is unselected in $S_y$ and thus,
    the number of neighbors it has in $S_y$ must be in $\rho$
    because $S_y$ is a partial solution.
    Especially, this number must be congruent $r\bmod \mname$
    because $\rho$ is $\mname$-structured.
    If a vertex is in $S_x$ and in $S_y$ but not in $U$,
    then the number of neighbors it has in $S_y$
    must be congruent $s \bmod \mname$ for the analogous reason.
    It remains to elaborate on the edges leaving from vertices in $U$.
    Such a vertex only contributes to the count if it is selected in $S_x$,
    i.e., if the $\sigma$-vector is $1$ at the considered position.
    The actual number of neighbors such a vertex receives in $S_y$
    is determined by the entry of the weight-vector for $S_y$
    at the corresponding coordinate.
    Combining these observations yields,
    \begin{align*}
        |{E(S_x \to S_y)}|
        \equiv_\mname~
         & r \cdot |S_x \setminus (S_y \cup U)|        \\
         & + s \cdot |(S_x \cap S_y) \setminus U|      \\
         & + \sigvec{x} \cdot \degvec{y}               \\
        \equiv_\mname~
         & r \cdot (|S_x \setminus U| - |S_x \cap S_y|
        + |S_x \cap S_y \cap U|)                       \\
         & + s \cdot |(S_x \cap S_y) \setminus U|      \\
         & + \sigvec{x} \cdot \degvec{y}
        .
    \end{align*}

    Flipping the roles of $S_x$ and $S_y$ gives
    \begin{align*}
        |{E(S_y \to S_x)}|
        \equiv_\mname & ~
        r \cdot (|S_y \setminus U| - |S_x \cap S_y|
        + |S_x \cap S_y \cap U|)                                   \\
                      & + s \cdot \abs{(S_y \cap S_x) \setminus U} \\
                      & + \sigvec{y} \cdot \degvec{x}
        .
    \end{align*}

    Observe that we have $|E(S_x \to S_y)| = |E(S_y \to S_x)|$
    since both sets count every edge from solution $S_x$
    to solution $S_y$ exactly once and the edges are undirected.
    Using these properties and the assumption that
    $|S_x \setminus U| \equiv_\mname |S_y \setminus U|$,
    we get from the combination of the above equations that
    \[
        \sigvec{x} \cdot \degvec{y} \equiv_\mname \sigvec{y} \cdot \degvec{x}
        .
        \qedhere
    \]
\end{proof}

Toward showing that sparse languages are indeed of small size,
we establish one auxiliary property that
immediately follows from the definition of a sparse language.

\begin{lemma}[{Compare~\cite[Lemma~4.6]{focke_tight_2023_ub}}]
    \label{thm:sparse_language_same_sigma}
    Let $\sigma$ and $\rho$ denote two residue classes
    modulo $\mname \ge 2$.
    Let
    $L \subseteq \mathbb{A}^n$ be a sparse language.
    For any three strings $x, y, z \in L$ with $\sigvec{x} = \sigvec{y}$,
    we have
    \[
        \left (\degvec{x} - \degvec{y} \right ) \cdot \sigvec{z}
        \equiv_\mname 0
        .
    \]
\end{lemma}
\begin{proof}
    We reproduce the proof of the original result in~\cite{focke_tight_2023_ub}.
    Let $L$ be a sparse language,
    and consider state-vectors $x, y, z \in L$
    with $\sigvec{x} = \sigvec{y}$.
    Because $L$ is sparse we know that
    \[
        \sigvec{z} \cdot \degvec{x} \equiv_\mname \sigvec{x} \cdot \degvec{z}
    \]
    and by $\sigvec{x} = \sigvec{y}$, this gives
    \[
        \sigvec{z} \cdot \degvec{x} \equiv_\mname \sigvec{y} \cdot \degvec{z}.
    \]
    When using the properties of sparsity for $y$ and $z$
    we have
    \[
        \sigvec{y} \cdot \degvec{z} \equiv_\mname \sigvec{z} \cdot \degvec{y}
        .
    \]
    Combining the last two results, yields
    \[
        \sigvec{z} \cdot \degvec{x} \equiv_\mname \sigvec{z} \cdot \degvec{y}
    \]
    and rearranging this equivalence concludes the proof.
\end{proof}

Intuitively, the condition of the previous lemma
significantly restricts the number of possible weight-vectors
of two different strings with the same $\sigma$-vector.
This restriction is the reason why sparse languages
turn out to be of small size.
To formalize how restrictive this condition is exactly,
we introduce the notion of $\sigma$-defining sets.
Such a set depends only on a set of binary vectors,
for example, the $\sigma$-vectors of a sparse language,
and provides us $\sigma$-vectors of the language that are very similar.

\begin{definition}[{$\sigma$-defining set; compare~\cite[Definition~4.7]{focke_tight_2023_ub}}]\label{def:sig-defining-set}
    \label{rem:partition-depends-on-sigma-vectors}
    A set \(S \subseteq \numb{n}\) is
    \emph{\(\sigma\)-defining for \(X \subseteq \{0,1\}^n\)}
    if \(S\) is an inclusion-minimal set of positions that
    uniquely characterize the vectors of $X$, that is,
    for all \(u, v \in X\),
    we have
    \begin{equation*}
        u\vposition{S} = v\vposition{S}
        \quad\text{implies}\quad
        u = v.
    \end{equation*}
    For a $\sigma$-defining set $S$ of a set $X \subseteq \{0,1\}^n$,
    we additionally define the complement of $S$
    as $\compl S \deff \numb{n} \setminus S$.
\end{definition}

The name of the $\sigma$-defining sets comes from the fact
that we usually use a set of $\sigma$-vectors for the set $X$.

\begin{fact}[{\cite[Remark~4.8]{focke_tight_2023_ub}}]
    \label{rem:witness}
    As a \(\sigma\)-defining \(S\) of a set $X \subseteq \{0,1\}^n$
    is (inclusion-wise) minimal,
    observe that, for each position \(i \in S\),
    there are pairs of witness vectors
    $\witnessvec[i]{1}, \witnessvec[i]{0}, \in X$
    that differ (on \(S\)) only at position \(i\),
    with \(\witnessvec[i]{1}\vposition{i} = 1\), that is,
    \begin{itemize}
        \item $\witnessvec[i]{1}\vposition{S \setminus i}
                  = w_{0,i}\vposition{S \setminus i}$,
        \item $\witnessvec[i]{1}\vposition{i} = 1$, and
        \item $\witnessvec[i]{0}\vposition{i} = 0$.
    \end{itemize}
    We write
    \(\CW_{S} \deff \{ \witnessvec[i]{1}, \witnessvec[i]{0} \mid i \in S \}\)
    for a set of witness vectors of \(X\).
\end{fact}

Before we start using these $\sigma$-defining sets,
we first show how to compute them efficiently
along with the witness vectors.
The result is similar to~\cite[Lemma~4.28]{focke_tight_2023_ub}.
However, we improve the running time of the algorithm
by using a more efficient approach to process the data.

\begin{lemma}
    \label{thm:sig-def-set}
    Given a set \(X \subseteq \{0,1\}^n\),
    we can compute a \(\sigma\)-defining set \(S\) for $X$,
    as well as a set of witness vectors \(\CW_{S}\) for \(S\),
    in time \(\Oh( |X| \cdot n^3)\).
\end{lemma}
\begin{proof}
    Our algorithm maintains a set $S \subseteq \numb{n}$
    as the candidate for the $\sigma$-defining set for $X$
    and iteratively checks if a position can be removed from $S$
    or whether it is part of the $\sigma$-defining set.
    Formally, the algorithm is as follows.
    \begin{itemize}
        \item
              Initialize the candidate set as $S \deff \numb{n}$.
        \item
              For all $i$ from $1$ to $n$, repeat the following steps:
              \begin{itemize}
                  \item
                        Initialize an empty map data structure,
                        for example, a map based on binary trees.
                  \item
                        Iterate over all $v \in X$:

                        If there is no entry with key $v\pos{S \setminus \{i\}}$ in the map,
                        we add the value $v$ with the key $v\pos{S \setminus \{i\}}$
                        to the map.

                        If there is already an entry with key $v\pos{S \setminus \{i\}}$
                        and value $v'$ in the map.
                        Then, define the two witness-vectors
                        $\witnessvec[i]{0}$ and $\witnessvec[i]{1}$
                        for position $i$
                        as $v$ and $v'$ (depending on $i$th bit of $v$ and $v'$),
                        and continue with the (outer) for-loop.

                  \item
                        Remove position $i$ from $S$
                        and continue with the next iteration of the for-loop.
              \end{itemize}
        \item
              Output $S$ as the $\sigma$-defining set
              and $\CW_{S} \deff \{ \witnessvec[i]{0}, \witnessvec[i]{1} \mid i \in S\}$
              as the set of witness vectors.
    \end{itemize}
    From the description of the algorithm,
    each operation of the for-loop takes time
    $\Oh(|X| \cdot \log(|X|) \cdot n) \subseteq \Oh(|X| \cdot n^2)$ when utilizing an efficient map data structure.
    It follows that $S$ can be computed in time $\Oh(|X| \cdot n^3)$.

    In order to prove the correctness,
    we denote by $S_i$ the set $S$ after the $i$th iteration of the for-loop
    and denote by $S_0 = \numb{n}$ the initial set.
    Hence, the algorithm returns $S_n$.
    We prove by induction that, for all $i \in \numbZ{n}$,
    if there are $u,v \in X$ with $x\pos{S_i} = y\pos{S_i}$,
    then $u = v$.
    This is clearly true for $i = 0$ by the choice of $S_0$.
    For the induction step, assume that the claim holds
    for an arbitrary but fixed $i-1 \in \numbZ{n-1}$
    where $S_i \neq S_{i-1}$
    as otherwise the claim follows by the induction hypothesis.
    By the design of the algorithm,
    we know that $S_i = S_{i-1} \setminus \{i\}$.
    Moreover, as $i$ was removed from $S_{i-1}$,
    there are no two vectors $u,v \in X$
    such that $u\pos{S_{i-1}\setminus \{i\}} = v\pos{S_{i-1}\setminus\{i\}}$.
    This directly implies that, for all $u,v \in X$,
    the vectors $u$ and $v$ differ on $S_i \setminus\{i\}$
    and thus, are not equal.

    It remains to prove that $S_n$ is minimal.
    Assume otherwise and let $i \in S_n$ be some index
    such that, for all $u,v \in X$, we still have $u = v$
    whenever
    \begin{equation}
        u \pos{S_n \setminus \{i\}} = v\pos{S_n \setminus \{i\}}
        \label{eqn:sigDef:first}
        .
    \end{equation}
    As the algorithm did not remove $i$ from $S_{i-1}$,
    we know that there must be at least two vectors $\hat u, \hat v \in X$
    with $\hat u \neq \hat v$ such that
    \begin{equation}
        \hat u\pos{S_{i-1} \setminus \{i\}} = \hat v \pos{S_{i-1}\setminus\{i\}}
        \label{eqn:sigDef:second}
        .
    \end{equation}
    Since the algorithm satisfies
    $S_0 \supseteq S_1 \supseteq \dots \supseteq S_n$,
    we especially have $S_n \subseteq S_{i-1}$
    and thus,
    whenever \cref{eqn:sigDef:second} holds,
    then also \cref{eqn:sigDef:first} holds
    which implies that $\hat u = \hat v$.
    As this contradicts $\hat u \neq \hat v$,
    we know that $S_n$ is a $\sigma$-defining set.

    To see that the vectors in $\CW_S$ are indeed witness-vectors,
    recall that the algorithm guarantees that
    the distinct vectors $\witnessvec[i]{0}$ and $\witnessvec[i]{1}$ satisfy $\witnessvec[i]{0} \pos{S_{i-1}\setminus\{i\}}
        = \witnessvec[i]{1} \pos{S_{i-1}\setminus\{i\}}$.
    Moreover, the inductive proof above shows that $\witnessvec[i]{0} \pos{S_{i-1}}
        \not = \witnessvec[i]{1} \pos{S_{i-1}}$, and hence,
    $\witnessvec[i]{0} \pos i \not = \witnessvec[i]{1} \pos i$.
    Finally, since $S_{i-1} \subseteq S_n$,
    the vectors agree on $S_n \setminus \{i\}$
    which concludes the proof.
\end{proof}

Now we can concretize how \cref{thm:sparse_language_same_sigma}
restricts the possible weight-vectors of two strings
with the same $\sigma$-vector.
Intuitively, what the $\sigma$-defining set is for the vector $\sigvec{v}$,
this is the set $\compl S$ for the weight-vector.
We could say that $\compl S$ is a \emph{weight-defining} set.
This extends~\cite[Lemma~4.9]{focke_tight_2023_ub}
to the case of residue classes.

\lemSComplDeterminesWeight
\begin{proof}
    We adjust the proof in~\cite{focke_tight_2023_ub} to our setting.
    Let $L$ be a sparse language
    with $S$ as a $\sigma$-defining set for $\sigvec{L}$.
    Furthermore, consider two state-vectors $x$ and $y$ of $L$
    with the same $\sigma$-vector, i.e., $\sigvec{x} = \sigvec{y}$,
    such that $\degvec{x}\vposition{\compl{S}} = \degvec{y}\vposition{\compl{S}}$.

    Fix a position $\ell \in S$.
    We have binary witness vectors $\witnessvec{1}$ and $\witnessvec{0}$,
    that agree on all positions of $S$ except for position $\ell$
    and satisfy $\witnessvec{1}\pos{\ell}= 1$ and $\witnessvec{0}\pos{\ell}= 0$.

    Using \cref{thm:sparse_language_same_sigma} twice we obtain
    \begin{align*}
        (\degvec{x} - \degvec{y}) \cdot (\witnessvec{1}- \witnessvec{0})\equiv_\mname 0.
    \end{align*}
    By the assumption about $x$ and $y$,
    their weight-vectors agree on all positions in $\compl S$.
    But since $\witnessvec{1}$ and $\witnessvec{0}$
    are identical on all positions of $S$
    except for position $\ell$,
    this means that
    \begin{align*}
        0 \equiv_\mname~
                       & (\degvec{x} - \degvec{y}) \cdot (\witnessvec{1} - \witnessvec{0}) \\
        \equiv_\mname~ & \degvec{x}\pos{\ell} - \degvec{y}\pos{\ell}
        .
    \end{align*}
    Since all values of the weight-vector are less than $\mname$,
    the consequence $\degvec{x}\pos{\ell} = \degvec{y}\pos{\ell}$
    concludes the proof.
\end{proof}

The previous result from \cref{thm:s_complement_determines_weight_vector}
gives rise to a strategy for counting all the strings of a sparse language.
Namely, we can fix a $\sigma$-vector
and count the number of strings with this fixed vector,
which is the same as the number of weight-vectors of strings
with this fixed $\sigma$-vector.
When counting the weight-vectors,
we can now use the property we have just seen
to show that not every weight-vector can actually occur.

\thmUpperBoundOnLanguage
\begin{proof}
    Let $L \subseteq \mathbb{A}^n$ be a sparse language.
    Compute a $\sigma$-defining set $S$ for $\sigvec{L}$.

    The size of $\sigvec{L}$ is bounded by $2^{|S|}$ per definition.
    Fix an arbitrary $\sigma$-vector $\vec s$ of $\sigvec{L}$.
    We now count the number of possible weight-vectors with $\sigma$-vector $\vec s$.
    By \cref{thm:s_complement_determines_weight_vector}, the positions $\compl S$ of the weight-vector uniquely determine the positions $S$ of the weight-vector.
    Hence, we obtain a bound on the size of the language of
    \[
        |L| \leq 2^{|S|} \cdot \mname^{|\compl S|} \leq \mname^n
        .
        \qedhere
    \]
\end{proof}

\subsection{Compressing Weight-Vectors}
\label{sec:upperTw:compression}

In the previous section we have seen that the size of a language
cannot be too large.
However, even though this already provides a bound
on the size of these languages,
the space where these languages come from is still large.
That is, even when considering a sparse language $L \subseteq \allStates^n$,
we would have to consider all states in $\allStates^n$
when doing convolutions as we do not know
which vectors might actually appear in the solution.
The goal of this section is to provide the concept of \emph{compression}
where we do not only compress the vectors of a language
but actually compress the space where these vectors are from.

In previous sections we have seen that certain positions of the weight-vectors
actually uniquely determine the remaining positions.
Hence, a natural idea for the compression of strings
is to store the values at important positions only,
while ensuring that the values at the other positions can be reconstructed.

We first introduce the notion of the \emph{remainder}.
Intuitively, the remainder can be used to reconstruct the values of the weight-vector
which have been compressed.

\begin{definition}[{Compare~\cite[Definition~4.21]{focke_tight_2023_ub}}]
    \label{def:remainder_vector}
    Let $\sigma$ and $\rho$ denote two residue classes
    modulo $\mname \ge 2$.
    Let $L \subseteq \allStates^n$ denote a (non-empty) sparse language.
    Consider a set $X$
    with $X \subseteq \sigvec{L}$ and $X \not = \emptyset$,%
    \footnote{Think of $X = \sigvec{L}$ for now.}
    and let \(S\) denote a \(\sigma\)-defining set for $X$
    with a corresponding set of witness vectors
    \(\CW_{S} \subseteq \sigvec{L}\).
    For two vectors \(u, o \in \fragment0{\mname-1}^n\)
    and a position \(\ell\in S\),
    we define the \emph{remainder}
    \(\remvec[\CW_{S}]{u}{o}\) at position \(\ell\) as
    \[
        \remvec[\CW_{S}]{u}{o}\vposition{\ell}
        \coloneqq \sum_{i \in \compl{S}} \big(u\vposition{i} -
        o\vposition{i}\big)\cdot\big(w_{1,\ell}\vposition{i} -
        w_{0,\ell}\vposition{i}\big).
        \qedhere
    \]
\end{definition}

As a next step we show that
this remainder is chosen precisely in such a way
that we can easily reconstruct the values of the weight-vector on positions $S$.

\begin{lemma}[{Compare~\cite[Remark~4.22]{focke_tight_2023_ub}}]
    \label{thm:remainder_shift}
    Let $\sigma$ and $\rho$ denote two residue classes
    modulo $\mname$.
    Let $L \subseteq \allStates^n$ denote a (non-empty) sparse language.
    Consider a set $X$
    with $X \subseteq \sigvec{L}$ and $X \not = \emptyset$,
    and let \(S\) denote a \(\sigma\)-defining set for $X$
    with a corresponding set of witness vectors
    \(\CW_{S} \subseteq \sigvec{L}\).
    Consider arbitrary weight-vectors $u$ and $o$
    of two strings from $L$ with a common $\sigma$-vector.
    Then,
    \begin{align*}
        u\pos{\ell} \equiv_\mname o\pos{\ell} - \remvec[\mathcal{W}_S]{u}{o}\pos{\ell}
    \end{align*}
    holds for all $\ell \in S$.
\end{lemma}
\begin{proof}
    We follow the sketch provided in~\cite{focke_tight_2023_ub}.
    Let $u$ and $o$ be two weight-vectors of two strings of $L$
    that share a common $\sigma$-vector.
    By the definition of the $\sigma$-defining set,
    there are, for all $\ell \in S$,
    two associated witness vectors
    $\witnessvec{0}, \witnessvec{1} \in \sigvec{L}$.
    Hence, there exist two strings $x, y \in L$ such that
    $\sigvec{x} = \witnessvec{0}$ and $\sigvec{y} = \witnessvec{1}$.
    From \cref{thm:sparse_language_same_sigma} we obtain
    \[
        (u - o) \cdot w_{1,\ell} \equiv_\mname 0
        \qquad\text{and similarly}\qquad
        (u - o) \cdot w_{0,\ell} \equiv_ \mname 0.
    \]
    Recall that $\witnessvec{0}$ and $\witnessvec{1}$
    agree on all positions of $S$ except for position $\ell$,
    which implies that $\witnessvec{1}\pos i - \witnessvec{0}\pos i = 0$
    for all $i \in S \setminus\{\ell\}$.
    Combining these observations we get
    \begin{align*}
        u\pos{\ell} - o\pos{\ell} + \remvec[\mathcal{W}_S]{u}{o}\pos{\ell}
         & = u\pos{\ell} - o\pos{\ell}
        + \sum_{i \in \compl S} (u\pos{i} - o\pos{i})
        \cdot (\witnessvec{1}\pos{i} - \witnessvec{0}\pos{i})    \\
         & = (u - o) \cdot (w_{1,\ell} - w_{0,\ell})             \\
         & = (u - o) \cdot w_{1,\ell} - (u - o) \cdot w_{0,\ell} \\
         & \equiv_\mname 0 - 0 = 0
    \end{align*}
    which yields the claim after rearranging.
\end{proof}

Once we have fixed a $\sigma$-defining set $S$,
\cref{thm:remainder_shift} provides a recipe
for compressing the weight-vectors:
Positions of $\compl S$ are kept as they are, and positions on $S$ are completely omitted.
Essentially the compression is just a projection of the weight-vector
to the positions from $\compl S$.
However, in order to be consistent
with~\cite[Definition~4.23]{focke_tight_2023_ub},
we refer to it as compression.
When one requires access to the omitted positions, they can be easily reconstructed using the remainder vector.
Now we have everything ready to formally define the compression of a weight-vector.

\begin{definition}[Compression of Weight-vectors; {compare~\cite[Definition~4.23]{focke_tight_2023_ub}}]
    \label{def:compressed_vector}
    Let $\sigma$ and $\rho$ denote two residue classes
    modulo $\mname \geq 2$.
    Let $L \subseteq \allStates^n$ denote a (non-empty) sparse language.
    Consider a set $X$
    with $X \subseteq \sigvec{L}$ and $X \not = \emptyset$,
    and let \(S\) denote a \(\sigma\)-defining set for $X$.

    For a weight-vector
    \(u \in \{ \degvec{x} \mid x \in L\}\),
    we define the
    \emph{\(\sigma\)-compression}
    as the following \(|\compl S|\)-dimensional vector \(\comp{}{u}\)
    where we set
    \begin{alignat*}{5}
        \comp{}u\vposition{\ell}
         & \deff {u}\vposition{\ell}                   &
         & \mod \mname,                                &
         & \qquad \text{for all } \ell \in  \compl{S}.
    \end{alignat*}
    Further, we write \(\CZ_{S}\) for the \(|\compl S|\)-dimensional space
    of all possible vectors
    where for each dimension the entries are computed modulo $\mname$.
\end{definition}

From the definition of the compression it already follows
that there cannot be too many compressed vectors.
This is quantized by the following simple observation.

\begin{fact}[{Compare~\cite[Remark~4.24]{focke_tight_2023_ub}}]
    \label{thm:compressio_space_small}
    With the same definitions as in \cref{def:compressed_vector},
    we have $|\CZ_{S}|\leq \mname^{|\compl S|}.$
\end{fact}

\subsection{Faster Join Operations}
\label{sec:upperTW:join}
From the previous result in~\cref{thm:upper:boundOnLanguage},
which bounds the size of a realized language,
we know that in our dynamic programming algorithm
the number of solutions tracked at each node is relatively small
if the language of the node is sparse.
With some additional bookkeeping, we can easily ensure that the languages
we keep track of in the dynamic programming algorithm are all sparse.
However, as the join operation
is usually the most expensive operation,
we must also be able to combine two different languages efficiently.

In order to be able to combine two languages formally,
we define the combination of two languages.
In contrast to the underlying definition from~\cite{focke_tight_2023_ub},
we must take into account (but can also exploit)
that we are now dealing with residue classes.
We start by defining the combination of two state-vectors.
Such a combination is, of course,
the central task of the join-operation in the dynamic programming algorithm.
Intuitively, it is clear that two partial solutions
for the same set of portal vertices can only be joined
when their $\sigma$-vectors agree.
In contrast to finite sets,
for residue classes the exact number of selected neighbors
of each vertex is not relevant, only the remainder when dividing by $\mname$ matters.
We capture exactly these intuitive properties in our following definition.

\begin{definition}[Combination of Strings and Languages; {compare~\cite[Definition~4.14]{focke_tight_2023_ub}}]
    \label{def:string_combination}
    Let $\sigma$ and $\rho$ denote two residue classes
    modulo $\mname \ge 2$.
    For two string $x,y \in \allStates^n$,
    the combination $x \oplus y \in (\mathbb{A} \cup \{\bot\})^n$ of $x$ and $y$
    is defined as
    \begin{equation*}
        (x \oplus y) \pos{\ell} \deff \begin{cases}
            \sigma_{(a+b) \bmod \mname} &
            \text{if } x\pos{\ell} = \sigma_a
            \land y\pos{\ell} = \sigma_b,
            \\
            \rho_{(a+b) \bmod \mname}   &
            \text{if } x\pos{\ell} = \rho_a
            \land y\pos{\ell} = \rho_b,
            \\
            \bot                        &
            \text{ otherwise},
        \end{cases}
    \end{equation*}
    for all $\ell \in \numb{n}$.
    When $x \oplus y \in \allStates^n$, we say that $x$ and $y$ can be combined.

    We extend this combination of two strings in the natural way to sets,
    that is, to languages.
    Consider two languages $L_1, L_2 \subseteq \allStates^n$.
    We define $L_1 \oplus L_2 \deff \{ x \oplus y \mid x \in L_1, y \in L_2\} \cap \allStates^n$
    as the combination of $L_1$ and $L_2$.
\end{definition}

As we can only combine two strings with the same $\sigma$-vector,
we have
$\sigvec{L_1 \oplus L_2} \subseteq \sigvec{L_1} \cap \sigvec{L_2}$.
Our next goal is to efficiently compute the combination of two sparse languages.

\begin{restatable*}[{Compare~\cite[Theorem~4.18]{focke_tight_2023_ub}}]{lemma}{quickcombinations}
    \label{thm:quick_combinations}
    Let $\sigma$ and $\rho$ denote two residue classes
    modulo $\mname \geq 2$.
    Let $L_1, L_2 \subseteq \allStates^n$ be two sparse languages
    such that $L_1 \oplus L_2$ is also sparse.
    Then, we can compute $L_1 \oplus L_2$ in time
    \(
    \mname^n \cdot n^{\Oh{(1)}}
    .
    \)
\end{restatable*}

The basic idea to achieve this bound
is to use a fast convolution technique
introduced by van Rooij~\cite{Rooij_FastJoinOperations}
for the general algorithm for \srDomSet when the sets are finite or cofinite.

We intend to compress two vectors, combine their compressions, and decompress them afterwards.
To ensure the safety of this operation, we need an auxiliary result that allows us to decompress these vectors.
Indeed, the compression is designed such that this reversal can be done.

\begin{definition}[Decompression]
    \label{def:decompressing}
    Let $\sigma$ and $\rho$ denote two residue classes
    modulo $\mname \geq 2$.
    Let $L\subseteq \allStates^n$ denote a (non-empty) sparse language,
    and let \(X \subseteq \sigvec{L}\) be a set
    with a corresponding set of witness vectors
    \(\CW_{S} \subseteq X\).

    Further, fix a vector \(\vec{s} \in \sigvec{L}\) and an origin vector
    \(o \in \numbZ{\mname-1}^n\) such that
    there exists a $y \in L$
    with $\degvec{y} = o$ and $\sigvec{y} = \vec s$.

    Then, for any vector $a \in \CZ_{S}$,
    we define the decompression $\decomp[o]{a}$ of $a$
    relative to $o$ as the following vector of length $n$:
    \[
        \decomp[o]{a}\pos \ell \deff
        \begin{cases}
            a \pos \ell
             & \text{ if } \ell \in \compl S, \\
            o \pos \ell
            - \remvec[\mathcal{W}_S]{a}{o}\pos{\ell} \mod \mname
             & \text{ if } \ell \in S.
        \end{cases}
        \qedhere
    \]
\end{definition}

With this definition we can move to the main result
leading to the correctness of the convolution algorithm.

\begin{lemma}
    \label{thm:decompressing}
    Let $\sigma$ and $\rho$ denote two residue classes
    modulo $\mname \geq 2$.
    Let $L\subseteq \allStates^n$ denote a (non-empty) sparse language.
    And let \(X \subseteq \sigvec{L}\) be a set
    with a corresponding set of witness vectors
    \(\CW_{S} \subseteq X\).

    Further, fix a vector \(\vec{s} \in \sigvec{L}\) and an origin vector
    \(o \in \numbZ{\mname-1}^n\) such that
    there exists a $y \in L$
    with $\degvec{y} = o$ and $\sigvec{y} = \vec s$.

    Then, for any vector $x \in L$ with $\sigvec{x} = \vec s$
    it holds
    that
    $
        \decomp[o]{(\comp{}{\degvec{x}})} = {\degvec{x}}
        .
    $
\end{lemma}
\begin{proof}
    To show that $\decomp[o]{(\comp{}{\degvec{x}})} = \degvec{x}$
    we consider two cases depending on the position of the vectors.

    For a position $\ell \in \compl S$,
    equality holds because the compression from \cref{def:compressed_vector}
    does not drop an entry on such a position
    and the decompression from \cref{def:decompressing} does not change it either.

    For a position $\ell \in S$, first observe that $\remvec[\mathcal{W}_S]{\comp{}{\degvec{x}}}{o}\pos{\ell} = \remvec[\mathcal{W}_S]{\degvec{x}}{o}\pos{\ell}$ as the remainder only depends on the positions from~$\compl S$ by \cref{def:remainder_vector} and by the definition of the compression.
    Then, the statement directly follows from \cref{thm:remainder_shift} and the definition of the decompression.
\end{proof}

Hence, when we are given a set of compressed weight-vector of some sparse language (all with the same $\sigma$-vector),
we can decompress these vectors
as long as we store \emph{a single} weight-vector (of a string with the same $\sigma$-vector) that is not compressed.
Therefore,
as we can easily decompress the compressed weight vectors,
we can work with the compressed vectors without any issues.

As a last ingredient we formally introduce the fast convolution technique
originally formalized by van Rooij~\cite{Rooij_FastJoinOperations}
and later restated in~\cite[Theorem~4.19]{focke_tight_2023_ub}.

\begin{fact}[{\cite[Theorem~4.19]{focke_tight_2023_ub}}]
    \label{thm:alg-convolution}
    For integers $d_1,\dots,d_n$ and
    $D \deff \prod_{i=1}^nd_i$,
    let~$p$ denote a prime such that in the field \(\FF_p\),
    a $d_i$-th root of unity%
    \footnote{As we always deal with \emph{primitive} roots of unity only,
        we just write \emph{roots of unity} to simplify notation henceforth.
    }
    exists
    for each $i \in \numb{n}$.
    Further, for two functions
    $f,g\from \ZZ_{d_1} \times \dots \times \ZZ_{d_n} \to \FF_p$,
    let $h\from \ZZ_{d_1} \times \dots \times \ZZ_{d_n} \to \FF_p$ denote
    the convolution
    \[h(a) \deff \sum_{a_1 + a_2 = a} f(a_1) \cdot g(a_2).\]
    Then, we can compute the function $h$ in $\Oh(D \log D)$ arithmetic operations
    (assuming a $d_i$-th root of unity $\omega_i$ is given for all $i \in \nset{n}$).
\end{fact}

Since the application of this technique relies on the existence
of certain roots of unity, we must be able to compute them efficiently.
We use \cite[Remark~4.20]{focke_tight_2023_ub} to find these roots
in our setting.

\begin{remark}[{\cite[Remark~4.20]{focke_tight_2023_ub}}]
    \label{rem:find-prime-alg}
    Suppose $M$ is a sufficiently large integer such that all images of the functions $f,g,h$ are in the range $\numbZ{M}$.
    In particular, suppose that $M \geq D$.
    Suppose $d_1',\dots,d_\ell'$ is the list of integers obtained from $d_1,\dots,d_n$ by
    removing duplicates.
    Let $D' \deff \prod_{i=1}^\ell d_i'$.
    We consider candidate numbers $m_j \coloneqq 1 + D'j$ for all $j \geq 1$.
    By the Prime Number Theorem for Arithmetic Progressions \cite[Theorem 1.3]{BennettMOR18}, there is a prime $p$ such that
    \begin{enumerate}
        \item $p = m_j$ for some $j \geq 1$,
        \item $p > M$, and
        \item $p = \Oh\Big(\max\big\{\varphi(D')M \log M,\exp(D')\big\}\Big)$,
    \end{enumerate}
    where $\varphi$ denotes Euler's totient function.
    Such a number can be found in time
    \[\Oh\Big(p \big(\log p\big)^c\Big)\]
    for some absolute constant $c$ exploiting that prime testing can be done in polynomial time.

    Now, fix $i \in \nset{n}$ and fix $k_i \deff {D' \cdot j}/{d_i}$.
    For every $x \in \FF_p^*$,
    we have that $x^{p-1} = 1$, and hence,
    $x^{k_i}$ is a $d_i$-th root of unity
    if and only $(x^{k_i})^i \neq 1$ for all $i < d_i$.
    Hence, given an element $x \in \FF_p^*$, it can be checked in time
    \[\Oh\Big( d_i \cdot (\log p)^c\Big)\]
    whether $x^{k_i}$ is a $d_i$-th root of unity.
    Due to our choice of $p$, this test succeeds for at least one $x \in \FF_p^*$.
    Thus, a $d_i$-th root of unity $\omega_i$ for every $i \in \nset{n}$ can be found in time
    \[
        \Oh\Big(
        n \cdot p \cdot \max_{i \in \nset{n}}d_i \cdot (\log p)^c
        \Big).
        \qedhere
    \]
\end{remark}

We are now ready to proof the main result of this section,
that is, to prove that the combination of two sparse languages
can be computed efficiently.
We follow the ideas in~\cite{focke_tight_2023_ub}
but for the case of residue classes
and state the procedure more explicitly.
Formally, we prove \cref{thm:quick_combinations},
which we restate here for convenience.

\quickcombinations
We first sketch the rough idea of the algorithm.
The idea is to treat each possible $\sigma$-vector
of the combined language separately.
Then, for each such vector $\vec s$,
we can use the previously introduced machinery to
(1) compress the corresponding weight-vectors,
(2) use the fast convolution techniques
to efficiently combine these compressed vectors,
and (3) decompress these combined vectors.
As a last step we then combine all these partial results
to obtain the combined language.

\begin{proof}
    Let $L_1, L_2 \subseteq \allStates^n$ be sparse
    such that $L_1 \oplus L_2$ is also sparse.
    We start by computing $\sigvec{L_1} \cap \sigvec{L_2}$.
    If this intersection is empty, the combined language is clearly also empty.
    Otherwise, we use \cref{thm:sig-def-set}
    to compute a $\sigma$-defining set $S$ for  $\sigvec{L_1} \cap \sigvec{L_2}$
    along with a corresponding set of witness vectors $\CW_S$.
    Next, we repeat the following three steps for each
    $\vec s \in \sigvec{L_1} \cap \sigvec{L_2}$:

    \begin{enumerate}
        \item
              Compute the two sets
              \begin{align*}
                  L_1^{\vec s} \deff \{x_1 \in L_1 \mid \sigvec{x_1} = \vec s\}
                  \quad \text{and} \quad
                  L_2^{\vec s} \deff \{x_2 \in L_2 \mid \sigvec{x_2} = \vec s\}
              \end{align*}
              by iterating over each string of $L_1$ and $L_2$.
              Use these sets to compute the corresponding sets of weight-vectors as
              \begin{align*}
                  \degvec{L_1^{\vec s}} \deff\{\degvec{x_1}\mid x_1 \in L_1^{\vec s}\}
                  \quad \text{and} \quad
                  \degvec{L_2^{\vec s}} \deff\{\degvec{x_2}\mid x_2 \in L_2^{\vec s}\}
                  .
              \end{align*}

        \item
              Define two functions
              $f,g \from \CZ_{S} \to \ZZ$
              where, for all $u_1 \in \degvec{L_1^{\vec s}}$
              and all $u_2 \in \degvec{L_2^{\vec s}}$,
              we set
              \begin{align*}
                  f(\comp{}{u_1}) \deff 1
                  \quad\text{and}\quad
                  g(\comp{}{u_2}) \deff 1
              \end{align*}
              and whenever there is no
              $u_1 \in \degvec{L_1^{\vec s}}$ with $\comp{}{u_1} = a$
              we set $f(a)\deff 0$,
              and likewise for $g$.

        \item
              Now, we can use the fast convolution technique
              of \cref{thm:alg-convolution} to compute
              the function $h \from \CZ_{S}\to \ZZ$ defined as
              \begin{align*}
                  h(a) \deff \sum_{a_1 + a_2 = a} f(a_1) \cdot g(a_2)
                  \qquad\text{for all } a \in \CZ_{S}
                  .
              \end{align*}
        \item
              Pick two arbitrary origin vectors $o_1 \in \degvec{L_1^{\vec s}}$
              and $o_2 \in \degvec{L_2^{\vec s}}$.
              We iterate over all entries $a \in \CZ_{S}$
              and check whether $h(a) > 0$.
              If not, then we discard the vector.
              Otherwise, we use \cref{thm:decompressing},
              to decompress the vector.
              Formally, we compute
              \begin{align*}
                  W_{1,2}^{\vec s} \deff \{\decomp{a} \mid a \in \CZ_S \land
                  h(a) > 0\}.
              \end{align*}
              Then, it is easy to compute the language
              \begin{align*}
                  L_{1,2}^{\vec{s}}  \deff \{ z \in \allStates^n \mid
                  \sigvec{z} = \vec{s}
                  \land \degvec{z} \in W_{1,2}^{\vec{s}} \}
              \end{align*}
              by combining each weight-vector of $W_{1,2}^{\vec s}$
              with the $\sigma$-vector $\vec{s}$.
    \end{enumerate}
    Finally, we obtain $L_{1,2}$ by simply taking the union
    of the ``partial languages'' over all $\sigma$-vectors, that is,
    \begin{align*}
        L_{1,2} \deff \bigcup_{\vec s \in \sigvec{L_1} \cap \sigvec{L_2}} L_{1,2}^{\vec{s}}
        .
    \end{align*}
    The algorithm outputs $L_{1,2}$ as the set of the combined languages.

    We now argue that this algorithm correctly computes $L_1 \oplus L_2$.
    \begin{claim}
        \label{clm:upper:quickcombinations:correctness}
        The algorithm is correct, that is, $L_{1,2} = L_1 \oplus L_2$.
    \end{claim}
    \begin{claimproof}
        For ease of notation,
        we write $L = L_1 \oplus L_2$ in the following.
        It is clear that
        $\sigvec{L} \subseteq \sigvec{L_1} \cap \sigvec{L_2}$
        by the way the combination of strings is defined.
        Hence, we directly get
        \begin{align*}
            L_{1} \oplus L_2
            = \bigcup_{\vec s \in \sigvec{L_1} \cap \sigvec{L_2}} L^{\vec{s}}
        \end{align*}
        where ${L^{\vec s}} \deff \{z \in L \mid \sigvec{z} = \vec s\}$.
        Therefore, it suffices to show that $L^{\vec s} = L_{1,2}^{\vec s}$.
        Since all string in both these sets share the same $\sigma$-vector,
        it is sufficient to prove that
        $\degvec{L^{\vec s}} = W_{1,2}^{\vec s}$
        which is equivalent to proving that
        \begin{align}
            \{\decomp{a} \mid a \in \CZ_S \land
            h(a) > 0\}
            = \{\degvec{z} \mid z \in L \land \sigvec{z} = \vec s\}
            \label{eqn:upper:correctnessOfJoinProcedure}
            .
        \end{align}
        Also observe that $o_1$ is the weight-vector of some string $v_1$ of $L_1$ with $\sigvec{v_1} = \vec s$, and $o_2$ is the weight-vector of some string $v_2$ of $L_2$ with $\sigvec{v_2} = \vec s$.
        Thus, $v_1$ and $v_2$ can be combined, $v_1 \oplus v_2 \in L$, and $o_1 + o_2 = \degvec{v_1 \oplus v_2}$.
        This allows us to use \cref{thm:decompressing} with $o_1 + o_2$ as the origin vector.

        We now prove \cref{eqn:upper:correctnessOfJoinProcedure} by considering both directions separately.

        \subparagraph*{\boldmath Case ``$\subseteq$''
            in \Cref{eqn:upper:correctnessOfJoinProcedure}.}
        Intuitively, this direction proves that the algorithm
        does not compute incorrect results.
        Consider an arbitrary vector
        $u \in \{\decomp{a} \mid a \in \CZ_S \land
            h(a) > 0\}$.
        Then, there exists $a \in \CZ_S$ such that $h(a) > 0$ and $u = \decomp{a}$.
        From $h(a) > 0$ we know there exist $a_1$ and $a_2$
        with $f(a_1) = 1$ and $g(a_2) = 1$ such that $a_1 + a_2 = a$.
        By the definition of $f$ and $g$,
        this means that there is some $x_1 \in L_1$
        with $a_1 = \comp{}{\degvec{x_1}}$
        and there is some $x_2 \in L_2$
        with $a_2 = \comp{}{\degvec{x_2}}$
        such that $\sigvec{x_1} = \sigvec{x_2} = \vec s$.
        Hence, $x_1$ and $x_2$ can be combined, and we have
        $a = \comp{}{\degvec {x_1}} + \comp{}{\degvec {x_2}} = \comp{}{(\degvec {x_1} \oplus \degvec {x_2})}$.
        By \cref{thm:decompressing}, we now know that $u \in \{\degvec{z} \mid z \in L \land \sigvec{z} = \vec s\}$.

        \subparagraph*{\boldmath Case ``$\supseteq$''
            in \Cref{eqn:upper:correctnessOfJoinProcedure}.}
        Intuitively, this direction proves that the algorithm
        computes all strings of the combined language.
        Now, consider an arbitrary $z \in L_1 \oplus L_2$
        with $\sigvec{z} = \vec s$.
        This means that there are $x_1 \in L_1$ and $x_2 \in L_2$
        such that $x_1 \oplus x_2 = z$.
        When setting $a_1 \deff \comp{}{\degvec{x_1}}$,
        setting $a_2 \deff \comp{}{\degvec{x_2}}$,
        and setting $a \deff a_1 + a_2$,
        we trivially have $a \in \CZ_{S}$.
        Moreover, we know that $f(a_1) = g(a_2) = 1$
        and thus, $h(a) > 0$
        as the convolution procedure from \cref{thm:alg-convolution}
        correctly computes the function $h$.
        Hence, $\decomp{a} = \decomp{\big(\comp{}{\degvec{x_1}} + \comp{}{\degvec{x_2}} \big)}$ is in the set $\{\decomp{a} \mid a \in \CZ_S \land
            h(a) > 0\}$.
        Moreover, $\comp{}{\degvec{z}} = \comp{}{\degvec{x_1}} + \comp{}{\degvec{x_2}}$, and hence $\decomp{\big(\comp{}{\degvec{x_1}} + \comp{}{\degvec{x_2}} \big)} = \decomp{(\comp{}{\degvec{z}})} = \degvec{z}$ by \cref{thm:decompressing}.

        This concludes the proof of the correctness of the combine operation.
    \end{claimproof}

    Finally, we argue that the described algorithm has the claimed runtime.
    \begin{claim}
        \label{clm:upper:quickcombinations:runtime}
        The running time of the above algorithm is
        $\mname^n \cdot n^{\Oh(1)}$.
    \end{claim}
    \begin{claimproof}
        We can compute the sets $\sigvec{L}$ for any $L \subseteq \allStates^n$
        in time
        \begin{align*}
            \Oh\left(|L| \cdot \log(|\sigvec{L}|) \cdot n\right) \subseteq \Oh\left(|L| \cdot n^2\right)
        \end{align*}
        by iterating over each string of the languages and by utilizing a suitable set data structure. Hence, we can compute $\sigvec{L_1} \cap \sigvec{L_2}$ in time
        \(
        \Oh\left( \max(|L_1|,|L_2|) \cdot n^2 \right)
        \)
        by first computing $\sigvec{L_1}$ and $\sigvec{L_2}$ separately,
        and then applying a standard algorithm for the intersection of sets.

        By \cref{thm:sig-def-set},
        the computation of the $\sigma$-defining set $S$
        for $\sigvec{L_1} \cap \sigvec{L_2}$ takes time
        \begin{align*}
            \Oh\left( \left |\sigvec{L_1} \cap \sigvec{L_2} \right | \cdot n^3\right)
            .
        \end{align*}

        The computation of the four sets $L_1^{\vec s}$, $L_2^{\vec s}$,
        $\degvec{L_1^{\vec s}}$, and $\degvec{L_2^{\vec s}}$
        takes time $\Oh(\max(|L_1|,|L_2|) \cdot n^2)$
        by iterating over each of $L_1$ and $L_2$ once,
        and assigning the vectors to the correct sets (to efficiently find the appropriate set, one can, for example, use balanced search trees).
        This can be done even before the vector $\vec s$ is fixed.

        Next, the algorithm fixes an arbitrary
        $\vec s \in \sigvec{L_1} \cap \sigvec{L_2}$.
        By iterating over all elements of $\degvec{L_1^{\vec s}}$
        and $\degvec{L_2^{\vec s}}$, and computing their compressions,
        we can compute $f$ and $g$ in time
        \begin{align*}
            |\CZ_{S}| \cdot n^{\Oh(1)}
            + \max\Bigl(\bigl| \degvec{L_1^{\vec s}} \bigr|,
            \bigl| \degvec{L_2^{\vec s}} \bigr|\Bigr)
            \cdot n^{\Oh(1)}
            = |\CZ_{S}| \cdot n ^{\Oh(1)}
        \end{align*}
        when considering that $|\degvec{L_1^{\vec s}}| \leq |\CZ_{S}|$
        and $|\degvec{L_2^{\vec s}}| \leq |\CZ_{S}|$.

        For the application of \cref{thm:alg-convolution}
        to compute the function $h$,
        we require that certain roots of unity and a suitable prime $p$ exist.
        We use \cref{rem:find-prime-alg} for this
        and observe that all our entries of the compressed weight-vector
        are reduced modulo $\mname$.
        Therefore,
        the number of moduli over which we do the computations
        (these are the variables $d_1', \dots, d_\ell'$
        from \cref{rem:find-prime-alg})
        is constant and actually just $1$ and thus, $D' = \mname$.
        Now, using \cref{rem:find-prime-alg}, this prime $p$
        and the root of unity (as we only have one single module)
        can be computed in time
        \begin{align*}
            p \cdot (\log p)^{\Oh(1)}
            .
        \end{align*}
        Hence, when setting $M \deff |\CZ_{S}|$, we see that
        \begin{align*}
            p
            = \Oh \left(\max(\varphi(\mname) M \log M, \exp(\mname)) \right)
            = \Oh (|\CZ_{S}| \cdot \log (|\CZ_{S}|)),
        \end{align*}
        since $\mname$ is a constant.
        Thus, by \cref{thm:alg-convolution},
        we can compute $p$, the root of unity, and $h$ in time
        \begin{align*}
            |\CZ_{S}| \cdot (\log |\CZ_{S}|)^{\Oh(1)} \cdot n^{\Oh(1)}
            .
        \end{align*}

        The computation of $W_{1,2}^{\vec s}$ can then be done in time
        \(
        |\CZ_{S}| \cdot n^{\Oh(1)}
        \)
        by iterating over all elements of $\CZ_{S}$
        and decompressing them whenever possible
        by \cref{thm:decompressing}.
        We conclude that for each $\vec s$,
        the running time of the algorithm is bounded by
        \begin{align*}
            |\CZ_{S}| \cdot \log(|\CZ_{S}|)^{\Oh (1)}
            \cdot n^{\Oh (1)}
            .
        \end{align*}
        Computing $L_{1,2}^{\vec s}$ then takes time $\Oh(|L_{1,2}^{\vec s}| \cdot n)$
        and the final computation of $L_{1,2}$ takes $\Oh(|L_{1,2}|)$ time.

        Hence, the algorithm needs time
        \begin{align*}
            \max(|L_1|,|L_2|,|L_1 \oplus L_2|) \cdot n^3
            + \sum_{\vec s \in \sigvec{L_1} \cap \sigvec{L_2}}
            |\CZ_{S}|
            \cdot \log(|\CZ_{S}|)^{\Oh (1)}
            \cdot n^{\Oh (1)}
        \end{align*}
        time in total.
        Now, we use \cref{thm:compressio_space_small} and obtain
        \begin{align*}
            \log(|\CZ_{S}|)
            \leq \log(\mname^n)
            = n^{\Oh(1)}.
        \end{align*}
        Using \cref{thm:compressio_space_small} once more
        and plugging in the running time of the algorithm,
        we obtain the final bound of
        \begin{align*}
             & \left(
            \max(|L_1|,|L_2|,|L_1 \oplus L_2|)
            + \sum_{\vec s \in \sigvec{L_1} \cap \sigvec{L_2}}
            |\CZ_{S}|
            \cdot n^{\Oh (1)}
            \right) \cdot n^{\Oh(1)} \\
             & \leq \left(
            \max(|L_1|,|L_2|,|L_1 \oplus L_2|)
            + 2^{|S|} \cdot \mname^{|\compl S|} \right)
            \cdot n^{\Oh(1)}.
        \end{align*}
        Finally, we apply~\cref{thm:upper:boundOnLanguage}
        to bound the size of $L_1$, $L_2$, and $L_1\oplus L_2$ as they are sparse,
        to obtain the bound of
        \[
            \mname^n \cdot n^{\Oh (1)}
        \]
        for the running time of the algorithm.
    \end{claimproof}

    This concludes the proof of \cref{thm:quick_combinations}
    by combining \cref{clm:upper:quickcombinations:correctness,%
        clm:upper:quickcombinations:runtime}.
\end{proof}

\subsection{Dynamic Programming Algorithm}
\label{sec:upperTW:dp}
With the results from the previous section,
we now have everything ready
to state the fast algorithm for \srDomSet
when the sets are residue classes modulo $\mname$.
We follow the standard dynamic programming approach on tree decompositions,
that was also used in~\cite{focke_tight_2023_ub} in a similar manner.
With this result we finally prove our algorithmic main theorem,
which we restate here for convenience.

\setcounter{mtheorem}{0}
\thmTwUpperBound
\setcounter{mtheorem}{2}
\begin{proof}
    Let $G$ be an instance of \srDomSet with $n = \abs{V(G)}$ vertices
    that is given along with a tree decomposition of width $\tw$.
    As we can transform this decomposition into a nice tree decomposition
    in time $n^{\Oh(1)}$,
    we assume without loss of generality
    that the provided decomposition is already nice
    and that the bags of the root node and the leaf nodes are empty.
    See, for example, \cite[Section~7.2]{cyganParameterizedAlgorithms}
    for further details about this standard procedure.

    Recall that
    for a node $t \in T$, we denote by $X_t$ the associated bag
    and define $V_t$ to be the set of all vertices
    that have been introduced in the subtree rooted at $t$ of $T$ (including $t$).

    To goal of the algorithm is to compute,
    for each node $t$ of the nice tree decomposition
    and each index $i \in \numbZ{n}$,
    a language $L\pos{t,i}$
    which contains all strings $x \in \allStates^{X_t}$
    such that there exists a partial solution $S_x$
    for the graph with portals $(G \pos{V_t}, X_t)$
    with $|S_x \setminus X_t| = i$.
    Using \cref{thm:structural_property},
    one can immediately observe that, for all $t$ and $i$,
    the language $L\pos{t,i}$ is sparse.
    Recall that from the definition of a partial solution,
    all vertices in $G\pos{V_t \setminus X_t}$
    must have the correct number of neighbors in $S_x$
    whereas this might not be the case for the vertices of $X_t$.

    The algorithm now reads as follows.
    We traverse the tree decomposition in post-order and
    depending on the type of the current node $t$,
    we apply one of the following procedures.
    \begin{description}[itemindent=!]
        \item[Leaf Node.]
              For a leaf node $t$ of the tree decomposition,
              we have $X_t = \emptyset$.
              Hence, we set $L\pos{t,0} \deff \{\epsilon\}$
              where $\epsilon$ denotes the empty vector.
              For all $i \in \numb{n}$, we set $L\pos{t,i} \deff \emptyset $.

        \item[Introduce Node.]
              Let $t$ be the introduce-node and $t'$ its child.
              Furthermore, let $v$ be the introduced vertex.
              Each solution in $L[t',i]$ provides two solutions for $t$
              depending on whether we select $v$ or not.

              Before moving to the definition of the table entry,
              we first introduce some notation.
              Let $y \in \allStates^{X_{t'}}$ be a state-vector
              for the child node $t'$
              and let $c \deff |\{u \in N(v) \mid y\pos{u} \in \sigStates\}$
              be the number of neighbors of $v$
              that are selected according to string $y$.
              The $\rho$-extension of $y$ is defined as the unique string $x$
              such that
              \begin{itemize}
                  \item $x\pos{v} \deff \rho_{c \bmod \mname}$
                        and
                  \item
                        $x\pos{X_{t'}} \deff y\pos{X_{t'}}$.
              \end{itemize}

              Similarly, we define the $\sigma$-extension of $y$
              as the unique string $x$ such that
              \begin{itemize}
                  \item $x\pos{v} \deff \sigma_{c \bmod \mname}$
                  \item
                        $x \pos{X_{t'} \setminus N(v)}
                            \deff y\pos{X_{t'} \setminus N(v)}$,
                  \item
                        for all $u \in X_{t'} \cap N(v)$
                        with $y\pos{u} = \tau_k$,
                        we set $x\pos{u} \deff \tau_{k+1 \bmod \mname}$.
              \end{itemize}

              Then, we define the resulting language as
              \begin{align*}
                  L\pos{t,i}  \deff
                   & \{x \mid x  \text{ is the $\rho$-extension of } y \in L\pos{t',i}\} \\
                   & \cup
                  \{x \mid x \text{ is the $\sigma$-extension of } y \in L\pos{t',i}\}
              \end{align*}
              for all $i \in \numbZ{n}$.

        \item[Forget Node.]
              Let $t$ be a forget node with child $t'$
              and let $v$ be the vertex that is forgotten.
              The language $L\pos{t,i}$ should contain those strings
              of languages of $t'$
              in which $v$ has a neighbor count that is contained in $\sigma$
              if $v$ is selected,
              and a neighbor count contained in $\rho$ if $v$ is not selected.
              Furthermore, if we forget a selected vertex,
              then we must take this into account with our second index $i$.

              Hence, we set
              \begin{align}
                  L\pos{t,0} \deff~ &
                  \{ x\pos{X_t} \mid x \in L\pos{t',0}
                  \land x\pos{v} = \rho_c
                  \land c \in \rho \}
                  \label{eqn:forget:caseOne}
                  \intertext{and, for all \(i \in \numb{n}\), we set}
                  L\pos{t,i} \deff~ &
                  \{ x\pos{X_t} \mid x \in L\pos{t',i}
                  \land x\pos{v} = \rho_c
                  \land c \in \rho \}
                  \nonumber
                  \\
                                    & \cup \{x\pos{X_t} \mid x \in L\pos{t',i-1}
                  \land x\pos{v} = \sigma_c \land c \in \sigma \}
                  \label{eqn:forget:caseTwo}
                  .
              \end{align}

        \item[Join Node.]
              Let $t$ be a join node with children $t_1$ and $t_2$.

              We clearly want to use the previous algorithm
              to join languages efficiently.
              However, when doing this naively with the languages from the child nodes,
              we count the selected neighbors in the bag twice
              by simply summing up the number of neighbors
              as the selected vertices are selected in \emph{both} partial solutions.
              Thus, we define an auxiliary language, that can be computed efficiently.
              For a language $L$, we define the auxiliary language $\hat L$ as
              \begin{align*}
                  \hat L \deff~       & \{\hat x \mid x \in L\}
                  \intertext{where, for each \(v \in X_t\), we set}
                  \hat x\pos v \deff~ & \begin{cases}
                                            \sigma_{\hat c} &
                                            \text{if }
                                            x\pos v = \sigma_c
                                            \text{, and }
                                            \hat c = (c - |\{w \in N(v) \cap X_t \mid x\pos w \in \mathbb{S}\}|) \bmod \mname, \\
                                            \rho_{\hat c}   &
                                            \text{if }
                                            x\pos v = \rho_c
                                            \text{, and }
                                            \hat c = (c - |\{w \in N(v) \cap X_t \mid x\pos w \in \mathbb{S}\}|) \bmod \mname
                                            .
                                        \end{cases}
              \end{align*}

              Now, we can use the procedure from \cref{thm:quick_combinations}
              to join the languages of the two children efficiently and set
              \begin{equation*}
                  L\pos{t,i} \deff
                  \bigcup_{j \in \numbZ{i}} \hat L\pos{t_1,j} \oplus L\pos{t_2,i-j}
              \end{equation*}
    \end{description}

    Finally, a solution of size exactly $i$ exists if and only if,
    for the root node $r$,
    the table entry $L\pos{r, i}$ contains the empty string $\epsilon$
    as the bag of the root node is empty.

    \begin{claim}
        The algorithm is correct.
    \end{claim}
    \begin{claimproof}
        The correctness of the algorithm for the leaf nodes,
        introduce nodes, and forget nodes
        follows directly from the above description of the algorithm.

        For the join node it remains to show that
        the language $\hat L$ is also sparse when $L\pos{t,i}$ is a sparse language.
        For every $x \in L\pos{t,i}$,
        we have that there exists a compatible partial solution $S_x$
        for the graph with portals $(G\pos{V_t}, X_t)$ with $|S_x \setminus U| = i$.
        If we now consider the graph with portals
        $G' = (G\pos{V_t} - E(X_t,X_t), X_t)$,
        where $E(X_t, X_t)$ are the edges
        with both endpoints in $X_t$,
        then we can observe that $S_x$ is a compatible partial solution
        for this graph that witnesses the string $\hat x$.
        In other words, for every $x \in \hat L\pos{t,i}$,
        there is a compatible partial solution for $G'$
        with $|S_x \setminus U| = i$
        and thus, $\hat L\pos{t,i}$ is sparse by \cref{thm:structural_property}.
    \end{claimproof}

    It remains to analyze the running time of the algorithm

    \begin{claim}
        The running time of the algorithm is $\mname^\tw \cdot n^{\Oh(1)}$.
    \end{claim}
    \begin{claimproof}
        Regarding the runtime, it is easy to see that introduce nodes and forget nodes
        can be computed in time $\mname^{\tw} \cdot n^{\Oh(1)}$
        because this bounds the size of the sparse languages of each tree node,
        and processing each string of each language
        can be done in time polynomial in $n$.
        Furthermore, there are exactly $n$ such languages,
        which only contributes to the polynomial factor.

        For join nodes, we first have to compute $\hat L$ in time
        $|L| \cdot n^{\Oh(1)} = \mname^{\tw} \cdot n^{\Oh(1)}$
        since we must only iterate over $L$ once
        and processing each string is possible in time polynomial in its length.
        Observe that, when we combine two sparse languages in the join operation,
        by our definitions,
        the resulting language is also sparse by \cref{thm:structural_property}.
        Therefore, computing the combination of the languages takes time
        $\mname^{\tw} \cdot n^{\Oh(1)}$ by \cref{thm:quick_combinations}.
        Thus, each node can be processed in time
        $\mname^\tw \cdot n^{\Oh(1)}$.

        For the leaf nodes
        it is easy to see that they can be processed in constant time
        when first initializing each possible table entry with the empty set.

        As the number of nodes of the tree decomposition is polynomial in $n$,
        the dynamic programming algorithm runs in time
        $\mname^\tw \cdot n^{\Oh(1)}$.
    \end{claimproof}

    This concludes the proof of the algorithmic result for solving \srDomSet
    from \cref{thm:algorithm_runtime}.
\end{proof}

We conclude this section by introducing a generalization
of the classical \srDomSet problem
where we provide a shift-vector for every vertex.

\begin{definition}[\srDomSetShift]
    Fix two non-empty sets $\sigma$ and $\rho$ of non-negative integers.
    For a graph $G$,
    and a vector $\pi \from V(G) \to \numbZ{\abs{V}}$,
    a set $S \subseteq V(G)$ is a $(\sigma,\rho)$-$\pi$-set for $G$,
    if and only if
    (1) for all $v \in S$,
    we have $\abs{N(v) \cap S} + \pi(v) \in \sigma$,
    and (2), for all $v \in V(G) \setminus S$,
    we have $\abs{N(v) \cap S} + \pi(v) \in \rho$.

    The problem \srDomSetShift asks for a given graph $G$
    and a shift vector $\pi$
    whether there is a $(\sigma,\rho)$-$\pi$-set $S$ or not.
\end{definition}
This covers the classical \srDomSet by setting $\pi(v)=0$
for all vertices $v \in V(G)$.

When taking a closer look at our dynamic program,
one can see that there are only two places
where we actually refer to the sets $\sigma$ and $\rho$.
These two positions are the definitions
of the table entries for the forget nodes
in \cref{eqn:forget:caseOne,eqn:forget:caseTwo}.
If we modify these definitions as follows
and set
\begin{align*}
    L\pos{t,0} \deff~ &
    \{ x\pos{X_t} \mid x \in L\pos{t',0}
    \land x\pos{v} = \rho_c
    \land c + \pi(v) \in \rho \}
    \intertext{and, for all \(i \in \numb{n}\), we set}
    L\pos{t,i} \deff~ &
    \{ x\pos{X_t} \mid x \in L\pos{t',i}
    \land x\pos{v} = \rho_c
    \land c + \pi(v) \in \rho \}                                   \\
                      & \cup \{x\pos{X_t} \mid x \in L\pos{t',i-1}
    \land x\pos{v} = \sigma_c \land c + \pi(v) \in \sigma \},
\end{align*}
then we directly get the following result
by using the same proof as for \cref{thm:algorithm_runtime}.

\begin{theorem}
    Let $\sigma, \rho \subseteq \Nat$ be two residue classes
    modulo $\mname \ge 2$.
    Then, in time $\mname^\tw \cdot \abs{G}^{\Oh(1)}$ we can decide
    simultaneously for all $s$
    if a given graph $G$
    and a shift-vector $\pi \from V(G) \to \numbZ{\mname-1}$,
    whether there is a $(\sigma,\rho)$-$\pi$-set of size $s$
    when a tree decomposition of width $\tw$ is given with the input.
\end{theorem}

We immediately get the following corollary for solving \LightsOut
in the general setting.

\begin{corollary}
    We can find in time
    $2^\tw \cdot \abs{G}^{\Oh(1)}$
    an optimal solution for \LightsOut
    for an arbitrary starting position
    assuming the input graph $G$
    comes with a tree decomposition of width $\tw$.
\end{corollary}

\section{Lower Bound for the Problem with Relations}
\label{sec:lowerRel}

As explained earlier,
the first step of the lower bound is to show the hardness
of \srDomSetRel which is the generalization of \srDomSet
where we additionally allow relations.
See \cref{def:srDomSetRel} for the formal definition.
In a second step, which is presented in \cref{sec:relations},
we then remove the relations.
Concretely, in the following we focus on the proof
of \cref{thm:lower_bound_relations}.

\lowerBoundsRelations*

We establish this lower bound by a reduction
from the constraint satisfaction problem \qcsp to \srDomSetRel.
We first explain the high-level idea in the following
and then describe the formal construction of the graph
in \cref{sec:lowerRel:construction}.
We prove the correctness of the construction
and analyze the pathwidth of the graph in \cref{sec:lowerRel:properties},
before we finish the proof in \cref{sec:lowerRel:finalizing}.

\subparagraph*{High-level Idea.}

The high-level idea behind the construction follows the ideas presented
by Curticapean and Marx in \cite{CurticapeanM16}
which have later been extended
to \textsc{Generalized Matching} \cite{MarxSS21,MarxSS22}
and \srDomSet \cite{focke_tight_2023_lb}.

In contrast to these known lower bounds,
our hardness result is \emph{not} based on a direct reduction from SAT,
instead
we use \qcsp as a basis.
This problem was introduced by Lampis~\cite{Lampis20}
to serve as a general starting point to rule out algorithms with running times
$(B-\epsilon)^k \cdot n^{\Oh (1)}$ for an appropriate parameter $k$.

\defCSP*

By reducing from \qcsp,
most technicalities resulting from the change of the basis in the running time
to achieve improved lower bounds are hidden.
As an immediate consequence the correspondence
between the vertex states and the variable assignments
becomes more transparent.

To prove the required lower bounds,
our construction must overcome certain obstacles
and must fulfill certain properties.
In the following we name the four most important ones.
\begin{enumerate}
    \item
          The constructed graph must have \emph{low pathwidth},
          that is, ideally the pathwidth of the constructed graph
          is the number of variables.
    \item
          We need to \emph{encode the assignments}
          for the variables of the CSP instance
          by establishing a correspondence between these values
          and the states of so-called information vertices
          of the constructed graph.
    \item
          The assignment to the variables must be \emph{consistent}
          across all constraints,
          that is, all information vertices corresponding to the same variable
          must have the same state.
    \item
          We need to embed the \emph{constraints} of the \qcsp instance
          by the relations of the \srDomSetRel instance.
\end{enumerate}

We outline the general concepts to address these points
in the following.

\subparagraph*{Low Pathwidth.}
Assume that we are given an instance of \qcspa
containing $\variableCount$ variables and $\constraintCount$ constraints.
The general structure of the graph of the \srDomSetRel
instance is a $\variableCount \times \constraintCount$ \emph{grid},
which provides reasonably low pathwidth, i.e, pathwidth $\variableCount$.
Concretely, the graph contains $\variableCount$ rows, one row for each variable,
and $\constraintCount$ columns, one for each constraint.

\subparagraph*{Encoding the Assignment.}
For each row-column pair of the grid, there is an \emph{information vertex}
which we connect to specific gadgets
that we later refer to as \emph{managers}.
The purpose of these information vertices is to encode
the assignment to the variables of the \qcspa instance.
More precisely,
each information vertex has a specific state in a $(\sigma, \rho)$-set,
depending on the number of selected neighbors
and whether it is selected itself
(although we later fix the vertices to be unselected).
Since each information vertex can have at least $\mname$ states
when fixing its selection status,
we can directly associate a state of the information vertex
to a specific value from $1$ to $\mname$ the variables can have.

\subparagraph*{Consistency of the Assignment.}
For the correctness of the reduction,
we need to ensure that the assignment for the variables
is encoded consistently across the construction.
By the above idea,
each row of the construction corresponds to one variable of the input instance.
Moreover, every such row contains multiple information vertices
which are not directly connected to each other.
The construction must guarantee that each vertex of a certain row
receives the same state
and thus, the CSP solution induced by a $(\sigma,\rho)$-set
assigns exactly one value to each variable.
It turns out that, by using the right type of additional relations,
for each fixed variable, the consistency across the row can be achieved.

\subparagraph*{Encoding the Constraints.}
As a last step, we ensure that, for every $(\sigma,\rho)$-set,
the respective states of the information vertices
correspond to a satisfying assignment to the CSP instance.
This is easily achieved by introducing a new relation
for each constraint of the CSP instance
and then adding it to specific vertices of the construction.
Combined with the previous ideas, the constraints of the CSP instance
are in direct correspondence to relations of the \srDomSetRel instance.

\subsection{Construction of the Graph}
\label{sec:lowerRel:construction}

Since our lower bound construction covers an infinite class of problems,
we construct the graph using several gadgets.
This allows us to reuse some intermediate results
from~\cite{focke_tight_2023_lb} which have been used to establish
the lower bound for \srDomSetRel when $\sigma$ and $\rho$ are finite.

We first focus on the definition of the managers.
These gadgets allow us to associate
the values of the variables of the CSP instance
with the states of the information vertices.

\begin{restatable}[{\manager{A}~\cite[Definition~4.7]{focke_tight_2023_lb}}]{definition}{defManager}
    \label{def:manager}
    Consider two sets $\sigma$ and $\rho$ that are both finite.
    For a set $A \subseteq \allStates$, an \emph{\manager{A}}
    is an infinite family $((G_\ell, \Port_\ell))_{\ell \geq 1}$ of pairs $(G_\ell, \Port_\ell)$ such that
    \begin{itemize}
        \item $G_\ell$ is a graph with relations and
        \item $\Port_\ell=\{\port_1,\dots,\port_{\ell}\} \subseteq V(G_\ell)$
              is a set of $\ell$ distinguished vertices.
    \end{itemize}
    Moreover, there is a non-negative integer $b$
    (that depends only on $\max\sigma$ and $\max\rho$)
    such that the following holds for every $\ell \geq 1$:
    \begin{itemize}
        \item
              The vertices from $V(G_\ell) \setminus \Port_\ell$
              can be partitioned into $2\ell$ vertex-disjoint sets
              $\Bl_1,\dots,\Bl_\ell$ and $\Br_1,\dots,\Br_\ell$
              (called \emph{blocks}),
              such that
              \begin{itemize}
                  \item $|\Bl_i| \leq b$ and $|\Br_i| \leq b$ for all $i \in \numb{\ell}$,
                  \item $N(\port_i) \subseteq \Bl_i \cup \Br_i$ for all $i \in \numb{\ell}$,
                  \item there are edges only between the following pairs of blocks:
                        $\Bl_i$ and $\Bl_{i+1}$,
                        $\Br_i$ and $\Br_{i+1}$, for each $i \in \numb{\ell-1}$,
                        and
                        $\Bl_\ell$ and $\Br_\ell$,
                        and
                  \item for each relation $R$ of $G_\ell$,
                        there is some $i \in \numb{\ell}$
                        such that either $\scope(R) \subseteq N \pos{\Bl_i}$
                        or $\scope(R) \subseteq N\pos{\Br_i}$.
              \end{itemize}
        \item
              Each $x \in A^\ell \subseteq \allStates^{\ell}$ is \emph{managed} in the sense that
              there is a unique $(\sigma,\rho)$-set $S_x$ of $G_\ell$
              such that for all $i\in\numb{\ell}$:
              \begin{itemize}
                  \item
                        If $x\position{i} = \sigma_s$,
                        then $\port_{i} \in S_x$.
                        Moreover, $\port_i$ has exactly $s$ neighbors in $\Bl_i\cap S_x$
                        and exactly $\max \sigma - s$ neighbors in $\Br_i\cap S_x$.
                  \item
                        If $x\position{i} = \rho_r$,
                        then $\port_{i} \notin S_x$.
                        Moreover, $\port_i$ has exactly $r$ neighbors in $\Bl_i\cap S_x$
                        and exactly $\max \rho - r$ neighbors in $\Br_i\cap S_x$.
              \end{itemize}
    \end{itemize}
    We refer to $G_\ell$ as the \manager{A} of \rank $\ell$.
\end{restatable}

For certain cases, especially when $\sigma$ and $\rho$ are finite,
the existence of such managers is already known~\cite{focke_tight_2023_lb}.

\begin{lemma}[Implication of~{\cite[Lemma~6.1]{focke_tight_2023_lb}}]
    \label{thm:rho_manager}
    For non-empty finite sets $\sigma$ and $\rho$,
    there is an \manager{R}
    with $R = \{\rho_i \mid i \in \numbZ{\max \rho}\}$.
\end{lemma}

Although these managers require finite sets,
we can use them in a black-box fashion for our construction
even when dealing with residue classes as sets.
Hence,
we simply use a finite subset of $\rho$ to obtain the managers.
We formalize this by the notion of a \emph{cut set},
which is a natural restriction of a residue class
to a finite subset of it.
We then extend this to the \emph{cut states}.
\begin{definition}[Cut Sets and States]
    For a residue class \(\tau\) modulo $\mname$,
    we set \[
        \cutset(\tau) \deff \{\min(\tau),\min(\tau) + \mname\}
        \quad\text{and}\quad \rhoStates_\cutset \deff
        \{\rho_i \mid i \in \numbZ{\max(\cutset(\rho))}\}.
        \tag*{\qedhere}
    \]
\end{definition}
For example, consider the residue class~$\tau = \{2,5,8,\dots\}$ of two modulo three.
With the above definition we get $\cutset(\tau) = \{2,5\}$.
Clearly, $\cutset(\tau)$ is finite for any residue class~$\tau$.

Before we move to the construction of the graph,
we introduce the inverse of a number.
This definition is inspired by the definition of an inverse of a state
from~\cite[Definition~3.10]{focke_tight_2023_lb}.

\begin{definition}[Inverse of a Number]
    \label{def:inverseOfNumber}
    Let $\sigma$ and $\rho$ be two residue classes
    modulo $\mname \ge 2$.

    For every number $n \in \Nat$,
    the \emph{inverse of $n$ relative to $\rho$} is defined as
    $\inverse[\rho]{n} \deff (\min \rho - n) \bmod \mname$.

    We extend this in the natural way to $\sigma$
    and denote the inverse of $n$ relative to $\sigma$ by $\inverse[\sigma]{n}$.
\end{definition}
Observe that for every $n \in \numbZ{\mname-1}$, we have
\[
    \inverse[\rho]{\inverse[\rho]{n}}
    = n
    \quad\text{and}\quad
    n + \inverse[\rho]{n} \in \rho
    .
\]

\subparagraph*{Construction of the Graph.}
Consider a \qcspa{} instance $I = (X, \CC)$ with
$\variableCount$ variables $X = \{\variable_1,\dots,\variable_\variableCount\}$
and $\constraintCount$ constraints
$\CC = \{\constraint_1,\dots,\constraint_{\constraintCount}\}$.
We construct an equivalent instance $\reductionGraph$ of \srDomSetRel
with low pathwidth.

We first require a suitable manager.
For this purpose, we use \cref{thm:rho_manager}
with the sets $\cutset(\sigma)$ and $\cutset(\rho)$
to construct an \manager{\cutStateSet} of rank $\variableCount$
(as an instance of \DomSetRel{\cutset(\sigma)}{\cutset(\rho)}).
Note that $\cutStateSet$ contains states
that are indistinguishable from other states, and in particular, $\rhoStates \subseteq \cutStateSet$.
Moreover, although this is technically not explicitly stated in \cref{thm:rho_manager}, a quick look at the proof of \cite[Lemma~6.1]{focke_tight_2023_lb} confirms that the required manager of rank $\variableCount$ can be computed in polynomial-time.

Since the constructed graph $\reductionGraph$ has a grid-like structure,
we often refer to objects by their row and column.
We use the convention that the row is written as a subscript index
while the columns are denoted by superscript indices.

Now we have everything ready to finally define the $\reductionGraph$
as follows:
\begin{itemize}
    \item
          For all $i \in \rowSet$ and all $j \in \colSet$,
          the graph $\reductionGraph$ contains a vertex $w^j_i$,
          which we refer to as \emph{information vertex}.
    \item
          For all $j \in \colSet$,
          there is a copy $M^j$ of the $\cutStateSet$-manager of rank $\variableCount$
          with $\{w^j_{i} \mid i \in \rowSet\}$ as the distinguished vertices.

          We denote by $\Bl^j_{i}$ and $\Br^j_{i}$ the vertices of
          the corresponding blocks of $M^j$.
          Moreover,
          we denote by $N_\Bl(w^j_{i})$ the neighborhood of $w^j_{i}$ in $\Bl^j_{i}$,
          and denote by $N_{\Br}(w^j_{i})$ the neighborhood of $w^j_{i}$ in $\Br^j_{i}$.
    \item
          For all $i \in \rowSet$ and all $j \in \fragment{2}{\colCount}$,
          we create a \emph{consistency relation} $\consistencyRelation^{j}_i$
          with scope $\Br^{j-1}_i \cup \Bl^{j}_i \cup \{w^{j}_i\}$
          which is defined later.

    \item
          For all $i \in \rowSet$,
          we create a consistency relation $\consistencyRelation^1_i$
          with scope $\Bl^1_i \cup \{w^1_i\}$
          which is defined later.

    \item
          For all $j \in \colSet$,
          we create a \emph{constraint relation} $\constraintRelation{j}$
          with scope
          $\bigcup_{x_i \in \scope(C_j)} B^j_{i}$
          which is defined later.
\end{itemize}

As a last step of the construction we define
the consistency relations $\consistencyRelation^j_i$
and the constraint relations $\constraintRelation{j}$.
For this we first introduce some notation.

Given a selection $S$ of vertices from $\reductionGraph$,
we define,
for all $i \in \rowSet$ and all $j \in \colSet$,
two numbers $b^j_{i}$ and $\inv b^j_{i}$.
We denote by $b^j_i \deff \abs{S \cap N_ \Bl(w^j_{i})}$
the number of selected neighbors of $w^j_{i}$ in the block $\Bl^j_{i}$,
and denote by $\inv b^j_i \deff \abs{S \cap N_{\Br}(w^j_{i})}$
the number of selected neighbors of $w^j_{i}$ in the block $\Br^j_{i}$.

It remains to elaborate on the relations, whose definition we postponed so far.
We start with the \emph{consistency relations}.
Let $S$ be a subset of the vertices of $\reductionGraph$.
For all $i \in \numb{n}$ and $j \in \numb{\colCount}$,
we denote by $S^j_i \deff S \cap \scope{(\consistencyRelation^j_i)}$
the selected vertices from the scope of $\consistencyRelation^j_i$.
Relation $\consistencyRelation^j_i$ accepts the set $S^j_i$
if and only if
\begin{itemize}
    \item
          vertex~$w_i^j$ is unselected,
    \item
          $b^{j}_i \in \numbZ{\mname-1}$,
          and
    \item
          $b^{j}_{i} = \inverse[\rho]{\inv b^{j-1}_{i}}$
          if~$j \ge 1$ (as $\inv b^0_i$ is undefined).
\end{itemize}

Finally, we define the \emph{constraint relations}
that realize the constraints of the \qcspa instance.
For each $j \in \colSet$,
denote by $\scope(\constraint_j)
    = (\variable_{\tupleindex{1}},\dots,\variable_{\tupleindex{q}})$
the variables appearing in this constraint.
Then, the set $S \cap \scope{(\constraintRelation{j})}$
is accepted by $\constraintRelation{j}$ if and only if
\[
    (b^j_{\tupleindex{1}}+1,
    \dots,b^j_{\tupleindex{q}}+1)
    \in \acc(\constraint_j)
    .
\]
That is, the states of the information vertices
corresponding to the variables of the constraint $\constraint_j$
must represent a satisfying assignment of $\constraint_j$.
Note that the variables takes values from $\numb{\mname}$
while the states $b_i^j$ take values from $\numbZ{\mname-1}$.

This concludes the construction of the instance $\reductionGraph$.

\subsection{Properties of the Constructed Graph}
\label{sec:lowerRel:properties}

In the following part we first prove the correctness of the reduction
and then provide the formal bound on the pathwidth.

We show the two directions of the correctness independently.
\begin{lemma}
    \label{lem:lowerRel:correctness1}
    If $I$ is a satisfiable instance of \qcspa,
    then the \srDomSetRel instance $\reductionGraph$ has a solution.
\end{lemma}
\begin{proof}
    Recall that $I = (X, \CC)$ is a \qcspa instance with $\variableCount$
    variables $X = \{\variable_1,\dots,\variable_\variableCount\}$
    and $\constraintCount$ constraints
    $\CC = \{\constraint_1,\dots,\constraint_{\constraintCount}\}$.
    Consider an assignment $\satAssignment \from X \to \numb{\mname}$
    that satisfies all constraints of $I$.

    For each $j \in \colSet$, we select the vertices of the \manager{\cutStateSet} $M^j$
    such that $b^j_{i} = \satAssignment(\variable_i)-1$ for all $i \in \rowSet$.
    We select no other vertices.
    Such a solution exists by the definition of a manager in \cref{def:manager}
    as $\rho_{\satAssignment(\variable_i)-1} \in \cutStateSet$.

    Let $S$ denote this set of selected vertices.
    Because of the definition of managers,
    this set $S$ is a solution for $\reductionGraph$
    when considering it as a \DomSet{\cutset(\sigma)}{\cutset(\rho)} instance
    where we ignore all relations.
    Since $\cutset(\sigma)$ and $\cutset(\rho)$
    are subsets of $\sigma$ and $\rho$,
    the set $S$ is actually a solution for the \srDomSet instance.

    Hence, all that remains is to argue that
    $S$ also satisfies all relations of $\reductionGraph$.
    Per construction, all information vertices are unselected.
    As the proof follows directly for $\consistencyRelation^1_i$
    where $i \in \rowSet$,
    we consider a relation $\consistencyRelation^j_{i}$
    for an arbitrary $i \in \rowSet$ and $j \in \fragment{2}{\colCount}$.
    We now prove that $b^{j}_{i} = \inverse[\rho]{\inv b^{j-1}_{i}}$.

    It directly follows from the definition of the managers
    and the choice of $S$ (which implies $b^{j-1}_{i} = b^{j}_{i}$)
    that $\inv b^{j-1}_{i} = (\min\rho + \mname) - b^{j-1}_{i}
        = \min\rho + \mname - b^{j}_{i}$.
    We furthermore know that $b^{j}_{i} \leq \mname - 1$ as $b^{j}_{i} = \satAssignment(\variable_i) -1$, and $\satAssignment(\variable_i) \in \numb{\mname}$.
    This in turn implies that $b^{j}_{i} \mod \mname = b^{j}_{i}$.
    Overall, we have
    \[
        \inverse[\rho]{\inv b^{j-1}_{i}}
        = \min\rho - (\min\rho + \mname - b^{j}_{i}) \mod \mname
        = b^{j}_{i} \mod \mname = b^{j}_{i}
        .
    \]

    Finally, we need to confirm that the constraint relations
    $\constraintRelation{j}$ are satisfied for all $j \in \colSet$.
    This directly follows from our selection of the vertices
    such that $b^j_i = \satAssignment(\variable_i) -1$
    for all $i \in \rowSet$ and $j \in \colSet$.
    Hence, all relations accept
    because $\pi$ is an assignment satisfying all constraints of $I$.
\end{proof}

As a next step we prove the reverse direction of the correctness.

\begin{lemma}
    \label{lem:lowerRel:correctness2}
    If the \srDomSetRel instance $\reductionGraph$ has a solution,
    then the \qcspa instance $I$ is satisfiable.
\end{lemma}
\begin{proof}
    Let $S$ be a $(\sigma,\rho)$-set for $\reductionGraph$
    that also satisfies all relations.
    We start by proving that
    \begin{equation}
        b^{j-1}_{i} = b^{j}_{i}
        \quad
        \text{ for all \(i \in \rowSet\) and \(j \in \fragment{2}{\colCount}\)}.
        \label{eq:lower:monotonicity}
    \end{equation}
    For this fix an arbitrary pair
    $i \in \rowSet$ and $j \in \fragment{2}{\colCount}$.
    The relation $\consistencyRelation_i^j$ ensures that
    $b^{j}_{i} = \inverse[\rho]{\inv b^{j-1}_i}$.
    Moreover, the relation guarantees that $w^j_{i}$ is not selected.

    The definition of $\consistencyRelation^j_i$ implies
    that $b^{j}_{i} = \inverse[\rho]{\inv b^{j-1}_{i}}
        = (\min\rho - \inv b^{j-1}_{i}) \bmod \mname$
    and moreover, since $w^{j}_i$ has a neighbor count in $\rho$,
    we must also have $b^{j-1}_{i} + \inv b^{j-1}_{i} \equiv_{\mname} \min\rho$.
    By rearranging the terms we get
    $\inv b^{j-1}_{i} \equiv_{\mname} \min\rho - b^{j-1}_{i}$,
    and when combining with the previous result, it follows that
    $b^{j}_{i} \equiv_{\mname} \min\rho - (\min\rho - b^{j-1}_{i})
        \equiv_{\mname} b^{j-1}_{i}$.
    Since $S$ is a solution, for all $k \in \numb{\colCount}$,
    the relation $\consistencyRelation_i^k$ ensures that
    $b^k_i \in \numbZ{\mname-1}$.
    Using $b^{j-1}_{i} \leq {\mname - 1}$ and $b^{j}_{i} \leq {\mname - 1}$,
    we conclude that $b^{j-1}_{i} = b^{j}_{i}$.
    Hence, the states of all information vertices for a single variable are the same.

    We can now define a satisfying assignment for $I$.
    We define the variable assignment
    $\satAssignment \from X \to \numb{\mname}$ by setting
    $\satAssignment(\variable_i) \deff b^{1}_{i}+1$
    for all $i \in \rowSet$.

    To see that $\pi$ is an assignment satisfying all constraints of $I$,
    consider an arbitrary constraint $\constraint_j$ of $I$.
    The relation $\constraintRelation{j}$ then ensures that
    $(b^{j}_{\lambda_1}+1, \dots,
        b^{j}_{\lambda_q}+1) \in \acc(\constraint_j)$,
    where $\scope(C_j) = (\variable_{\tupleindex{1}},\dots,\variable_{\tupleindex{q}})$ are the variables of the constraint.
    Hence, the definition of $\satAssignment$
    implies that $\constraint_j$ is satisfied.
\end{proof}

This concludes the proof of the correctness of the reduction.
It remains to analyze the size and the pathwidth
such that the lower bound of \qcspa transfers to \srDomSetRel.
However, for this we first have to formally introduce
the notion of pathwidth of a graph with relations.

\begin{definition}[{Width Measures for Graphs with Relations;~\cite[Definition~4.4]{focke_tight_2023_lb}}]
    \label{def:graphWithRelationsWidthMeasures}
    Let $G=(V,E,\CC)$ be a graph with relations.
    Let $\hat G$ be the graph we obtain from $(V,E)$
    when, for all $C \in \CC$,
    we additionally introduce a complete set of edges on the scope $\scope(C)$.
    The \emph{treewidth of a graph with relations} $G$,
    is the treewidth of the graph $\hat G$.
    Analogously, we define tree decompositions, path decompositions, and pathwidth of $G$
    as the corresponding concepts in the graph $\hat G$.
\end{definition}

Next, we show that after this transformation,
the pathwidth of the constructed graph is not too large.

\begin{lemma}
    \label{thm:reduced_instance_size_pathwidth}
    There is a function $f$
    depending only on $\sigma$ and $\rho$
    such that $\reductionGraph$ has size at most
    $\variableCount \cdot \constraintCount \cdot f(q)$,
    and arity at most $f(q)$.
    Moreover, a path decomposition of $\reductionGraph$ of width at most $\variableCount + f(q)$
    can be computed in polynomial time.
\end{lemma}
\begin{proof}
    We first elaborate on the number of vertices of $\reductionGraph$.
    Graph $\reductionGraph$ contains $\colCount$ copies
    of an \manager{\cutStateSet} of rank $\rowCount$.
    By \cref{def:manager}, the size of each block of an \manager{\cutStateSet}
    is at most $b$, for some constant $b$ depending
    only on $\cutset(\sigma)$ and $\cutset(\rho)$.
    Since $\cutset(\sigma)$ and $\cutset(\rho)$ only depend on $\sigma$ and $\rho$,
    it is evident that $b$ only depends on $\sigma$ and $\rho$.
    Since each manager consists of exactly $2 \rowCount$ blocks
    as well as $\rowCount$ information vertices,
    we see that $\reductionGraph$ consists of
    at most
    $\colCount \cdot (2 \rowCount \cdot b + \rowCount)$ vertices.
    Observe that
    \begin{align*}
        \colCount \cdot (2 \rowCount \cdot b + \rowCount) = \rowCount \cdot \constraintCount \cdot c_\textsf{v}
    \end{align*}
    for an appropriately chosen constant $c_\textsf{v}$
    only depending on $\sigma$ and $\rho$.

    For the arity of the relations,
    we notice that each consistency relations has arity at most $2 \cdot (b + 1)$,
    each constraint relation has arity at most $q \cdot (b + 1)$,
    and the arity of each relation stemming from a manager is in $\Oh(b)$.
    Hence, the maximum arity of each relation is bounded by
    \(\Oh(q \cdot b).\)

    Recall from \cref{def:graphWithRelations},
    that the size of a graph with relations
    is defined as the number of vertices plus the size of each relation
    (which might be exponential in the arity).
    The graph $\reductionGraph$
    contains exactly $\rowCount \cdot \colCount$ consistency relations
    and $\colCount$ constraint relations
    plus the relations of the managers
    each of which has at most~$n \cdot 2^{\Oh(b)}$ relations
    of arity at most $\Oh(b)$.

    Combining our knowledge of the number of vertices,
    with the maximal arity of any relation, and the number of relations,
    we can conclude that the size of $\reductionGraph$ is bounded by
    \begin{align*}
        \variableCount \cdot \constraintCount \cdot c_\textsf{v}
        + \variableCount \cdot \constraintCount \cdot 2^{\Oh(q \cdot b)}
        = \variableCount \cdot \constraintCount \cdot f_1(q,\sigma, \rho),
    \end{align*}
    for an appropriately chosen function $f_1(q,\sigma, \rho)$.

    It remains to bound the pathwidth of $\reductionGraph$.
    We do this by the standard approach of providing a node search strategy
    (see, for example, \cite[Section~7.5]{cyganParameterizedAlgorithms} or \cite{FominT08_annotated_bibliography_graph_searching}).
    From \cref{def:graphWithRelationsWidthMeasures},
    we know that the pathwidth of a graph $G$ with relations
    is defined as the pathwidth of the graph
    we obtain when making the vertices in the scope of each relation a clique.
    Let $\hat\reductionGraph$ be the graph obtained from $\reductionGraph$
    by this modification (while keeping all indexed vertices/sets the same).

    The graph is cleaned in $\colCount + 1$ stages,
    where each stage consists of $\rowCount$ rounds.
    Intuitively, each stage is responsible for cleaning the left side
    of one column of the construction,
    and each round for cleaning a block of the column.

    For each round, we list the vertices on which searchers are placed.
    This makes it clear that one can go from one stage to the next
    without recontaminating already cleaned parts
    and without the use of additional searchers.
    For notational convenience, we define
    \begin{itemize}
        \item
              $w^j_i$ as a dummy vertex that is not part of the graph $\hat\reductionGraph$ whenever $i \not \in \rowSet$ or $j \not \in \colSet$,
        \item
              $B^{j}_{i}$ and $\inv B^j_i$ to be the empty set whenever $i \not \in \rowSet$ or $j \not \in \colSet$,
        \item
              $\scope(\constraintRelation{j})$ to be the empty set when $j \not \in \colSet$.
    \end{itemize}

    Let $\CS^j_{i}$ denote the set of vertices on which searchers are placed
    in round $i$ of stage $j$.
    We define this set as
    \begin{align*}
        \CS^j_{i}
        =\  & \{w^{j-1}_{x} \mid i \leq x \leq n\}
        \cup \{w^{j}_{x} \mid 1 \leq x \leq i+1\}                   \\
            & \cup N\pos{\Br^{j-1}_{i}} \cup N\pos{\Bl^{j}_{i}}
        \cup N\pos{\Br^{j-1}_{i+1}} \cup N\pos{\Bl^{j}_{i+1}}
        \\
            & \cup \Br^{j-1}_{\rowCount} \cup \Bl^{j-1}_{\rowCount} \\
            & \cup \scope{(\constraintRelation{j})}
    \end{align*}
    where we use the closed neighborhood for the block
    to cover all possible relations appearing in the managers.

    First observe that every vertex and both endpoints of every edge of the graph
    are contained in some set $\CS_i^j$.
    It remains to argue that the graph does not get recontaminated.
    Consider the intersection of the vertices
    from two consecutive rounds of the same stage,
    that is,
    \begin{align*}
        \CS_i^j \cap \CS_{i+1}^j
        = \,&
        \{w^{j-1}_{x} \mid i+1 \leq x \leq n\}
        \cup \{w^{j}_{x} \mid 1 \leq x \leq i+1\}
        \\&
        \cup N\pos{\Br^{j-1}_{i+1}}
        \cup N\pos{\Bl^{j}_{i+1}}
        \cup N\pos{\Br^{j-1}_{\rowCount}}
        \cup N\pos{\Bl^{j-1}_{\rowCount}}
        \cup \scope{(\constraintRelation{j})}
        .
    \end{align*}
    As these vertices form a separator of the graph,
    the cleaned part of the graph does not get recontaminated.

    When moving from one stage to the next one,
    we can use the same technique by observing that
    \begin{align*}
        \CS_\variableCount^j \cap \CS_{1}^{j+1}
        =
        \{w^{j}_{x} \mid 1 \leq x \leq \variableCount\}
        \cup \Bl^{j}_{\variableCount}
    \end{align*}
    separates the graph.
    We conclude that the node search number of
    $\hat \reductionGraph$ is at most
    \begin{align*}
        \max_{\substack{i \in \fragment{1}{\rowCount}, \\j \in \fragment{1}{\colCount + 1}}}
        \Big|\CS^j_{i} \Big|
        = \variableCount + 2 + 10b + q \cdot (b + 1)
        ,
    \end{align*}
    which means that the pathwidth of $\hat \reductionGraph$
    and thus, of $\reductionGraph$ is at most
    \begin{align*}
        \variableCount + f_2(q, \sigma, \rho),
    \end{align*}
    for an appropriately chosen function $f_2(q, \sigma, \rho)$.

    Note that the search strategy directly corresponds to a path decomposition.
    Hence,
    a corresponding decomposition can be computed in polynomial time.
    This concludes the proof by choosing the function $f$
    from the statement as the maximum of $f_1$ and $f_2$.
\end{proof}

\subsection{Combining the Results}
\label{sec:lowerRel:finalizing}

As a last ingredient to our proof of \cref{thm:lower_bound_relations},
we need the SETH-based lower bound for \qcspa
from \cref{thm:qcsp_hardness}.

\qcspHardness*

Finally, we are ready to conclude the proof of the
lower bound for \srDomSetRel from \cref{thm:lower_bound_relations}.

\lowerBoundsRelations
\begin{proof}
    We assume for contradiction's sake,
    that, for some $\epsilon > 0$, there exists an algorithm
    that can solve \srDomSetRel in time
    $(\mname - \varepsilon)^\pw \cdot \abs{G}^{\Oh(1)}$
    when the input contains a graph $G$
    and a path decomposition of $G$ of width $\pw$.

    Let $q$ be the arity from \cref{thm:qcsp_hardness}
    such that there is no algorithm that can solve \qcspa
    in time $(\mname - \varepsilon)^n \cdot (n+\ell)^{\Oh(1)}$
    for an instance with $n$ variables and $\ell$ constraints.

    Given an arbitrary \qcspa instance $I$ as input,
    let $\reductionGraph$ denote the corresponding \srDomSetRel instance
    from \cref{sec:lowerRel:construction}.
    We use \cref{thm:reduced_instance_size_pathwidth}
    to compute a corresponding path decomposition of width $n + f(q)$
    in polynomial time.
    Then, we run the hypothetical algorithm for \srDomSetRel on this instance
    and return the output as the output of the \qcspa instance.

    Since, by \cref{lem:lowerRel:correctness1,lem:lowerRel:correctness2},
    the instance $\reductionGraph$ has a solution if and only if $I$ is satisfiable,
    this algorithm correctly decides if the \qcspa instance $I$ is satisfiable.

    It remains to analyze the running time of the algorithm.
    The construction of $\reductionGraph$ (which bounds the size) and its path decomposition takes time
    polynomial in the size of $I$ (and thus, also in $\reductionGraph$).
    Hence, the final algorithm runs in time
    \[
        (\mname -\varepsilon)^{n + f(q)} \cdot (n+\ell)^{\Oh(1)}
        = (\mname - \varepsilon)^n \cdot (n+\ell)^{\Oh(1)}
    \]
    since $q$ is a constant that only depends on $\sigma$ and $\rho$,
    which are fixed sets and thus, not part of the input.
    Therefore, this directly contradicts SETH by \cref{thm:qcsp_hardness}
    and concludes the proof.
\end{proof}

\section{Realizing Relations}
\label{sec:relations}

\Cref{sec:lowerRel} covers the proof of the conditional lower bound
for \srDomSetRel based on a reduction from \qcspa,
see \cref{thm:lower_bound_relations} for the precise statement.
From a higher level perspective, this reduction just considers
how the values for the variables are chosen
and encodes this choice by a graph problem;
the constraints are hardly changed.
More concretely, the constraints of the CSP instance
are not (yet) encoded by graph gadgets,
but the reduction just transformed them into relations
of the resulting \srDomSetRel instance.
Therefore, our next step is to remove these relations
or rather replace them by appropriate gadgets.
Recall that the reduction from \cref{thm:lower_bound_relations}
also introduces relations that do not correspond to constraints
of the CSP instance due to technical reasons.
In the following we design a reduction which replaces all relations
that appear in the given \srDomSetRel instance by suitable gadgets.

Formally, this section is dedicated entirely to the proof of
\cref{lem:relations:reductionToRemoveRelations}.

\reductionToRemoveRelations*

Similar to the known lower bound for general classes of problems
(especially the bound for \srDomSet
when $\sigma$ and $\rho$ are finite or simple cofinite
from~\cite[Section~8]{focke_tight_2023_lb}
but also the constructions in \cite{CurticapeanM16,MarxSS21,MarxSS22}),
we divide the reduction into smaller steps.
As a first step, we replace the general relations,
which might be very complex, by relations that are simple and easy to define.
We mainly base our construction on the existence of \HWeq{1}
relations and variants there of which we formally introduce as follows.
\begin{definition}[Relations]
  \label{def:usefulRelations}
  Consider an integer $k \ge 1$
  and an arbitrary non-empty set $\tau$ of non-negative integers.
  We define \HWin[k]{\tau} as the relation of arity $k$ that only accepts
  exactly $t$ selected vertices in the scope if and only if $t \in \tau$.
  Otherwise, the relation does not accept.

  To simplify notation, we write \HWeq[k]{1} for the \HWin[k]{\{1\}}.
  Moreover, we denote by \HWin{\tau} (or similarly for the other relations)
  to the set of the relations \HWin[k]{\tau}
  for every arity $t$.
  \qedhere
\end{definition}

As a first step, we make use of~\cite[Corollary~8.8]{focke_tight_2023_lb},
which holds for arbitrary, non-empty sets,
to replace the arbitrary relations by gadgets using only \HWeq{1} as relations.
We denote by \srDomSetRel[\HWeq{1}]
the variant of \srDomSetRel where all relations are \HWeq{1} relations.

\begin{lemma}[{\cite[Corollary~8.8]{focke_tight_2023_lb}}]
  \label{lem:relations:reductionArbitraryToHWone}
  Let $\sigma,\rho$ denote non-empty sets with $\rho \neq \{ 0 \}$.
  For all constants~$d$,
  there is a polynomial-time, parsimonious reduction from \srDomSetRel
  on instances of arity at most $d$ given with a path decomposition of width $p$
  to \srDomSetRel[{\HWset{1}}]
  on instances of arity at most $2^d+1$ given with a path decomposition of width $p + \Oh(1)$.
\end{lemma}

With the reduction from \cref{lem:relations:reductionArbitraryToHWone},
it remains to design gadgets realizing \HWeq{1} relations.
Intuitively \emph{a gadget $F$ realizes some relation $R$}
if removing the relation $R$ and inserting the gadget $F$
does not introduce any new solutions but also does not remove solutions
when projecting onto the remaining graph.
We formalize this intuition in the following definition.
\begin{definition}[{Realization of a Relation;~\cite[Definition~8.2]{focke_tight_2023_lb}}]
  \label{def:realization}
  For a set of vertices $S$ with $d=\abs{S}$,
  let $R\subseteq 2^{S}$ denote a $d$-ary relation.
  For an element \(r \in R\), we
  write \(x_r\) for the length-\(d\) string
  that is \(\sigma_0\) at every position \(v\in r\),
  and \(\rho_0\) at the remaining positions, i.e.,
  \[
    x_r\position{v} \deff
    \begin{cases}
      \sigma_0 & \text{if \(v \in r\)}, \\
      \rho_0   & \text{otherwise.}
    \end{cases}
  \]
  We set $L_R \deff \{ x_r \mid r \in R \}$.
  A graph with portals $(H,U)$ \emph{realizes} $R$
  if it realizes $L_R$.
  We say that $R$ is realizable
  if there is a graph with portals that realizes $R$.
\end{definition}

Unfortunately, the gadget constructions for the lower bound
when the sets $\sigma$ and $\rho$ are finite~\cite{focke_tight_2023_lb},
do not apply in our setting.
A main reason is that the gadgets for realizing \HWeq{1}
are heavily tailored to the set $\sigma$ and $\rho$
and thus, do usually rely on the fact that the sets are finite or cofinite.
In particular, we need to create new gadgets
that can realize \HWeq{1} relations from scratch.

For constructing the \HWeq{1} gadgets,
a significant obstacle is that in our setting
the sets $\sigma$ and $\rho$ do not necessarily have a largest element
which can be exploited in the construction.
Moreover, directly creating gadgets that realize the \HWeq{1} relation
for large arity
seems to be difficult; for a single vertex, having one selected neighbor
is essentially the same as having $\mname+1$ selected neighbors.
Instead, we first realize
the \HWone relation%
\footnote{For a set $S \subseteq \Nat$ and an integer $k \in \Nat$,
  we write $S-k$ for the set $\{s - k \mid s \in S\}$.}
which enforces that at least one vertex is selected
while not giving too many other guarantees.

We prove \cref{lem:relations:hwOne} in \cref{sec:relations:hwOne}.
\relationsHWOne*

This gadget can directly be used to instantiate \HWeq{1}
for arity one, two, or three
as stated by \cref{lem:relations:hwEqOneGeneralLowDegree}.

\begin{corollary}
  \label{lem:relations:hwEqOneGeneralLowDegree}
  Let $\sigma$ and $\rho$ be two difficult residue classes
  modulo $\mname$.
  Then, the relation $\HWeq[k]{1}$ can be realized for $k \in \numb{3}$.
\end{corollary}
\begin{proof}
  We know that $1 \in \rho - \min\rho + 1$.
  Since $\sigma$ and $\rho$ are difficult, we furthermore have $\mname \geq 3$,
  which in turn implies that $0,2,3, \not \in \rho - \min\rho + 1$.
  Thus, the realization from \cref{lem:relations:hwOne} already provides a realization of $\HWeq[k]{1}$ for $k \in \numb{3}$.
\end{proof}

\subparagraph*{Finalizing the Proof.}

With the gadgets from
\cref{lem:relations:hwEqOneGeneralLowDegree},
we can now realize \HWeq{1} relations in the general setting.
Before diving into the details of this construction, we introduce a
useful gadget provided by Focke et al.~\cite{focke_tight_2023_lb}.

\begin{lemma}[{\cite[Lemma~5.2]{focke_tight_2023_lb}}]
    \label{thm:vertex_happy_provider}
    Let $\sigma$ and $\rho$ be two arbitrary non-empty sets.
    For any $s \in \sigma$ and $r \in \rho$, there is a $\{\sigma_s,\rho_r\}$-provider.
\end{lemma}

The provider of \cref{thm:vertex_happy_provider} is extremely useful
because it allows us to introduce a vertex to a graph
that can be both selected and unselected
in feasible solutions, as long as this vertex receives no additional selected
neighbors.

Next, we show how to realize $\HWeq{1}$ with an arbitrary arity, given
gadgets realizing $\HWeq{1}$ of arity one, two, and three.

\graphicspath{{figures/tikz}}
\begin{figure}[tp]
    \centering
    \includegraphics[scale=2]{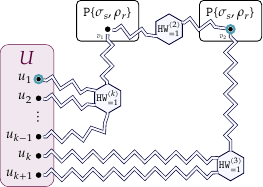}
    \caption{The gadget constructions from \cref{lem:relations:hwEqOneGeneral}.}\label{fig:2.13-1}
\end{figure}

\begin{lemma}
    \label{lem:relations:hwEqOneGeneral}
    Let $\sigma$ and $\rho$ denote sets such that
    \HWeq[1]{1}, \HWeq[2]{1}, and \HWeq[3]{1} can be realized.
    Then, for all $k \ge 1$, the relation \HWeq[k]{1} can be realized
    by a gadget of size $\Oh(k)$.
\end{lemma}
\begin{proof}
    We use the same ideas as in the proof of~\cite[Lemma 4.4]{MarxSS21},
    which uses a simple approach of obtaining a higher degree relation
    by combining \HWeq[2]{1} and \HWeq[3]{1} relations in a path-like manner.

    We proceed by strong induction.
    The base cases, $1 \leq k \leq 3$ hold by assumption.

    For the induction step $(k \geq 4)$, we assume that we can realize the relation for
    all arities from $1$ to $k$, and show that we can realize the relation for arity
    $k+1$.
    We first describe the gadget construction; then we argue about its properties.
    Also consult \cref{fig:2.13-1} for a visualization of the construction.

    Denote by $u_1,\dots,u_{k+1}$ the vertices of the relation scope.
    First, we set $s \deff \min \sigma \in \sigma$
    and $r \deff \min\rho \in \rho$, and we add to the graph two independent
    copies of the gadget from \cref{thm:vertex_happy_provider} for \(s\) and \(r\);
    call the portal vertices of said gadgets $v_1$ and $v_2$, respectively.
    Next, we add the relation $\HWeq[k]{1}$ with scope $u_1,\dots,u_{k-1},v_1$ to the graph.
    Then, we add the relation $\HWeq[2]{1}$ with scope $v_1,v_2$ to the graph.
    After that, we add the relation $\HWeq[3]{1}$ with scope $v_2,u_k,u_{k+1}$ to the graph.
    Finally, we replace all relations with the respective realization gadgets, which exist by assumption.
    For the rest of this proof, call the resulting graph \(G\).

    \begin{claim}
        The graph \(G\) realizes the $\HWeq[k+1]{1}$ relation.
    \end{claim}
    \begin{claimproof}
        First, assume that no vertex of the relation scope is selected.
        In this case, both $v_1$, and $v_2$ must be selected, which is not possible due to the $\HWeq[2]{1}$ relation.

        Next, it is not possible that two vertices of $u_1,\dots,u_{k-1}$ or two vertices
        of $u_{k},u_{k+1}$ are selected, due to the $\HWeq[k]{1}$ and $\HWeq[3]{1}$ relations.
        If one vertex of $u_1,\dots,u_{k-1}$ and one vertex of $u_k$ and $u_{k+1}$ are selected,
        then both $v_1$, and $v_2$ cannot be selected, which is once again impossible in
        any solution.

        Finally, if one vertex of $u_1,\dots,u_{k+1}$ is selected, then exactly one of
        $v_1$ and $ v_2$ must be selected, and all relations are fulfilled.
        Moreover, the vertices $v_1$ and $v_2$ can always receive a feasible number of neighbors,
        regardless of their selection status.
    \end{claimproof}

    To conclude the proof, we analyze the size of the gadget.
    The gadgets for arities $1$, $2$ and $3$, and the gadgets from \cref{thm:vertex_happy_provider} have constant size each.
    Hence, the size of the gadget grows only by a constant amount as we go from one arity
    to the next, proving that the size of the gadget for arity $k$ is linear in $k$.
\end{proof}

With \cref{lem:relations:hwEqOneGeneralLowDegree,lem:relations:hwEqOneGeneral} at hand,
we can now finally prove the main result of this section,
that is, prove \cref{lem:relations:reductionToRemoveRelations}
which we restate here for convenience.

\reductionToRemoveRelations

\begin{proof}
  Let \(\sigma\) and \(\rho\) denote sets as in the statement of the lemma.
  Further, let $I_1$ denote an instance of \srDomSetRel{},
  let $\pw$ denote the pathwidth of the graph corresponding to \(I_1\),
  and let $d$ denote the arity of the graph corresponding to \(I_1\).

  First, we apply \cref{lem:relations:reductionArbitraryToHWone}
  to obtain an equivalent instance $I_2$ of \srDomSetRel[\HWeq{1}] with
  pathwidth $\pw + \Oh(1)$ and arity $2^d + 1$.
  By \cref{lem:relations:hwEqOneGeneralLowDegree,lem:relations:hwEqOneGeneral}, we can replace all remaining relations of the graph with their realizations.
  To do this, observe that any remaining relation is a $\HWeq{1}$ relation.
  To replace such a relation with a graph, we add the graph that realizes this relation, and unify its portal vertices with the vertices of the relation.
  Write \(I_3\) to denote the resulting instance of \srDomSet{}.

  \begin{claim}
    \label{cl:2.11-1}
    The instances \(I_2\) and \(I_3\) are equivalent.
  \end{claim}
  \begin{claimproof}
    First, assume that $I_2$ is a yes-instance.
    Then, selecting the same vertices that are selected in $I_2$, and extending this solution to the newly added graphs results in a solution for instance $I_3$.
    Because all relations of $I_2$ are fulfilled, such an extension is indeed possible by \cref{def:realization}.

    Now, assume that $I_3$ is a yes-instance.
    Since no solution can select neighbors of vertices of $I_2$ that were not yet present in $I_2$, restricting the solution of $I_3$ to graph $I_2$ ensures that all vertices of $I_2$ receive a feasible number of neighbors.
    Moreover, the gadgets that were added to replace the relations ensure that exactly one vertex of each relation of $I_2$ must be selected in any solution, hence, this solution also fulfills all relations.
  \end{claimproof}

  \begin{claim}
    \label{cl:2.11-2}
    The instance \(I_3\) has a pathwidth of \(\pw + \Oh(2^d)\).
  \end{claim}
  \begin{claimproof}
    Recall from \cref{def:graphWithRelationsWidthMeasures} that the pathwidth of $I_2$ is defined as the pathwidth of the graph obtained by forming a clique out of all vertices in the relation scope for each relation.
    Let $\hat I_2$ denote the graph that is obtained from $I_2$ by applying the aforementioned transformation.
    Consider a path decomposition of $\hat I_2$.
    For any relation, there exists a bag of the decomposition in which all vertices of the relation are present.
    We can duplicate this bag and reconnect the bags in the natural way.
    Then, we simply add all vertices of the gadget that realizes the relation to the duplicated bag.
    It is easy to see that one can obtain a path decomposition of $I_3$ by performing this operation for each relation such that we never add the vertices of two realization gadgets to the same bag.
    The width of this decomposition is the width of the decomposition of $I_2$ plus the size of the largest gadget that was added to the graph.
    Using \cref{lem:relations:hwEqOneGeneral}, we observe that this results in a decomposition of width $\pw + \Oh(2^d)$, as desired.
  \end{claimproof}
  Combining \cref{cl:2.11-1,cl:2.11-2}, we obtain the claimed result.
\end{proof}

\subsection[\texorpdfstring{Realizing \HWone}%
{Realizing HW\textunderscore(in rho-min(rho))}]{\boldmath{}Realizing \HWone}
\label{sec:relations:hwOne}
In this section, we construct gadgets that realize the \HWone{} relation.
Formally, we prove \cref{lem:relations:hwOne},
which we restate here for convenience.
\graphicspath{{figures/tikz}}
\begin{figure}[tp]
  \renewcommand{\sigMin}{\min\sigma}
  \renewcommand{\rhoMin}{\min\rho}
  \centering

  \begin{subfigure}[t]{0.4\textwidth}
    \centering
    \includegraphics[height=20ex]{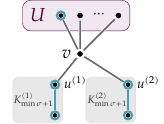}
    \caption{The case $\rhoMin \geq 2$ illustrated for $\sigMin = 1, \rhoMin = 3$.}
  \end{subfigure}\quad
  \begin{subfigure}[t]{0.4\textwidth}
    \centering
    \includegraphics[height=20ex]{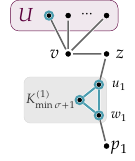}
    \caption{The case $\rhoMin = 1, \sigMin \geq 2$ illustrated for $\sigMin = 2, \rhoMin = 1$ and $r = 1$.}
  \end{subfigure}~

  \begin{subfigure}[t]{0.4\textwidth}
    \centering
    \includegraphics[height=20ex]{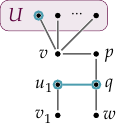}
    \caption{The case $\rhoMin = \sigMin = 1$ illustrated for
      $r = 1, s = 1$.}
  \end{subfigure}\quad
  \begin{subfigure}[t]{0.4\textwidth}
    \centering
    \includegraphics[height=20ex]{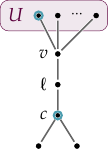}
    \caption{The case $\rhoMin = 1, \sigMin = 0$.}
  \end{subfigure}~

  \caption{The gadget constructions from \cref{lem:relations:hwOne}.
    In each sketched construction we mark vertices corresponding to a feasible solution within the gadget.
    From the relation scope $U$, we select an arbitrary vertex.
  }
  \label{fig:realizing_hw_1}
\end{figure}

\relationsHWOne

\begin{proof}
  \renewcommand{\sigMin}{\min\sigma}
  \renewcommand{\rhoMin}{\min\rho}
  We consider different cases, depending on $\sigma$ and $\rho$.
  Note that $\sigma$, $\rho$ being difficult implies that $\mname \geq 3$ and $\rhoMin \geq 1$.
  \begin{itemize}
    \item
          If $\rhoMin \ge 2$,
          then we use \cref{claim:hw_geq_rmin_2}.
    \item
          If $\rhoMin = 1$ and $\sigMin \geq 2$,
          then we use \cref{claim:hw_geq_rmin_1_smin_2}.
    \item
          If $\rhoMin = 1$, $\sigMin = 1$,
          then we use \cref{claim:hw_geq_rmin_1_smin_1_m_3}.
    \item
          If $\rhoMin = 1$, $\sigMin = 0$,
          then we use \cref{claim:hw_geq_rmin_1_smin_0_mRho_3}.
  \end{itemize}
  Consult \cref{fig:realizing_hw_1} for a visualization of the gadgets used in the
  different cases of the proof.

  In the following,
  we describe the gadgets in the strongest-possible way,
  that is, we describe the minimal requirements for the gadgets to work.

  We start with the first case where $\rhoMin \ge 2$.
  Clearly, $\mname \geq 3$ implies that $\sigMin +1 \notin \sigma$.

  \begin{claim}
    \label{claim:hw_geq_rmin_2}
    Let $\sigma$ and $\rho$ be arbitrary non-empty sets such that
    $\rhoMin \geq 2$ and $\sigMin+1 \notin \sigma$.
    Then, \HWone{} is realizable.
  \end{claim}
  \begin{claimproof}
    For all $i \in \numb{\rhoMin-1}$,
    we create a clique $K_{\sigMin+1}^{(i)}$ on $\sigMin+1$ vertices.
    We create a vertex $v$ that is made adjacent to all vertices in the scope of the relation.
    Finally, we select a vertex $u^{(i)}$ from the clique $K_{\sigMin+1}^{(i)}$,
    and make $v$ adjacent to $u^{(i)}$ for all $i \in \numb{\rhoMin-1}$.

    The correctness relies on the fact that all vertices of each clique must be selected in a solution.
    To see this, fix some $i \in \numb{\rhoMin-1}$
    and consider the clique $K_{\sigMin+1}^{(i)}$.

    If $\sigMin+1 = 1$, then the clique consists of a single vertex $u^{(i)}$
    that must be selected since $u^{(i)}$ only has a single neighbor,
    and if the vertex would be unselected, it would need at least two neighbors
    because of $\rhoMin \geq 2$.

    If $\sigMin+1 = 2$,
    then the clique consists of two vertices connected by an edge.
    Let $w$ be the vertex of the clique that is not $u^{(i)}$.
    Since $w$ has only a single neighbor, vertex $w$ must be selected.
    But, since $\sigMin = 1$, this means that also $u^{(i)}$,
    the only neighbor of $w$, must be selected.

    If $\sigMin+1 \geq 3$,
    assume that some vertex $v^{(i)}$ of the clique that is not $u^{(i)}$
    is not selected in a solution.
    Then, $v^{(i)}$ must have two selected neighbors in the clique.
    At least one of them, denote it by $w^{(i)}$,
    must be different from $u^{(i)}$,
    Hence, $w^{(i)}$ is selected and requires $\sigMin$ selected neighbors,
    which implies that all vertices of the clique must be selected,
    contradicting that $v^{(i)}$ is not selected.

    By the previous argument vertex $v$ is adjacent to a vertex $u^{(i)}$
    such that $u^{(i)}$ has exactly $\sigMin$ selected neighbors in the clique.
    Since $\sigMin+1 \not \in \sigma$, vertex $v$ cannot be selected.
    Furthermore, $v$ is adjacent to $\rhoMin -1$ vertices that are selected,
    and we know that $\rhoMin -1 \not \in \rho$.
    Thus, at least one vertex in the scope of the relation must be selected,
    so that $v$ can have enough neighbors.
    Moreover, exactly $r$ vertices from the scope must be selected
    where $r +\rhoMin-1 \in \rho$.
  \end{claimproof}

  For all the remaining cases we assume that $\rhoMin = 1$.
  We proceed with the case where $\sigMin \ge 2$.

  \begin{claim}
    \label{claim:hw_geq_rmin_1_smin_2}
    Let $\sigma$ and $\rho$ be arbitrary non-empty sets such that
    $\rhoMin = 1$,
    there exists an $r \in \rho$ with $r+1 \notin \rho$,
    $\sigMin \geq 2$, and $\sigMin+1 \notin \sigma$.
    Then, \HWin{\rho} is realizable.
  \end{claim}
  Observe that in our case choosing $r = \rhoMin$ is possible.
  \begin{claimproof}
    The gadget contains the $r+2$ vertices $v, z, p_1,\dots, p_r$
    and $r$ cliques $K^{(i)}_{\sigMin+1}$ on $\sigMin+1$ vertices
    each where $i \in \numb{r}$,
    and $u_i$ and $w_i$ denote two distinct vertices
    of the clique $K^{(i)}_{\sigMin+1}$.
    We make $v$ adjacent to all vertices of the scope and to $z$.
    Vertex $z$ is adjacent to all $u_i$
    and vertex $w_i$ is adjacent to $p_i$ for all $i$.

    In any solution, all vertices of each clique $K^{(i)}_{\sigMin+1}$
    must be selected due to vertex $p_i$:
    Vertex $p_i$ has only a single neighbor
    and hence, cannot be selected as $\sigMin \ge 2$.
    Moreover, since $\rhoMin = 1$,
    the unique neighbor $w_i$ of $p_i$ must also be selected.
    The selected vertex $w_i$ then needs $\sigMin$ selected neighbors,
    and so the whole clique must be selected.
    Then, vertex $u_i$ has $\sigMin$ selected neighbors,
    which implies that its only other neighbor $z$ cannot be selected.
    Since the unselected vertex $z$ already has $r$ selected neighbors $u_i$,
    vertex $v$ cannot be selected as $r+1 \notin \rho$.
    However, as $v$ is unselected and has no selected neighbors,
    it requires to have $t$ selected neighbors for some $t \in \rho$
    which have to stem from the scope of the relation.
  \end{claimproof}

  Next we change the requirement for $\sigMin$ by assuming $\sigMin=1$.

  \begin{claim}
    \label{claim:hw_geq_rmin_1_smin_1_m_3}
    Let $\sigma$ and $\rho$ be arbitrary non-empty sets such that
    $\rhoMin = 1$,
    there exists an $r \in \rho$ with $r+1 \notin \rho$,
    $\sigMin = 1$,
    $\sigMin+1 \notin \sigma$,
    and there is an $s \in \sigma$ with $s+1, s+2 \notin \sigma$.
    Then, \HWin{\rho} is realizable.
  \end{claim}
  Observe that this covers our case by setting $r = 1$ and $s = 1$
  as $\mname\ge 3$.
  \begin{claimproof}
    Assume we can construct a gadget $F$ where a distinguished vertex $p$ is forced to be not selected,
    and the only possible solution also provides exactly one selected neighbor for this vertex $p$.
    Then, we can realize the relation as follows.
    Create $r$ copies of $F$ which we denote by $F_1, \dots, F_r$ where we identify all vertices $p_1, \dots, p_r$ with a new vertex $p$
    We additionally add a vertex $v$ and make $v$ adjacent to all vertices in the scope of the relation and $p$.

    By the properties of the gadget $F$,
    the vertex $p$ is not selected and has one neighbor in each copy $F_i$.
    Hence, vertex $v$ cannot be selected as $r+1 \notin \rho$.
    Therefore, the number of vertices that are selected from the scope must be $t$ for some $t \in \rho$.

    It remains to construct the gadget we assumed to exist above.
    For this we introduce $2s+3$ vertices
    $p, q, w$ and $u_1, v_1, \dots, u_s, v_s$.
    We connect the three vertices~$p, q, w$ in that order to a path.
    Moreover, for all~$i \in \numb{s}$,
    we also connect the three vertices~$q, u_i, v_i$ in that order to a path.

    Since $\rhoMin=\sigMin=1$, all vertices $u_i$ must be selected because of the requirement of each $v_i$.
    Moreover, because the only neighbor of $w$ is $q$,
    vertex $q$ must be selected as well.
    Hence, as $\sigMin+1 \notin \sigma$ and $q$ is selected,
    each vertex $u_i$ forces $v_i$ to not be selected.
    With this selection vertex $q$ has $s$ selected
    neighbors.
    Since $s+1, s+2 \notin \sigma$,
    it is not possible to select any of the vertices $w$ and $p$
    which have one selected neighbor each.

    This concludes the construction of the auxiliary gadget $F$
    with $p$ as the distinguished vertex.
  \end{claimproof}

  For the remaining case we now assume that $\sigMin = 0$.

  \begin{claim}
    \label{claim:hw_geq_rmin_1_smin_0_mRho_3}
    Let $\sigma$ and $\rho$ be arbitrary non-empty sets such that
    $\rhoMin = 1$, and
    $2, 3 \not \in \rho$,
    $\sigMin = 0$, and
    $\sigMin+1 \notin \sigma$.
    Then, \HWin{\rho} is realizable.
  \end{claim}
  Observe that our case is covered since $\mname \geq 3$, $\rhoMin = 1$ and $\sigMin = 0$ implies all required conditions.
  \begin{claimproof}
    We start by creating a star graph $S_3$ with $3$ leaves.
    Let $c$ be the center of the star.
    Add a vertex $v$ to the graph, and make $v$ adjacent to one leaf of $S_{3}$, and to all vertices in the scope of the relation.

    We first argue that in any solution, vertex $c$ must be selected,
    and furthermore, all of its neighbors cannot be selected.
    Towards a contradiction, assume that $c$ is not selected.
    Then, there are two leaves of $c$ that must be selected
    as $\min\rho = 1$ and $0 \in \sigma$.
    However, for $c$, two out of three of its neighbors are now selected, and $2 \not \in \rho, 3 \not \in \rho$.
    Thus, $c$ must be selected in any solution.
    The leaves of $c$ that have no other neighbors cannot be selected now, as they have selected neighbor $c$ and $1 \not \in \sigma$.

    Consider the leaf $\ell$ of $S_3$ that is a neighbor of $v$.
    This leaf has two neighbors, one of which is selected.
    As $c$ is selected and $1 \notin \sigma$,
    vertex $\ell$ cannot be selected.
    Furthermore, $v$ cannot be selected either,
    because selecting $v$ would give $\ell$ two selected neighbors
    and $2 \not \in \rho$.
    Thus, $v$ is not selected and requires at least one more selected vertex from the scope of the relation.

    Finally, if $r \in \rho$ vertices of the relation scope are selected,
    we claim that it suffices to select only vertex~$c$
    such that every vertex of the gadget
    has a valid number of selected neighbors.
    For the leaves of~$c$ this is immediately true
    because they have a single selected neighbor
    and are unselected.
    For vertex $c$ this also holds as $c$ is selected and has no selected neighbor.
    Furthermore, vertex $v$ has no additional selected neighbors
    besides the~$r$ selected neighbors in the scope of the relation
    and we know that~$r \in \rho$.
  \end{claimproof}

  This finishes the proof of \cref{lem:relations:hwOne} by combining
  \cref{claim:hw_geq_rmin_2,%
    claim:hw_geq_rmin_1_smin_2,%
    claim:hw_geq_rmin_1_smin_1_m_3,%
    claim:hw_geq_rmin_1_smin_0_mRho_3,%
  }.
\end{proof}

\section{Turn the Lights Off!}
\label{sec:lightsOut}
\graphicspath{{figures/lights_out/}}

Recall that \srDomSet{}
captures multiple variants of \LightsOut{}.
Concretely, we denote by \ReflAllOff
the version of \srDomSet
where $\sigma = \EVEN$ and $\rho = \ODD$,
and by \AllOff the variant where $\sigma = \rho = \ODD$.
As already mentioned in the introduction, the decision versions of \ReflAllOff{}  and \AllOff{}
can be solved in polynomial-time by Gaussian elimination
on the corresponding set of linear equations~\cite{Sutner89,HalldorssonKT00-mod-2,GoldwasserKT97,AndersonF98,dodisUniversalConfigurationsLightflipping2001}.
However, as the minimization version of these two problems
is \NP-complete~\cite{sutner1988additive,%
    caroOddResidueDomination2001,HalldorssonKT00-mod-2},
it is open whether our algorithm from \cref{thm:twUpperBound}
is optimal or not.
In this section we answer this question in the positive
by formally proving \cref{thm:lightsOut:lowerPW}.

\thmLightsOutLowerPW*

We use the \NP-hardness result by Sutner~\cite[Theorem~3.2]{sutner1988additive}, as presented in \cite[Theorem 19]{FleischerY13}, as a basis
and strengthen our results by directly proving the bounds
for the larger parameter pathwidth.
We consider \ReflAllOff in \cref{sec:alloff:reflexive} and \AllOff in \cref{sec:alloff:nonreflexive}.
The proof of \cref{thm:lightsOut:lowerPW} then directly follows from
the combination of the main results of the aforementioned sections.

\subsection{Lower Bound for \ReflAllOff}
\label{sec:alloff:reflexive}

In the following we prove the lower bound for \ReflAllOff
by a reduction from \kSAT
to an equivalent instance of \ReflAllOff with small pathwidth.

\begin{theorem}
    \label{thm:allOff:reflexive:lowerTW}
    Unless SETH fails,
    for all $\epsilon > 0$,
    there is no algorithm for \ReflAllOff
    that can decide
    in time $(2 - \epsilon)^{\pw} \cdot |G|^{\Oh(1)}$
    whether there exists a solution
    of arbitrary size (the size is given as input)
    for a graph~$G$
    that is given with a path decomposition of width~$\pw$.
\end{theorem}
\begin{proof}
    We prove the lower bound by a reduction from CNF-SAT.
    Fix some $\epsilon > 0$ for this
    and let $k$ be the smallest integer such that
    \kSAT does not have a $(2-\epsilon)^n \cdot (n+m)^{\Oh(1)}$ algorithm
    where $n$ is the number of variables and $m$ the number of clauses.

    Consider an arbitrary \kSAT formula $\varphi$
    with $n$ variables $x_1,\dots,x_n$ and $m$ clauses $C^1,\dots,C^m$ as input.%
    \footnote{%
        We assume that every clause contains exactly $k$ literals.
        This restriction is not of technical nature as the constructions
        works for the general case but rather to keep notation simple and clean.%
    }
    In the following we construct a graph $G_\varphi$ as an instance of \ReflAllOff{}.
    The graph is built based on variable gadgets, clause gadgets,
    and a single negation gadget.

    We first construct the gadgets and then describe how they are connected.
    For every variable $x_i$ where $i \in \numb{n}$,
    the \emph{variable gadget} $V_i$ consists of the two vertices
    $v_i$ and $\bar v_i$ that are connected by an edge.

    For every clause $C^j = \lambda^j_1 \lor \dots \lor \lambda^j_k$
    where $j \in \numb{m}$,
    the \emph{clause gadget} $D^j$
    contains the following vertices and edges.
    There are $k$ \emph{literal vertices} $t^j_1, \dots, t^j_k$
    where each vertex corresponds to one literal of the clause.
    Moreover, the gadget $D^j$ contains so-called \emph{subset vertices}
    $s^j_L$ for all $L \subsetneq \numb{k}$,
    that is, for every proper subset of the literals of the clause,
    there exists a vertex labeled with this subset (and the gadget index).
    For each subset $L \subsetneq \numb{k}$,
    the subset vertex $s^j_L$ is connected to the literal vertex $t^j_\ell$
    if and only if $\ell \in L$.
    Moreover, all subset vertices together form a clique on $2^k - 1$ vertices.

    The negation gadget consists of three vertices $q_0$, $q_1$, and $q_2$
    that are connected to a path on three vertices with $q_1$ in the middle.

    As a last step it remains to connect the vertices of the different gadgets.
    Intuitively, each literal vertex of the clause gadget
    is connected to the corresponding variable vertex of the variable gadgets.
    Consider literal $\lambda^j_\ell$,
    that is, the $\ell$th literal in the $j$th clause.
    If this literal is positive,
    i.e., if $\lambda^j_\ell = x_i$ for some variable $x_i$,
    then the vertex $t^j_\ell$ is adjacent to vertex $v_i$.
    If the literal is negative,
    i.e., if $\lambda^j_\ell = \lnot x_i$ for some variable $x_i$,
    then the vertex $t^j_\ell$ is also adjacent to vertex $v_i$
    but additionally also to vertex $q_1$.

    This concludes the description of $G_\varphi$.
    See \cref{fig:alloff:reflexive:reduction} for an illustration of the clause gadget and the negation gadget.
    To prove the correctness of this reduction we set $m+n+1$
    as the upper bound for the number of selected vertices.

    \graphicspath{{figures/tikz}}
    \begin{figure}[t]
        \centering
        \includegraphics[scale=2]{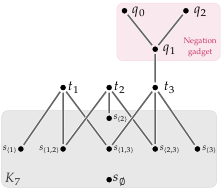}
        \caption{A depiction of the clause gadget for the clause $x_1 \lor x_2 \lor \neg x_3$ as well as the negation gadget. Some indices are omitted for simplicity.}
        \label{fig:alloff:reflexive:reduction}
    \end{figure}

    \begin{claim}
        \label{thm:all_off_reduction_correctness_1}
        If $\varphi$ is a yes-instance of \kSAT,
        then $G_\varphi$ has a solution for \ReflAllOff of size at most $n+m+1$.
    \end{claim}
    \begin{claimproof}
        Let $\pi$ be a satisfying assignment for the \kSAT formula $\varphi$.
        We select the following vertices:
        \begin{itemize}
            \item
                  In the variable gadget of variable $x_i$,
                  we select $v_i$ if $\pi(x_i)= 1$ and $\bar v_i$ otherwise.
            \item
                  In the clause gadget of clause $C^j$,
                  let $L \subsetneq \numb{k}$ be the set of literal indices
                  of this clause that are \emph{not} satisfied.
                  (This is well-defined as $\pi$ is a satisfying
                  which means that not all literals are unsatisfied.)
                  We select vertex $s^j_L$.
            \item
                  We select vertex $q_1$ from the negation gadget.
        \end{itemize}
        Let $S$ denote the set of all selected vertices.

        Since we select exactly one vertex from each gadget,
        the size of $S$ is precisely $n+m+1$.
        It remains to prove that $S$ is indeed a solution,
        that is, $S$ is a $(\sigma, \rho)$-set where $\sigma=\EVEN$ and $\rho=\ODD$.

        First, consider the vertices of the negation gadget.
        Since only vertex $q_1$ is selected,
        the vertices $q_0$ and $q_2$ have exactly one selected neighbors.
        As none of the literal vertices are selected,
        vertex $q_1$ is adjacent to zero selected vertices.

        Consider any vertex of the variable gadgets
        and observe that none of them have a selected neighbor outside the gadget.
        By our choice of $S$, exactly one of the two vertices is selected
        and has no selected neighbors,
        while the other vertex is not selected
        and has exactly one selected neighbor, the one in the variable gadget.

        It remains to check the vertices of the clause gadget.
        For the subset vertices we notice that in every clause gadget
        exactly one of them is selected.
        Recall that all subset vertices of one clause gadget are connected to each other,
        and they are not connected to any other vertices outside the gadget
        but only to the unselected literal vertices of the gadget.
        This implies that exactly one of the subset vertices is selected
        and has no selected neighbors,
        whereas the other subset vertices are not selected
        and have exactly one selected neighbor.

        Next we check the literal vertices.
        Consider a positive literal, say $\lambda^j_\ell = x_i$.
        By our selection $t^j_\ell$ is unselected and has
        vertex $v_i$ as selected neighbor if $\pi(x_i)=1$.
        In this case no subset vertex $s^j_L$ with $\ell \in L$ is selected
        by definition of $S$.
        Hence, the vertex has exactly one selected neighbor.
        If $\pi(x_i)=0$, then the literal does not satisfy the clause
        and hence, by the definition of $S$,
        a vertex $s^j_L$ is selected where $\ell \in L$.

        As a last step we check the literal vertices
        corresponding to negated variables, say $t^j_\ell = \lnot x_i$.
        Once more, this vertex is not selected
        but always adjacent to the selected vertex $q_1$ from the negation gadget.
        If $\pi(x_i)=1$, then also the neighboring vertex $v_i$ is selected.
        However, this is not a problem
        since in this case the literal does not satisfy the clause
        and hence, a subset vertex $s^j_L$ with $\ell \in L$ must be selected.
        Thus, the literal vertex is adjacent to three selected vertices.
        If $\pi(x_i) = 0$, then the literal satisfies the clause
        and no subset vertex $s^j_L$ with $\ell \in L$ is selected;
        the literal vertex has exactly one selected neighbor.

        We conclude that every selected vertex of the graph
        has no selected neighbors,
        whereas every unselected vertex has either exactly one
        or exactly three selected neighbors.
    \end{claimproof}

    As a next step we show the reverse direction of the correctness.

    \begin{claim}
        \label{thm:all_off_reduction_correctness_2}
        If $G_\varphi$ has a solution for \ReflAllOff of size at most $n + m + 1$,
        then $\varphi$ is a yes-instance of \kSAT.
    \end{claim}
    \begin{claimproof}
        Consider a solution $S$ of the \ReflAllOff instance
        of size at most $n+m+1$.
        Recall that $S$ is a $(\sigma,\rho)$-set for $G_\varphi$
        where $\sigma = \EVEN$ and $\rho = \ODD$.

        In the negation gadget at least one vertex must be selected
        as $q_0$ and $q_2$ must be selected themselves
        or require one selected neighbor which can only be $q_1$.
        Similarly, in each variable gadget $V_i$,
        the set $S$ contains at least one vertex;
        vertex $\bar v_i$ is either selected itself
        or it requires a selected neighbor (which must be $v_i$).

        In each clause gadget $D^j$,
        we must also select at least one vertex which must be a subset vertex.
        Indeed, the subset vertex $s^j_\emptyset$ must either be selected
        or have a selected neighbor.
        As all neighbors of this subset vertex are also subset vertices,
        the claim follows.

        Hence, we see that any solution $S$ of size at most $n+m+1$
        must select \emph{exactly one} vertex of each variable gadget,
        as well as exactly one subset vertex of each clause gadget,
        and vertex $q_1$ of the negation gadget.

        Based on these observations we define an assignment $\pi$
        for the formula $\varphi$ by setting $\pi(x_i) \deff 1$
        if and only if $v_i \in S$ and $\pi(x_i) \deff 0$ otherwise.

        In the following we prove that $\pi$ satisfies $\varphi$.
        For this consider an arbitrary clause $C^j$.
        By the above discussion,
        we know that there is some selected subset vertex $s^j_L$
        of the clause gadget $D^j$.
        Since $L \subsetneq \numb{k}$,
        there is, by the construction of $G_\varphi$,
        a literal vertex $t^j_\ell$ where $\ell \in \numb{k} \setminus L$.
        Since this vertex $t^j_\ell$ is not selected,
        it must have an odd number of neighbors in $S$.
        We first consider the case that the corresponding literal $\lambda^j_\ell$
        is positive, that is, $\lambda^j_\ell = x_i$ for some variable $x_i$.
        In this case the only selected neighbor of~$t^j_\ell$ is~$v_i$.
        From $v_i \in S$ the definition of $\pi$ gives $\pi(x_i) = 1$
        which implies that the clause~$C^j$ is satisfied
        by literal~$\lambda^j_\ell$.
        Now consider the case when $\lambda^j_\ell = \lnot x_i$
        for some variable $x_i$.
        In this case vertex $t^j_\ell$ is adjacent to the selected vertex $q_1$
        of the negation gadget, by the construction of $G_\varphi$
        and the above observations.
        As the literal vertex is still unselected
        and only adjacent to unselected subset vertices
        and one additional variable vertex,
        the variable vertex $v_i$ cannot be selected.
        Hence, the definition of the assignment $\pi$ gives $\pi(x_i) = 0$
        which directly implies that the clause $C^j$ is satisfied
        by the literal $\lambda^j_\ell$.
    \end{claimproof}

    Before we combine all parts of the proof to obtain the lower bound,
    we first provide a bound on the pathwidth of the constructed graph.

    \begin{claim}
        \label{thm:all_off_pathwidth}
        $G_\varphi$ has pathwidth at most $2^k +k + n$.
        Moreover, a path decomposition of $G_\varphi$ of width at most  $2^k + k + n$ can be computed in polynomial-time.
    \end{claim}
    \begin{claimproof}
        If we delete vertex~$q_1$ and, for all $i \in \numb{n}$,
        the variable vertex $v_i$,
        then the graph decomposes into small components
        of size at most $2^k-1+k$.
        This allows us to get a path decomposition of small width.

        Concretely, we create a node for every variable gadget, clause gadget,
        and the negation gadget.
        We add the vertices of the corresponding gadgets to the bag of the node.
        We connect these nodes in an arbitrary way to form a path.
        As a last step we extend all bags by adding the vertices in
        $\{ v_i \mid i \in \numb{n}\} \cup\{q_1\}$.
        It is easy to see that this is a valid path decomposition.

        The bags corresponding to the variable gadgets have size $1 + (n + 1)$
        and the bags of the clause gadgets have size $2^k - 1 + k + (n + 1)$.
        Finally, the bag of the negation gadget has size $2 + (n + 1)$.
        Hence, the pathwidth of the graph is at most
        $\max( 2^k + k + n - 1, n + 2 ) \leq 2^k + k + n$.
    \end{claimproof}

    Recall that we fixed some $\epsilon > 0$
    and set $k$ to the smallest integer such that
    \kSAT does not have an algorithm with running time
    $(2-\epsilon)^n \cdot n^{\Oh(1)}$
    where $n$ is the number of variables.
    For a given \kSAT formula $\varphi$ with $n$ variables and $m$ clauses
    we constructed an \ReflAllOff instance $G_\varphi$
    that has a solution of size $n + m +1$ if and only if $\varphi$ is satisfiable,
    and a path decomposition of $G_\varphi$ of width at most $2 +k + k + n$, in polynomial-time (recall that $k$ is a constant depending only on $\epsilon$).

    Towards a contradiction,
    assume that the minimization variant of \ReflAllOff can be solved in time
    $(2-\epsilon)^\pw \cdot N^{\Oh(1)}$ on instances of size $N$.
    Applying this algorithm to the constructed instance $G_\varphi$ and asking for a solution of size at most $n+m+1$,
    yields, by using \cref{thm:all_off_reduction_correctness_1,%
        thm:all_off_reduction_correctness_2,thm:all_off_pathwidth},
    an algorithm for \kSAT with running time
    \[
        (2 - \epsilon)^\pw \cdot (n+m)^{\Oh(1)}
        \leq (2 - \epsilon)^{2^k+k+n} \cdot (n+m)^{\Oh(1)}
        = (2- \epsilon)^n \cdot (n+m)^{\Oh(1)}
    \]
    as the constant $k$ only depends on the fixed value $\epsilon$.
    This directly contradicts SETH and finishes the proof.
    \qedhere
\end{proof}

\subsection{Lower Bound for \AllOff}
\label{sec:alloff:nonreflexive}

In the previous section we have seen the lower bound
for \ReflAllOff.
In this section we focus on the non-reflexive version \AllOff instead.
Recall, that when representing this problem as \srDomSet
we have $\sigma=\rho=\ODD$,
that is, every vertex needs and odd number of selected neighbors
independent of whether the vertices are themselves selected or not.

We again use the ideas by Sutner~\cite{sutner1988additive}, as presented in \cite{FleischerY13},
which allows a reduction in a similar spirit as the previous one.
The main difference is that the subset vertices now do not form a clique.
As there are naturally other minor modifications we present the full proof.

\begin{theorem}
    \label{thm:allOff:nonreflexive:lowerTW}
    Unless SETH fails,
    for all $\epsilon > 0$,
    there is no algorithm for \AllOff that
    can decide in time $(2-\epsilon)^\pw \cdot |G|^{\Oh(1)}$
    whether there exists a solution
    of arbitrary size (the size is given as input)
    that is given with a path decomposition of width $\pw$.
\end{theorem}
\begin{proof}
    Similar to \cref{thm:allOff:reflexive:lowerTW},
    we prove the bound by a reduction from CNF-SAT.
    For this fix some $\epsilon > 0$
    and let $k$ be the smallest integer such that
    \kSAT has no $(2-\epsilon)^n \cdot (n+m)^{\Oh(1)}$ algorithm
    under SETH where $n$ is the number of variables
    and $m$ the number of clauses.

    Let $\varphi$ be an arbitrary \kSAT instance consisting of
    $n$ variables $x_1,\dots,x_n$ and $m$ clauses $C^1,\dots,C^m$ as input.%
    \footnote{%
        As for the proof of \cref{thm:allOff:reflexive:lowerTW}
        we assume purely for the ease of the presentation
        that each clause contains exactly $k$ literals.
    }
    We construct a graph $G_\varphi$
    that consists of variable and clause gadgets.
    See \cref{fig:allOff:nonreflexive:twLower}
    for an illustration of the construction.

    For every variable $x_i$ where $i \in \numb{n}$,
    the graph $G_\varphi$ contains a variable gadget $V_i$.
    Gadget $V_i$ consists of a path on three vertices $v_i,w_i,\bar v_i$ (in this order).

    For every clause $C^j = \lambda^j_1 \lor \dots \lor \lambda^j_k$
    where $j \in \numb{m}$, the graph $G_\varphi$ contains a clause gadget $D^j$.
    This gadget $D^j$ contains $k$ literal vertices $t^j_1,\dots,t^j_k$,
    that is, one distinguished vertex for every literal of the clause.
    Additionally, $D^j$ contains, for every subset $L \subsetneq \numb{k}$,
    a so called \emph{subset}-vertex $s^j_L$,
    that is, for every proper subset of the literals of the clause,
    there exists a vertex labeled with the subset (and the gadget index).
    These subset vertices are used to indicate which literals of the clause
    are \emph{not} satisfied by the encoded assignment.
    As a last vertex there is an additional vertex $h^j$ (for happy) in $D^j$.
    There are two different groups of edges in the gadget;
    first, each subset vertex $s^j_L$ is connected to each literal vertex $t^j_\ell$
    if and only if $\ell \in L$.
    Second, the vertex $h^j$ is connected to all subset vertices
    of the gadget $D^j$.

    As a last step of the construction we describe the edges
    encoding the appearance of variables in the clauses.
    Intuitively each literal vertex of the clause gadget is connected
    to the corresponding variable vertex of the variable gadgets.
    Formally, if the $\ell$th literal of the $j$th clause is positive,
    i.e., if $\lambda^j_\ell = x_i$ for some variable $x_i$,
    then vertex $t^j_\ell$ is adjacent to vertex $v_i$.
    If the $\ell$th literal is negative,
    i.e., if $\lambda^j_\ell = \neg x_i$ for some variable $x_i$,
    then vertex $t^j_\ell$ is again adjacent to vertex $v_i$
    but also to vertex $h^j$.

    This concludes the description of $G_\varphi$.
    We prove in the following that the \AllOff instance $G_\varphi$
    has a solution of size at most $2m + 2n$ if and only if $\varphi$ is satisfiable.

    \begin{figure}[t]
        \graphicspath{{figures/tikz}}
        \begin{subfigure}[t]{.48\textwidth}
            \centering
            \includegraphics[height=10ex]{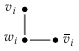}
            \caption{A depiction of the variable gadget.}
            \label{fig:variable_gadget}
        \end{subfigure}%
        \hfill%
        \begin{subfigure}[t]{.48\textwidth}
            \centering
            \includegraphics[height=25ex]{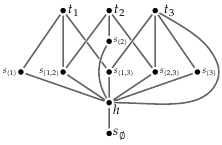}
            \caption{A depiction of the clause gadget for the clause $x_1 \lor x_2 \lor \neg x_3$.}
            \label{fig:clause_gadget}
        \end{subfigure}
        \caption{
            A depiction of a literal gadget and a clause gadget
            from the proof of the lower bound for \AllOff.
            Some indices are omitted for simplicity.
        }
        \label{fig:allOff:nonreflexive:twLower}
    \end{figure}

    We first show that if the \kSAT formula $\varphi$ is satisfiable,
    then the constructed \AllOff instance is a yes-instance.

    \begin{claim}
        \label{thm:reflexive_all_off_reduction_correctness_1}
        If $\varphi$ is a yes-instance of \kSAT,
        then $G_\varphi$ has a solution for \AllOff of size at most $2m + 2n$.
    \end{claim}
    \begin{claimproof}
        Consider a satisfying assignment $\pi$ for $\varphi$.
        We select the following vertices:
        \begin{itemize}
            \item
                  For all $i \in \numb{n}$,
                  if $\pi(x_i) = 1$, then we select the two vertices $v_i$ and $w_i$.
                  If $\pi(x_i) = 0$,
                  then we select $\bar v_i$ and $w_i$.

            \item
                  For all $j \in \numb{m}$,
                  we select in the clause gadget the vertex $h^j$.
                  Moreover, let $L \subsetneq \numb{k}$
                  be the set of all positions of the literals of the clause
                  that are \emph{not} satisfied by $\pi$.
                  Since the clause is satisfied by the assignment,
                  the set $L$ cannot contain all literals of the clause.
                  Thus, there is a subset vertex $s^j_L$ corresponding to this set $L$.
                  We select this vertex $s^j_L$.
        \end{itemize}
        Let $S$ denote the set of all selected vertices.
        Clearly, this set contains exactly $2m+2n$ vertices
        as we select exactly two vertices from every gadget.
        It remains to show that $S$ is indeed a solution,
        that is, every vertex of $G_\varphi$ has an odd number of neighbors in $S$.

        Consider any vertex of the variable gadgets
        and observe that none of them have a selected neighbor outside the gadget.
        With this it is easy to see that each vertex of a variable gadget
        has exactly one selected neighbor in $S$.

        Each subset vertex of a clause gadget $D^j$
        has exactly one selected neighbor, namely vertex $h^j$.
        Furthermore, vertex $h^j$ has exactly one selected neighbor,
        namely the selected subset vertex of the clause gadget.
        Now consider an arbitrary literal vertex $t^j_\ell$.
        First assume that the corresponding literal is positive,
        that is, $\lambda^j_\ell = x_i$ for some variable $x_i$.
        If this literal is satisfied,
        the vertex $v_i$ in the variable gadget is selected,
        otherwise if the literal is not satisfied,
        then vertex~$t^j_\ell$ is a neighbor of the selected subset vertex.

        Now assume that the corresponding literal is negative,
        that is, $\lambda^j_\ell = \lnot x_i$ for some variable~$x_i$.
        If the variable $x_i$ is not satisfied by $\pi$,
        then the only selected neighbor of $t^j_\ell$ is $h^j$.
        If the variable $x_i$ is satisfied by $\pi$,
        then the literal $\lambda^j_\ell$ is not satisfied
        which implies that the literal vertex is adjacent
        to the selected subset vertex of this clause gadget.
        Moreover, vertex $t^j_\ell$ is adjacent to $v_i$
        and thus, to three selected vertices which is a valid number.

        We conclude the proof by observing that every vertex
        has either one or three selected neighbors and thus,
        we constructed a valid solution.
    \end{claimproof}

    As a next step we prove the reverse direction of the correctness.

    \begin{claim}
        \label{thm:reflexive_all_off_reduction_correctness_2}
        If $G_\varphi$ has a solution for \AllOff of size $2m + 2n$,
        then $\varphi$ is a yes-instance of \kSAT.
    \end{claim}
    \begin{claimproof}
        Consider a solution $S$ to the \AllOff instance
        of size at most $2m + 2n$.
        Recall that every vertex of $G_\varphi$ is adjacent
        to an odd number of vertices in $S$.

        We first start with some observations about the solution $S$.
        In each variable gadget $V_i$,
        vertex $\bar v_i$ must have a neighbor that is in $S$, and thus, $w_i$ is also in $S$.
        Then, vertex $w_i$ requires a selected neighbor too, so either $v_i$ or $\bar v_i$ must be selected as well.
        Hence, at least two vertices must be selected from each variable gadget.

        In each clause gadget $D^j$, the vertex $s^j_\emptyset$
        must have a selected neighbor which forces its only neighbor,
        which is vertex $h^j$, to be selected.
        The vertex $h^j$ also requires a selected neighbor.
        If this neighbor was a literal vertex, say vertex $t^j_\ell$, then, the subset vertex $s^j_{\{\ell\}}$ would have exactly two selected neighbors, namely $h^j$ and $t^j_\ell$.
        As $2$ is neither in $\sigma$ nor in $\rho$, we see that the selected neighbor is not a literal vertex, and must thus be a subset vertex.
        From the given bound on the solution size,
        we conclude that in every variable gadget and every clause gadget
        exactly two vertices are selected.
        Moreover, from each clause gadget $D^j$ exactly one subset vertex, and vertex $h^j$ are selected.

        We define the assignment $\pi$ for the variables of $\varphi$
        such that $\pi(x_i) = 1$ if and only if $v_i \in S$
        and $\pi(x_i)=0$ otherwise.

        It remains to show that $\pi$ satisfies $\varphi$.
        To prove this, consider an arbitrary clause $C^j$.
        Let $s^j_L$ be the selected subset vertex from the clause gadget $D^j$.
        As $L \neq \numb{k}$ by the construction of $G_\varphi$,
        there is some literal vertex $t^j_\ell$ such that $\ell \notin L$.
        If the corresponding literal is positive,
        i.e., if $\lambda^j_\ell = x_i$,
        then $t^j_\ell$ is not adjacent to $h^j$.
        As $t^j_\ell$ must have one selected neighbor in $S$
        and since the vertex is not adjacent to any selected subset vertex,
        the only remaining neighbor of $t^j_\ell$, i.e.,
        vertex $v_i$, must be selected
        which implies that $\pi$ was defined such that $\pi(x_i) = 1$.

        If the literal is negative,
        i.e., if $\lambda^j_\ell = \lnot x_i$,
        then the vertex $t^j_\ell$ is adjacent to $h^j$
        by the construction of $G_\varphi$.
        Since $h^j$ is selected,
        vertex $t^j_\ell$ cannot have further selected neighbors
        as all other adjacent subset vertices are unselected.
        Hence, the vertex $v_i$ is also unselected
        which implies that, by the definition of $\pi$,
        the variable $x_i$ is not satisfied but the literal $\lambda^j_\ell$
        is satisfied which makes the clause $C^j$ true.
    \end{claimproof}

    As a last step we prove a bound on the pathwidth of the constructed graph.

    \begin{claim}
        \label{thm:reflexive_all_off_pathwidth}
        $G_\varphi$ has pathwidth at most $n+k+1$.
        Moreover, a path decomposition of $G_\varphi$ of width at most $n + k + 1$ can be computed in polynomial-time.
    \end{claim}
    \begin{claimproof}
        Intuitively the idea is as follows.
        If we delete all the variable vertices $v_i$ for all $i \in \numb{n}$
        in the variable gadgets, the graph decomposes into small components.
        We use this to construct a path decomposition in the following
        by providing a node search strategy (see, for example, \cite[Section~7.5]{cyganParameterizedAlgorithms} or \cite{FominT08_annotated_bibliography_graph_searching}).

        We start by placing one searcher on each vertex $v_i$
        for every $i \in \numb{n}$.
        Each of the variable gadget can be cleaned
        by using $2$ additional searchers which we just place on all vertices.

        For the clause gadgets we use a more complex approach to clean all vertices.
        Fix a clause gadget $D^j$ for this.
        We first place $k$ new searchers on the $k$ literal vertices of the gadget
        and one more searcher on the vertex $h^j$.
        The remaining subset vertices can then be cleaned
        by using one additional searcher
        which we put one subset vertex after the other.
        Repeating this procedure for all clause gadgets cleans the entire graph.

        This approach uses at most $n+k+2$ searchers simultaneously, and essentially directly provides a path decomposition.
        Thus, the claim follows immediately.
    \end{claimproof}

    Recall that we fixed some $\epsilon > 0$
    and chose $k$ as the smallest integer such that
    \kSAT has no $(2-\epsilon)^n \cdot (n+m)^{\Oh(1)}$ algorithm
    under SETH where $n$ is the number of variables
    and $m$ the number of clauses.
    For a given \kSAT instance $\varphi$ with $n$ variables and $m$ clauses
    we constructed an \AllOff instance $G_\varphi$ that has a solution of size $2m + 2n$ if and only if $\varphi$ is satisfiable,
    together with a path decomposition of width at most $n + k + 1$,
    in polynomial-time (recall that $k$ is a constant only depending on $\varepsilon$).

    For the sake of a contradiction, now assume that the minimization version of \AllOff
    can be solved in time $(2-\epsilon)^\pw \cdot N^{\Oh(1)}$
    on graphs of size $N$.
    If we apply this algorithm to the constructed instance and ask for a solution of size at most $2m +2n$,
    we also solve, by \cref{thm:reflexive_all_off_reduction_correctness_1,%
        thm:reflexive_all_off_reduction_correctness_2,%
        thm:reflexive_all_off_pathwidth} the \kSAT instance in time
    \[
        (2-\epsilon)^\pw \cdot N^{\Oh(1)}
        \leq (2-\epsilon)^{n+k+1} \cdot (n+m)^{\Oh(1)}
        = (2-\epsilon)^{n} \cdot (n+m)^{\Oh(1)}
    \]
    as $k$ depends only on the fixed value $\epsilon$
    and thus, only contributes a constant factor to the running time.
    This then directly contradicts SETH and finishes the proof.
\end{proof}

\thmLightsOutLowerPW
\begin{proof}
    Follows from
    \cref{thm:allOff:reflexive:lowerTW,thm:allOff:nonreflexive:lowerTW}.
\end{proof}

\bibliography{main}

\end{document}